\documentclass[aps,prd,reprint,nofootinbib,superscriptaddress,floatfix]{revtex4-2}
\usepackage[T1]{fontenc}
\usepackage[utf8]{inputenc}
\usepackage{lmodern}
\usepackage{amsmath,amssymb,amsthm,mathtools,bm}
\usepackage{graphicx}
\graphicspath{{figures/}}
\usepackage{xcolor}
\usepackage{booktabs}
\usepackage{enumitem}
\usepackage{hyperref}
\hypersetup{colorlinks=true,citecolor=blue,linkcolor=blue,urlcolor=blue,hypertexnames=false}
\usepackage{orcidlink}
\numberwithin{equation}{section}

\theoremstyle{plain}
\newtheorem{theorem}{Theorem}
\newtheorem{proposition}{Proposition}
\newtheorem{corollary}{Corollary}

\theoremstyle{definition}
\newtheorem{definition}{Definition}
\newtheorem{hypothesis}{Hypothesis}

\theoremstyle{remark}
\newtheorem{remark}{Remark}

\newcommand{\AH}{\mathcal{A}_H}
\newcommand{\Ainf}{\mathcal{A}_{\infty}}
\newcommand{\Qu}{Q_u}
\newcommand{\Ku}{K_u}
\newcommand{\Aone}{A_1(e)}
\newcommand{\Atwo}{A_2(e)}
\newcommand{\Azero}{A_0(e)}
\newcommand{\C}{\mathbb{C}}
\newcommand{\R}{\mathbb{R}}
\newcommand{\Oa}{\mathbb{O}}
\newcommand{\Cu}{\mathbb{C}_u}
\newcommand{\id}{\mathbf{1}}
\newcommand{\idtwo}{\mathbf{1}_2}
\newcommand{\Tr}{\mathrm{Tr}}
\newcommand{\Jor}{\mathrm{Jor}}
\newcommand{\SBH}{S_{\mathrm{BH}}^{\mathrm{RN}}}
\newcommand{\eexc}{E_{\mathrm{exc}}}
\newcommand{\dd}{\mathrm{d}}

\newcommand{\FTS}{\mathfrak{F}}

\begin{document}

\title{Albertian Channel Memory in Black-Hole Evaporation}

\author{Rafael B. Frigori \orcidlink{0000-0002-4861-7240}}
\email{frigori@utfpr.edu.br}
\affiliation{Universidade Tecnol\'ogica Federal do Paran\'a, Rua Cristo Rei 19, 85902-490 Toledo, PR, Brazil}

\begin{abstract}
The AMPS paradox assumes a globally associative tensor-product stage for the early radiation, the exterior Hawking mode, and the interior partner. We study a retained attractor sector of octonionic magical supergravity whose horizon symbols form the Albert algebra $J_3(\Oa)$. This induces an Albertian algebraic-quantum description: states are positive normalized functionals, events are Jordan idempotents, reversible motions are algebra automorphisms, and ordinary quantum mechanics is recovered on associative readout blocks. Peirce theory then splits the horizon data into a hidden exceptional complement, an interface relay, and a two-helicity exterior detector. Eliminating the relay gives a source-fixed Volterra memory law on a neutral-source fixed-charge Reissner--Nordstr{\"o}m evaporation trajectory. In real time, the leading one-time occupation follows the sourced evaporation clock, while the retained-memory imprint appears as a spectral-overlap connected two-time coherence of windowed helicity/Stokes observables in the emitted history. In Euclidean time, the Peirce--Volterra kernel becomes a transfer kernel with two branchwise superstatistical limits: a regular-opening Tsallis/Lomax onset and a near-extremal shifted-L{\'e}vy residence branch. The lower admissible envelope of the endpoint actions then reconstructs the Page-curve envelope. The result is an ordinary emitted readout with exceptional memory, not a restored AMPS tensor factorization.
\end{abstract}

\maketitle

\section{Introduction}

The Almheiri--Marolf--Polchinski--Sully (AMPS) trilemma sharpens the black-hole information problem. It uses local effective field theory to motivate an ordinary tripartite Hilbert-factor description of the early radiation $R$, exterior Hawking mode $B$, and interior partner $A$, schematically $\mathcal H_R\otimes\mathcal H_B\otimes\mathcal H_A$ \cite{Hawking1975,Page1993,AMPS2013,Harlow2016}. In the algebraic language used here, this is the globally associative tensor-product stage whose applicability to the retained horizon algebra is questioned. Semiclassical developments based on generalized entropy, quantum extremal surfaces, replica wormholes, and Islands reorganize the Page-curve side of this problem \cite{Almheiri2020,AlmheiriRMP2021,Penington2020}. Dressing, non-factorization, and gravitational subregion algebras also sharpen the limits of naive subsystem localization \cite{DonnellyFreidel2016,DonnellyGiddings2017,Raju2022}. The question asked here is more primitive: is the globally associative stage presupposed by AMPS the right observable stage for the relevant horizon degrees of freedom?

The answer developed below keeps the ordinary outgoing detector and changes the retained horizon stage. The Jordan--von Neumann--Wigner classification leaves a single non-associative exceptional finite-dimensional simple formally real Jordan algebra, the Albert algebra $J_3(\Oa)$ \cite{Jordan1934,Albert1934,McCrimmon2004}. Algebraic quantum mechanics (AQM) starts from the ordered real algebra of observables rather than from a Hilbert space chosen in advance. States are positive normalized linear functionals, sharp events are idempotents, and reversible motions are automorphisms generated infinitesimally by derivations \cite{Segal1947,HaagKastler1964,Haag1992}. For $H_n(\mathbb C)$ this is ordinary finite-dimensional quantum mechanics in observable language. For $J_3(\Oa)$ it is the unique exceptional finite Jordan-algebraic quantum mechanics, realized here as Albertian AQM \cite{GunaydinPironRuegg1978,Barnum2014,Niestegge2015}.

Octonionic magical supergravity supplies a black-hole setting in which the exceptional Jordan structure arises from charge, adjoint, and cubic norm data naturally organized by $J_3(\Oa)$ \cite{GunaydinSierraTownsend1983,GunaydinSierraTownsend1984,FerraraGunaydin2006}. The Enriques realization gives controlled string-theoretic support for treating the two-derivative octonionic attractor data as retained Wilsonian data \cite{BianchiFerrara2008}. Section~\ref{sec:Albertian} identifies the retained attractor--Wilsonian symbol sector with the effective ordered observable algebra of the retained horizon-side data. In this sector the independent gauge-invariant symbols are generated by $Q,Q^{\#},N(Q)$, with $Q\in J_3(\Oa)$, and the associated ordered cone reconstructs $J_3(\Oa)$. The finite quantum kinematics is therefore the Albertian AQM kinematics induced by this retained observable algebra. Section~\ref{sec:Peirce} then explains how Peirce theory extracts the associative readout block $\Qu\simeq H_2(\C)$.

A source-fixed evaporation trajectory is provided by the neutral-source fixed-charge Reissner--Nordstr{\"o}m (RN) member of the charge-ray-preserving class. Along this RN trajectory the attractor-adapted frame remains frozen while the excitation energy decays. This separates the structural input from the channel calculation: the AMPS obstruction and the detector block are algebraic, while the RN trajectory supplies the clock, the Wentzel--Kramers--Brillouin (WKB) barrier, the Volterra kernel, and the branch weights \cite{ParikhWilczek2000,Gripenberg1990}. The quantitative benchmark used throughout the paper is this controlled readout problem: a neutral-source, fixed-charge Reissner--Nordstr{\"o}m evaporation trajectory with frozen attractor frame, fixed Peirce pair $(e,u)$, selected transverse spin-$1$ two-helicity exterior readout, and no independent leakage outside that selected readout sector. The source supplies the evaporation clock; the Peirce relay supplies the memory law; the exterior Maxwell problem supplies the ordinary photonic support. State-dependent interiors and algebraic descriptions of gravitational subregions address related aspects of the factorization problem \cite{PapadodimasRaju2013,PapadodimasRaju2016,Harlow2014PR,Witten2022Crossed,JensenSorceSperanza2023,FaulknerSperanza2024}. Here the starting point is the exceptional retained horizon algebra, followed by the ordinary exterior detector selected from it.

The mechanism is Peirce--Volterra memory \cite{McCrimmon2004,Gripenberg1990}. The hidden Albertian complement reaches the exterior detector through an interface relay, leaving an accumulated two-time coherence in the emitted history. The primary diagnostic is therefore a correlation in the emitted history rather than a large one-time spectral distortion. Page's analysis treats information as a property of the radiation history, and Mathur's theorem shows that arbitrary small local corrections to Hawking pair creation do not supply the purification mechanism \cite{Page1993,Mathur2009SmallCorrections}. Multi-time analyses sharpen the relevant diagnostic by showing how Hawking radiation can carry history-sensitive temporal correlations \cite{AnastopoulosSavvidou2020MultiTime}. The present construction gives that diagnostic a channel-specific origin: the connected covariance $G_{\rm conn}^{(\chi)}(v,s)$ of windowed helicity/Stokes observables is generated by the Peirce relay and is realized in the selected two-helicity exterior detector. Related soft and exterior-channel perspectives also point toward weak information carriers outside a violent horizon-scale spectral distortion \cite{Giddings2012Nonviolent,HawkingPerryStrominger2016SoftHair}.

The exterior detector has a controlled photonic realization: the gauge-invariant spin-$1$ problem on adiabatic Reissner--Nordstr{\"o}m slices, with the two transverse Maxwell helicities furnishing the two-component support \cite{lindquist1965,bonnorvaidya1970,Page1976I,Page1976II,HiscockWeems1990,moncrief1974a,moncrief1975,zerilli1974,chandrasekhar1983}. The low-frequency pure-photon luminosity scales as $r_+^4T_H^6$, distinct from the Stefan-like $A_+T_H^4$ law \cite{hod2016,ngampitipan2013}; the RN clock remains an effective evaporation clock after the photonic support is made explicit. Single-time occupations follow the greybody-filtered Hawking law at the adiabatic order used here. The new exterior observable is the connected two-time covariance of the emitted history: the mode-resolved fixed-support kernel induces a spectral overlap $W(v,s)$ and a phase-locked polarimetric covariance whose prefix integral gives the scalar coherence $c(v)$.

Superstatistics enters as the macroscopic Euclidean reading of the relay memory. Once the hidden complement has been integrated out and the exterior block has closed as a Volterra law on the RN clock, the opening and late regimes define different temporal ensembles. In the Beck--Cohen sense, a superstatistical representation is a positive Laplace superposition of local exponential weights, appropriate when local exponential response survives while the effective intensive data fluctuate on a slower scale \cite{BeckCohen2003,Beck2007,Bernstein1929,Widder1941}. Such mixtures underlie Tsallis--Pareto tails in nonequilibrium high-energy spectra \cite{Beck2009HEP,Tsallis1988,Tsallis2009}, and related nonextensive lattice gauge settings \cite{Frigori2014NonextensiveLGT}; the one-sided onset used here is represented by the Tsallis--Pareto/Lomax family \cite{Shalizi2007,Lomax1954}. The near-extremal branch is a shifted-L{\'e}vy residence law, the standard canonical form associated with broad waiting-time, first-passage, and residence-time statistics \cite{Feller1971,Redner2001,ZaburdaevDenisovKlafter2015}.

The Lorentzian channel solution gives the reduced detector entropy and the connected two-time coherence of the selected exterior block. The Euclidean transfer form of the Peirce--Volterra kernel gives branch weights on the RN clock: a regular-opening Tsallis/Lomax branch and a near-extremal shifted-L{\'e}vy residence branch. In the complete selected-detector background, the accumulated readout algebra is the outgoing algebra of that detector. The two Euclidean endpoint actions are then placed on the common source-fixed RN clock used to compare the branch actions, and their branch-admissible lower envelope reconstructs the Page-curve envelope. This gives a channel-level analogue of saddle/envelope dominance, rather than an Island prescription. The same positive branch measures also admit a replica/Mellin reading: shared-scale moments generate non-factorizing branch weights, providing the channel counterpart of the connected replica saddle used in Island calculations. The same Peirce--Volterra memory therefore has two controlled representations: operationally, it appears as a two-time exterior covariance; thermodynamically, it appears as a two-branch Euclidean transfer problem with an effective replica structure.

The paper is organized as follows. Section~\ref{sec:Albertian} proves the attractor--Wilsonian Albertian closure of the retained symbol algebra. Section~\ref{sec:Peirce} derives the Peirce relay architecture and identifies the ordinary detector block. Section~\ref{sec:background} fixes the RN trajectory and the electric Freudenthal slice used in the $5$D$\to4$D descent. Section~\ref{sec:channel} derives the Volterra memory law and the accumulated readout algebra. Section~\ref{sec:state} extracts the readout qubit, its spectral gap, its photonic realization, its transport rates, and the connected two-time readout observable. Section~\ref{sec:superstatistics} gives the Euclidean transfer reading of the Peirce--Volterra kernel and derives the branch-admissible Page envelope. Section~\ref{sec:conclusion} returns to the AMPS question posed above. Appendix~\ref{app:kernel} fixes the Albertian conventions and collects the source-fixed derivations behind the attractor map, the ordinary-slice support, the WKB barrier factorization, the RN clock, and the explicit kernel. Appendix~\ref{app:qubit} closes the hidden-sector history and records the reduced-channel admissibility criterion. Appendix~\ref{app:photonic} gives the gauge-invariant photonic support, the fixed-support emitted-history coherence, and the channel-rate formulas. Appendix~\ref{app:branchappendix} records the outgoing-history identification and the branch calculations used in Sec.~\ref{sec:superstatistics}: the pushforward form of the branch measure, its replica/Mellin tower, the regular-opening Tsallis/Lomax endpoint, the near-extremal shifted-L{\'e}vy residence endpoint, and the no-global-positive-transform argument for the hard branch envelope.

\section{The sectorial Albertian horizon and canonical structure}
\label{sec:Albertian}

This section fixes the finite horizon-side algebra used below. The claim is
sectorial: it identifies the ordered real algebra retained by the
bosonic two-derivative Bogomol'nyi--Prasad--Sommerfield (BPS) attractor sector of octonionic magical
supergravity. The Peirce readout and the sourced evaporation benchmark are
introduced only after this algebra has been fixed.

\begin{definition}[Retained attractor--Wilsonian symbol sector]
\label{def:SAW}
Let \(\mathcal S_{\rm AW}\) denote the retained attractor--Wilsonian sector
obtained from the bosonic two-derivative fields of octonionic magical
supergravity by imposing gauge-invariant horizon-side observables,
near-horizon BPS attractor boundary conditions, and retention of the
 two-derivative attractor data. A retained horizon-side symbol is a
gauge-invariant zero-mode function of the two-derivative near-horizon
attractor data that survives this Wilsonian projection. Independent soft,
boundary, higher-derivative, non-BPS, dyonic, and propagating symbols are not
generators of \(\mathcal S_{\rm AW}\) unless the sector is explicitly enlarged
to include them. This definition fixes the physical regime; it does not assume
an Albertian observable algebra in advance.
\end{definition}

\begin{theorem}[Attractor--Wilsonian Albertian closure]
\label{thm:AlbertianClosure}
Consider \(\mathcal S_{\rm AW}\) in the sense of
Definition~\ref{def:SAW}. Then the independent retained horizon-side symbols
are generated by the octonionic attractor package
\begin{equation}
\label{eq:attractor}
\begin{gathered}
Q\in J_3(\Oa),
\qquad
Q^{\#},
\qquad
N(Q),
\\[2pt]
X_*(Q)=\frac{Q^{\#}}{N(Q)^{2/3}} .
\end{gathered}
\end{equation}
The retained real symbol space carries the ordered package
\begin{equation}
\label{eq:ordered-package}
\bigl(
V_{\rm ret},
C_+,
\mathbf 1,
\mathrm{Tr},
\#,
N
\bigr),
\end{equation}
where \(C_+\) is the positive cone inherited from the retained
BPS-admissible octonionic attractor data. Equivalently, it is the cone generated
by positive combinations of primitive BPS-admissible idempotent attractor data.
This is the homogeneous self-dual cone with order unit \(\mathbf 1\) carried
by the retained attractor symbols. By Jordan cone reconstruction, the
corresponding Euclidean Jordan algebra is the Albert algebra \(J_3(\Oa)\).
Hence the retained finite horizon-side ordered symbol algebra is Albertian:
\begin{equation}
\label{eq:Hret}
\AH^{\rm ret}\cong J_3(\Oa).
\end{equation}
\end{theorem}

\begin{proof}
The proof has two logically separate parts. First, the physical restriction
defining \(\mathcal S_{\rm AW}\) leaves no independent retained generators
beyond the attractor charge package. Second, the ordered cone carried by that
package reconstructs the Albert algebra.

The two-derivative bosonic field content consists of the metric, vector
fields, and scalar fields. Gauge invariance removes vector potentials as
elementary local symbols in the retained near-horizon sector. The
gauge-invariant vector data are the conserved fluxes. In octonionic magical
supergravity these fluxes are organized by an element
\[
Q\in J_3(\Oa)
\]
with cubic norm \(N(Q)\) and quadratic adjoint \(Q^{\#}\)
\cite{GunaydinSierraTownsend1983,GunaydinSierraTownsend1984,FerraraGunaydin2006,BianchiFerrara2008,GunaydinKidambi2022}.

The BPS attractor mechanism fixes the horizon scalar data in terms of the
charge. In the Jordan form used here this map is
\[
X_*(Q)=\frac{Q^{\#}}{N(Q)^{2/3}},
\]
so the retained scalar zero-mode symbols add no independent generators
\cite{FerraraKalloshStrominger1995,FerraraGibbonsKallosh1997,CeresoleFerraraMarrani2007}.
The retained near-horizon geometric scale and the two-derivative BPS entropy
are controlled by the same cubic invariant \(N(Q)\). Thus the metric, vector,
and scalar zero-mode symbols retained in \(\mathcal S_{\rm AW}\) are generated
by
\[
Q,
\qquad
Q^{\#},
\qquad
N(Q).
\]

It remains to identify the ordered algebra carried by these symbols. The
octonionic attractor package is not merely a list of charge components. In the
octonionic magical theory, the five-dimensional charge sector is organized by
the cubic Jordan algebra \(J_3(\Oa)\); this is the standard Jordan form of the
magical Maxwell--Einstein system
\cite{GunaydinSierraTownsend1983,GunaydinSierraTownsend1984,FerraraGunaydin2006,BianchiFerrara2008,GunaydinKidambi2022}.
Thus the retained charge variable comes with the structural maps of the
Albert algebra: the order unit \(\mathbf 1\), the trace form, the quadratic
adjoint \(X^{\#}\), and the cubic norm \(N(X)\)
\cite{Jordan1934,Albert1934,Schafer1966,McCrimmon2004,Krutelevich2007Jordan}.

These maps determine the cubic Jordan structure. In particular, polarizing the
quadratic adjoint gives the cross product
\begin{equation}
\label{eq:cross-product}
X\times Y
=(X+Y)^{\#}-X^{\#}-Y^{\#}.
\end{equation}
Together with the trace form and the order unit, this is the standard cubic
Jordan package of the Albert algebra
\cite{McCrimmon2004,Krutelevich2007Jordan,FarautKoranyi1994SymmetricCones}.

The cone \(C_+\) used here is the positive cone of the retained octonionic
attractor symbols. It can be described equivalently as the cone generated by
positive combinations of primitive BPS-admissible idempotent attractor data,
or as the spectral positive cone of \(J_3(\Oa)\) in the retained BPS
orientation. On the electric Freudenthal slice this is the exceptional
positive cone associated with \(J_3(\Oa)\). By the Koecher--Vinberg
reconstruction theorem, a finite-dimensional homogeneous self-dual cone with
order unit reconstructs a Euclidean Jordan algebra; for this exceptional cone
the reconstructed algebra is precisely \(J_3(\Oa)\)
\cite{Koecher1957Positivitaetsbereiche,Vinberg1963HomogeneousCones,FarautKoranyi1994SymmetricCones,BellissardIochum1978,Krutelevich2007Jordan}.
This proves Eq.~\eqref{eq:Hret}.
\end{proof}

With this definition in place, there is no second hidden observable algebra
inside the retained sector. The gauge-invariant quantities that survive the
attractor--Wilsonian projection are precisely the effective observables of
\(\mathcal S_{\rm AW}\). Hence the ordered symbol algebra reconstructed above
is the effective ordered observable algebra of the retained sector:
\begin{equation}
\label{eq:symbols-observables}
\mathcal A_{\rm ret}^{\rm obs}
\equiv
\mathcal A_{\rm ret}^{\rm sym}
\cong
J_3(\Oa).
\end{equation}

The algebraic-quantum description now follows from the ordered observable algebra.
Algebraic quantum mechanics assigns quantum kinematics to an ordered real
algebra of observables \cite{Segal1947,HaagKastler1964,Haag1992}.
For shorthand, within this retained finite horizon-side sector, we write
\begin{equation}
\label{eq:H}
\AH
\cong
J_3(\Oa).
\end{equation}
This shorthand always refers to the effective ordered observable algebra
identified above. The finite dimension of \(J_3(\Oa)\) is an algebraic
dimension. It is not a microscopic count of the Bekenstein--Hawking
degeneracy, which belongs to the microscopic completion rather than to the
retained finite symbol sector.

Thus the Albertian AQM identification is structural rather than postulated.
Once the retained horizon-side observables form \(J_3(\Oa)\), the
Jordan--von Neumann--Wigner classification places the sector in the unique
exceptional finite-dimensional simple formally real Jordan case. Since this
algebra is non-special, it is not realizable as the self-adjoint part of an
associative Hilbert-space operator algebra with the global Jordan product
preserved. The construction is therefore a finite Jordan-AQM description
induced by the effective ordered observable algebra above, not a deformation
quantization of the full supergravity phase space
\cite{Jordan1934,GunaydinPironRuegg1978,HancheOlsen1983,McCrimmon2004}.

\begin{corollary}[Albertian AQM of the retained sector]
\label{cor:AlbertianAQM}
For the retained sector of Theorem~\ref{thm:AlbertianClosure}, the
corresponding finite quantum kinematics is the exceptional algebraic quantum
mechanics of G\"unaydin, Piron, and Ruegg \cite{GunaydinPironRuegg1978}.
States are normalized positive linear functionals on \(J_3(\Oa)\), sharp events
are primitive idempotents, and reversible transformations are generated by
\begin{equation}
\label{eq:f4}
\mathrm{Der}(J_3(\Oa))\cong\mathfrak f_4 .
\end{equation}
\end{corollary}

Therefore, the resulting Heisenberg and dual state laws
\cite{Barnum2014,Niestegge2015} are the internal reversible laws of this
retained finite Jordanian algebra. Complementary sectors enter only after the
Wilsonian projection used in Theorem~\ref{thm:AlbertianClosure} is enlarged.
We write \(\delta\) for derivations of \(J_3(\Oa)\) and reserve \(D\) for the
projected channel blocks introduced after the Peirce/readout reduction. If
\(\delta_t\in\mathrm{Der}(J_3(\Oa))\), then the Heisenberg law is
\begin{equation}
\label{eq:heis}
\frac{\dd}{\dd t}A_t=\delta_t(A_t),
\qquad
A_t=\alpha_{t,s}(A_s),
\end{equation}
with automorphisms \(\alpha_{t,s}\) satisfying
\begin{equation}
\label{eq:auto-flow}
\frac{\dd}{\dd t}\alpha_{t,s}
=\delta_t\circ\alpha_{t,s},
\qquad
\alpha_{s,s}=\mathrm{id}.
\end{equation}
The dual evolution of states is
\begin{equation}
\label{eq:dualstate}
\varpi_t(A)=\varpi_0\!\left(\alpha_{t,0}(A)\right),
\qquad
\frac{\dd}{\dd t}\varpi_t(A)=\varpi_t(\delta_t A).
\end{equation}
On associative special sectors, Eqs.~\eqref{eq:heis}--\eqref{eq:dualstate}
reduce to the ordinary Heisenberg and von Neumann laws.

\begin{proposition}[Canonical Albertian package]
\label{prop:canonical-structure}
Given Eq.~\eqref{eq:H}, the state space, the cubic invariant, and the
Freudenthal envelope are canonically associated:
\begin{equation}
\label{eq:package}
\AH\cong J_3(\Oa)
\Longrightarrow
\bigl(
\Omega(J_3(\Oa)),
\ I_3=N,
\ \FTS(J_3(\Oa)),
\ I_4
\bigr).
\end{equation}
In particular, the five-dimensional entropy functional and the electric
four-dimensional descent used below can be written in Albertian variables
\cite{FerraraGunaydin2006,BianchiFerrara2008,CeresoleFerraraMarrani2007}.
\end{proposition}

\begin{remark}
Equation~\eqref{eq:package} records the canonical structures attached to the
retained Albertian sector. It does not yet select a Peirce representative, an
exterior detector, or a sourced evaporation trajectory. Those choices enter in
the Peirce/readout construction and in the sourced RN benchmark developed
below.
\end{remark}

\section{Peirce relay architecture and ordinary readout}
\label{sec:Peirce}

Peirce theory \cite{McCrimmon2004,Petersson2004,Petersson2019} turns the Albertian structure of Section~\ref{sec:Albertian} into the observable architecture used in the horizon problem. It first isolates the minimal sharp image associated with the chosen asymptotic sector, while the internal complex slice identifies the associative radiative block that carries the exterior detector.

For a primitive idempotent $e\in J_3(\Oa)$, the standard Peirce theory gives the decomposition
\begin{equation}
J_3(\Oa)=\Atwo\oplus \Aone\oplus \Azero,
\qquad \Atwo=\R e.
\label{eq:Peirce}
\end{equation}
The one-dimensional piece $A_2(e)=\R e$ is the minimal sharp image selected by $e$.
Its associated Peirce projector $\Pi_e$ is the projector onto $A_2(e)$.
The corresponding minimal sharp asymptotic image is therefore
\begin{equation}
\Ainf(e):=\Atwo=\R e.
\label{eq:Ainf}
\end{equation}
In the present setting, $A_\infty(e)$ is the smallest associative image extracted from the exceptional horizon algebra. The full detector algebra appears after the rank-two readout face and its complex slice are selected.

Primitive idempotents form the Cayley plane \cite{Baez2002,Petersson2019},
\begin{equation}
\Oa P^2\cong F_4/\mathrm{Spin}(9),
\label{eq:CayleyPlane}
\end{equation}
so the admissible minimal asymptotic image is rigid but not absolute.
\begin{proposition}[Uniqueness up to automorphism]
For any two primitive idempotents $e$ and $f$ there exists $g\in F_4$ such that $g(e)=f$.
Consequently,
\begin{equation}
g\Pi_e g^{-1}=\Pi_f,
\end{equation}
so the admissible minimal asymptotic image is unique only up to automorphism.
\end{proposition}

\begin{proof}
The transitive $F_4$ action on $\Oa P^2$ sends any primitive idempotent to any other.
Automorphisms preserve the Jordan product, hence the Peirce decomposition and the associated rank-one projector.
\end{proof}

There is an equally sharp converse.
\begin{proposition}[No absolute channel]
The abstract Albert algebra does not select an absolute asymptotic channel.
\end{proposition}

\begin{proof}
If a projector $\Pi_*$ were selected canonically by the abstract Albert algebra alone, it would have to be invariant under all automorphisms.
The corresponding primitive idempotent would then be fixed by the full $F_4$ action.
But the $F_4$ action on primitive idempotents is transitive and nontrivial, so no such fixed point exists.
\end{proof}

Albertian AQM therefore determines a unique orbit of admissible asymptotic images, not an absolute representative.
A concrete representative is fixed only after macroscopic asymptotic data select a superselection sector---for example a charge sector, asymptotic vacuum, or boundary datum inside the same asymptotically flat solution.
A primitive idempotent $e$ fixes the sharp asymptotic sector and the corresponding Peirce frame. It labels the asymptotic channel relative to which the observable algebra is represented and does not become an additional dynamical variable on the sourced benchmark.

In the source-fixed ray-preserving trajectory, the charge datum $Q_0$ fixes the attractor-adapted Jordan frame. The exterior detector selects a rank-two readout face inside that frozen frame, and $e$ is the complementary primitive idempotent. The complex slice $u$ is fixed operationally by the transverse two-helicity photon detector. Thus $(e,u)$ is the Peirce/readout datum selected by the frozen attractor frame and by the exterior detector.

The standard Peirce multiplication rules and the detector-selected ordinary slice are used in the form stated here; only the later sourced benchmark derivations are deferred to Appendices~\ref{app:attractor} and \ref{app:sourcedclockdetail}, while the projected-block and kernel derivations are deferred to Appendices~\ref{app:projectedblocks}--\ref{app:kernel_source}.
To recover an ordinary internal channel, choose a unit imaginary octonion $u\in\mathrm{Im}\,\Oa$ and the associated complex slice
\begin{equation}
\Cu=\mathrm{span}_{\R}\{1,u\}\subset \Oa.
\label{eq:Cu}
\end{equation}
Here $\Cu$ is a real two-dimensional associative subalgebra of the octonions isomorphic to $\C$; it is not a complexification of the Albert algebra.
The choice of $u$ fixes the ordinary associative readout sector. Different choices are related by the octonionic automorphism group and therefore belong to the same admissible orbit.
There is a canonical associative embedding \cite{McCrimmon2004,Petersson2004,Petersson2019}
\begin{equation}
H_3(\Cu)\hookrightarrow J_3(\Oa),
\label{eq:ordinaryslice}
\end{equation}
inside which the ordinary radiative block
\begin{equation}
\Qu\cong H_2(\C)
\label{eq:Qu}
\end{equation}
appears as a genuine two-level matrix system.
For the canonical primitive idempotent
\begin{equation}
e=\mathrm{diag}(1,0,0),
\label{eq:canonical-e}
\end{equation}
the Peirce blocks admit a concrete matrix realization that will be used later in the relay reduction. A generic interface element $a\in A_1(e)$ and a generic complementary element $q\in A_0(e)$ may be written as
\begin{equation}
a=
\begin{pmatrix}
0 & \bar z & \bar w\\
z & 0 & 0\\
w & 0 & 0
\end{pmatrix},
\qquad
q=
\begin{pmatrix}
0 & 0 & 0\\
0 & \alpha & x\\
0 & \bar x & \beta
\end{pmatrix},
\label{eq:Peirce-explicit}
\end{equation}
with $\alpha,\beta\in\R$ and $w,x,z\in\Oa$. Thus $A_1(e)$ is the explicit Peirce bridge linking the hidden sector to the lower $2\times 2$ block, while $A_0(e)$ is the lower block on which the ordinary readout slice is selected. Choosing a unit imaginary octonion $u$ and the associated complex slice $\Cu\subset\Oa$ restricts the lower block to the ordinary associative copy $H_2(\C)\subset H_2(\Cu)\subset A_0(e)$; this is the readout block $Q_u$, while its orthogonal complement inside $A_0(e)$ is the exceptional remainder $K_u$. Appendix~\ref{app:attractor} records this realization in full and fixes the notation used later in the source-fixed channel reduction.
Within the selected complex slice, this is the ordinary rank-two readout block seen by the asymptotic detector; trace, positivity, and von Neumann entropy survive there without reinstating the full globally associative AMPS stage.
On the controlled neutral-source benchmark fixed in Sec.~\ref{sec:background}, the attractor datum, the adapted frame $(e,u)$, and the electric polarization remain fixed, so these choices do not become additional dynamical variables.

The complementary exceptional sector inside $\Azero$ is denoted by $\Ku$, so that
\begin{equation}
\Azero=\Qu\oplus \Ku,
\qquad 1+16+4+6=27.
\label{eq:refineddecomp}
\end{equation}
Here the four summands count the real dimensions of $A_2(e)$, $A_1(e)$, $Q_u$, and $K_u$ respectively: $\dim_\R A_2(e)=1$, $\dim_\R A_1(e)=16$, $\dim_\R Q_u=4$, and $\dim_\R K_u=6$, adding up to $\dim_\R J_3(\Oa)=27$.
The sector $K_u$ is the genuinely exceptional remainder inside $A_0(e)$: it has no ordinary matrix analogue and is not directly accessible to the asymptotic detector.
The resulting relay architecture is
\begin{equation}
\Ku\longrightarrow \Aone \longrightarrow \Qu.
\label{eq:relaychain}
\end{equation}
The meaning of Eq.~\eqref{eq:relaychain} is physical, not merely diagrammatic: information stored in the complementary exceptional sector can reach the radiative block only through the interface $A_1(e)$, which functions as the relay between the hidden horizon sector and the ordinary outgoing channel.

Figure~\ref{fig:penrose} summarizes the channel-memory architecture rather than a literal causal Penrose diagram.
The upper red region represents the globally exceptional horizon sector governed by $J_3(\Oa)$.
The gold dashed Peirce projection marks the extraction of the minimal sharp image $A_\infty(e)=\R e$.
The full ordinary readout is the rank-two block $Q_u\simeq H_2(\C)$, selected only after the complex slice $u$ is fixed; the interface $A_1(e)$ relays hidden exceptional data from $K_u$ into that radiative block.
The figure is the geometric shorthand for the three ingredients that organize the rest of the paper: an exceptional horizon sector, a Peirce-selected readout, and a relay architecture that turns hidden exceptional data into outgoing channel memory.

\begin{figure}[t!]
    \centering
    \includegraphics[width=1.0\linewidth]{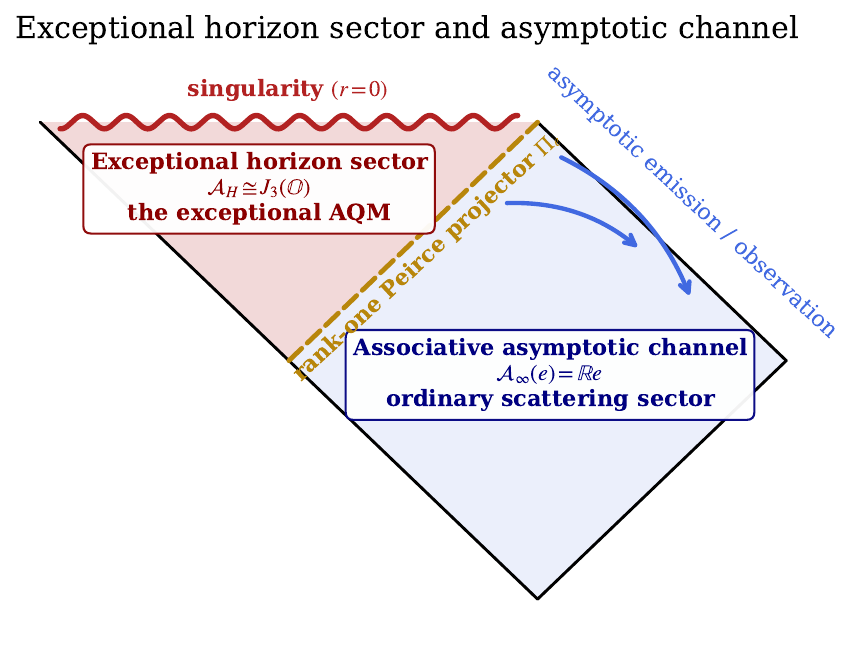}
    \caption{Schematic summary of the channel-memory architecture used throughout the paper; it is not a literal causal Penrose diagram.
The upper red region represents the exceptional horizon sector $\mathcal A_H\cong J_3(\mathbb O)$.
The gold dashed line indicates the Peirce projection onto the minimal sharp image $A_\infty(e)=\mathbb R e$.
The full ordinary readout is the rank-two block $Q_u\cong H_2(\mathbb C)$, selected after the complex slice $u$ is fixed.
The interface $A_1(e)$ relays hidden exceptional data from $K_u$ into that radiative block.
The figure condenses the logic of Section~\ref{sec:Peirce}: local two-generated special sectors remain compatible with no-drama effective reasoning, an ordinary readout survives, and the obstruction concerns the global exceptional stage required by AMPS rather than the detector algebra itself.}
    \label{fig:penrose}
\end{figure}

\begin{proposition}[Albert--Freudenthal architecture]
Assume Eq.~\eqref{eq:H}.
Then the relevant horizon physics is organized by the nested chain
\begin{equation}
\mathcal{A}^{(2)}_{\mathrm{loc}}
\subset H_3(\Cu)
\subset J_3(\Oa)
\subset \FTS\bigl(J_3(\Oa)\bigr),
\label{eq:squeeze}
\end{equation}
with the following properties:
\begin{enumerate}[label=(\roman*)]
\item the local two-generated sectors $\mathcal A_{\mathrm{loc}}^{(2)}$ remain effectively special and therefore compatible with smooth-horizon effective-field-theory reasoning;
\item the internal associative channel $H_3(\Cu)$ supports ordinary trace, positivity, and von Neumann entropy, with radiative block $\Qu\cong H_2(\C)$;
\item the global exceptional sector $J_3(\Oa)$ obstructs the globally associative AMPS tensor-factor stage with ordinary qubit-like radiation;
\item the electric four-dimensional slice sits naturally inside the Freudenthal envelope $\FTS(J_3(\Oa))$, whose quartic invariant reduces on the natural Jordanian slices to the cubic Albert data.
\end{enumerate}
\end{proposition}
\begin{proof}
Property (i) follows from two-generated specialness: in the octonionic realization this is Artin's theorem, and in Jordan form it is the Shirshov--Cohn two-generator specialness theorem \cite{Baez2002,Schafer1966,Cohn1954,JacobsonPaige1957}.
Property (ii) follows from the embedding~\eqref{eq:ordinaryslice} and the ordinary block~\eqref{eq:Qu}.
Property (iii) is the exceptional compositional obstruction: Hanche-Olsen excludes the standard qubit composite for an exceptional Jordan factor, and Barnum--Graydon--Wilce sharpen the same conclusion in the language of Euclidean Jordan composites \cite{HancheOlsen1983,BarnumGraydonWilce2020}.
Property (iv) is the content of Proposition~\ref{prop:canonical-structure}, which identifies the canonical Freudenthal envelope and its reduction on the Jordanian slices.
\end{proof}
The compositional obstruction used here is a statement about the retained Jordan-AQM observable sector. It is not asserted as a theorem about the full operator algebra of quantum fields in the black-hole spacetime. Its role is to show that the AMPS Hilbert-factor stage cannot be installed inside the retained Albertian horizon algebra without replacing that sector by an associative composite.

The sectorial Albertian identification isolates the three scales used below: local smooth-horizon reasoning in the two-generated sectors, an ordinary radiative channel in the complex slice, and the global obstruction at the full exceptional level. In this sense Peirce theory does not merely decompose the algebra; it organizes the horizon.
The next sections keep precisely this architecture fixed while adding the sourced background and the reduced channel dynamics.

\section{Sourced background, attractor anchor, and electric descent}
\label{sec:background}

The five-dimensional attractor equations may be written in base-free Jordan form as \cite{FerraraKalloshStrominger1995,FerraraKallosh1996,FerraraGibbonsKallosh1997,MeessenOrtinPerzShahbazi2012d5,CeresoleFerraraMarrani2007}
\begin{equation}
Q=Z_*X_*^{\#},
\qquad
X_*(Q)=\frac{Q^{\#}}{N(Q)^{2/3}}.
\label{eq:basefree}
\end{equation}
Suppose that the sourced charge trajectory preserves a Jordanian charge ray,
\begin{equation}
Q(v)=s(v)Q_0.
\label{eq:charge-ray}
\end{equation}
Because the adjoint is quadratic and the cubic norm is homogeneous of degree three, one obtains the identity
\begin{equation}
X_*(Q(v))=X_*(Q_0).
\label{eq:survivalid}
\end{equation}
The exponent $2/3$ is fixed by cubic homogeneity: under $Q\mapsto sQ$, one has $Q^{\#}\mapsto s^2Q^{\#}$ and $N(Q)^{2/3}\mapsto s^2N(Q)^{2/3}$. Thus the attractor-adapted horizon frame remains frozen on the charge-ray-preserving class. If a non-benchmark fluctuation $\delta Q$ is retained, then to first order
\begin{equation}
\delta X_* = N(Q)^{-2/3}\left[\delta Q^{\#}-\frac{2}{3}Q^{\#}\frac{\delta N}{N}\right]+O(\delta Q^2),
\label{eq:delta-attractor-frame}
\end{equation}
so Peirce-frame transport is suppressed only when $\|\delta Q\|/\|Q\|$ is small and the charge is nondegenerate. Such stochastic charge-ray or frame-transport corrections are outside the controlled neutral-source benchmark. Appendix~\ref{app:attractor} proves Eq.~\eqref{eq:survivalid} line by line.
The reference background used throughout the quantitative part of the paper is the neutral-source fixed-charge member. Here ``neutral-source RN'' means a fixed-charge Reissner--Nordstr{\"o}m background whose null source drains mass but carries no independent charge current. Thus
\begin{equation}
Q(v)=Q_0,
\label{eq:neutralbg}
\end{equation}
for which the excitation above extremality is
\begin{equation}
\eexc(v)=M(v)-M_{\mathrm{BPS}}(Q_0),
\label{eq:Eexc}
\end{equation}
and the effective evaporative clock is
\begin{equation}
\Gamma_{\mathrm{evap}}(v)=-\frac{\dd}{\dd v}\log \eexc(v).
\label{eq:GammaEvap}
\end{equation}
The equivalent emitted-fraction parametrization is
\begin{equation}
r(v)=1-\frac{\eexc(v)}{\eexc(0)},
\qquad
\eexc(v)=\eexc(0)[1-r(v)].
\label{eq:rdef}
\end{equation}
The variables $v$, $\eexc(v)$, and $r(v)$ therefore describe the same history in this setting. Appendix~\ref{app:sourcedclockdetail} derives the normalized sourced evaporative clock used later in the benchmark figures and rate formulas.
At the bosonic level, the sourced problem is governed by a null-fluid sourced Einstein--Maxwell system of Vaidya type \cite{Vaidya1943}
\begin{align}
G_{MN}&=T^{\mathrm{SUGRA}}_{MN}+T^{H(5)}_{MN},
\label{eq:sourcedEinstein}
\\
\nabla_M\bigl(a_{IJ}(\phi)F^{JMN}\bigr)&=J^{H(5)N}_I,
\label{eq:sourcedMaxwell}
\end{align}
with null source
\begin{equation}
\begin{aligned}
T^{H(5)}_{MN}&=\rho_5(v,r)\,\ell_M\ell_N,\\
J^{H(5)N}_I&=j_I(v,r)\,\ell^N,\\
\ell_M\ell^M&=0.
\end{aligned}
\label{eq:nullsource}
\end{equation}
In the neutral background one sets
\begin{equation}
J^{H(5)N}_I=0,
\label{eq:no-current}
\end{equation}
so the null fluid drains the mass while the attractor datum and the scalar couplings remain frozen.
This makes the reference background both controlled and nontrivial.
For the sourced metric
\begin{equation}
\begin{split}
 ds^2 &= -f(u,r)du^2 - 2\,du\,dr + r^2d\Omega_2^2, \\
 f(u,r) &= 1 - \frac{2M(u)}{r} + \frac{Q_0^2}{r^2},
\end{split}
\label{eq:vaidya-rn-main}
\end{equation}
one finds
\begin{equation}
G_{\mu\nu}-8\pi T^{\mathrm{EM}}_{\mu\nu}
= -\frac{2\dot M(u)}{r^2}\,\delta_\mu^u\delta_\nu^u,
\label{eq:extrastress-main}
\end{equation}
so the geometric mass-loss term requires an additional null-fluid source.
The species-resolved photonic calculation is therefore performed on adiabatic RN slices of this sourced background, not by reinterpreting the effective metric itself as a pure-photon Einstein--Maxwell solution. The adiabatic treatment is restricted to modes whose scattering time is short compared with the sourced-background variation time; equivalently, for a representative mode frequency $\omega$ one requires $\omega\tau_{\rm bg}(v)\gg1$, with $\tau_{\rm bg}$ defined explicitly in Appendix~\ref{app:spin1}.
Appendix~\ref{app:spin1} gives the component calculation and the Maxwell uplift.

We use the standard Kaluza--Klein reduction from the five-dimensional Jordanian charge sector to the four-dimensional Freudenthal charge description \cite{GunaydinSierraTownsend1984,FerraraGunaydin2006,Krutelevich2007Jordan}. In this notation $y$ is the compact coordinate, $A^0$ is the Kaluza--Klein vector, $\varphi$ is the reduction scalar, and $\zeta^I$ are the axions descending from the five-dimensional gauge fields. The metric and gauge fields are written as
\begin{align}
ds_5^2&=e^{\varphi/\sqrt3}ds_4^2+e^{-2\varphi/\sqrt3}(\dd y+A^0)^2,
\label{eq:KKmetric}
\\
A^I_{(5)}&=A^I_{(4)}+\zeta^I(\dd y+A^0).
\label{eq:KKgauge}
\end{align}
The four-dimensional charge vector is organized as an element of the Freudenthal triple system (FTS) $\FTS(J_3(\Oa))$. The electric slice is parametrized by the Freudenthal normalization constant $\alpha_{\mathrm{FTS}}$, an overall scale convention in the electric FTS block that is absorbed into the charge normalization $Q_0$ below. On this slice the magnetic Jordan block and the axionic Wilson-line components are set to zero, so one imposes
\begin{equation}
\zeta^I=0,
\qquad
x_e=
\begin{pmatrix}
\alpha_{\mathrm{FTS}} & A\\
0 & 0
\end{pmatrix},
\label{eq:xelectric}
\end{equation}
so that the quartic Freudenthal invariant reduces to
\begin{equation}
I_4(x_e)=-4\alpha_{\mathrm{FTS}}N(A),
\qquad
Q_0=|I_4(x_e)|^{1/4}.
\label{eq:I4slice}
\end{equation}
This electric slice is the one used in the source-fixed closure of the bridge coefficients. It is not a generic four-dimensional Freudenthal charge vector: the complementary FTS block is set to zero in this benchmark, so the four-dimensional realization remains controlled by the cubic Albertian invariant inherited from the five-dimensional Jordanian charge sector. The normalization $\alpha_{\mathrm{FTS}}$ fixes the conventional scale of the FTS singlet and is absorbed into $Q_0$ in the channel calculation.
Appendix~\ref{app:attractor} records the descent in the form used later in the qubit and superstatistical analyses.
Figure~\ref{fig:sourcedbg} shows the sourced background on the common normalized interval $0\le v\le 10$.
The upper panel plots the normalized excitation energy together with the peak-normalized evaporative pulse $\Gamma_{\mathrm{evap}}(v)/\Gamma_{\mathrm{evap}}^{\max}$, so that the decay and rate profiles can be compared on the same scale.
The lower panel uses the exact identity $\eexc(v)=M(v)-M_{\mathrm{BPS}}(Q_0)$ and displays the frozen BPS anchor together with the sourced excitation above it on the same interval.
The benchmark is controlled by the decay of $\eexc(v)$, not by an externally imposed pulse.
\begin{figure}[t!]
    \centering
    \includegraphics[width=1.00\linewidth]{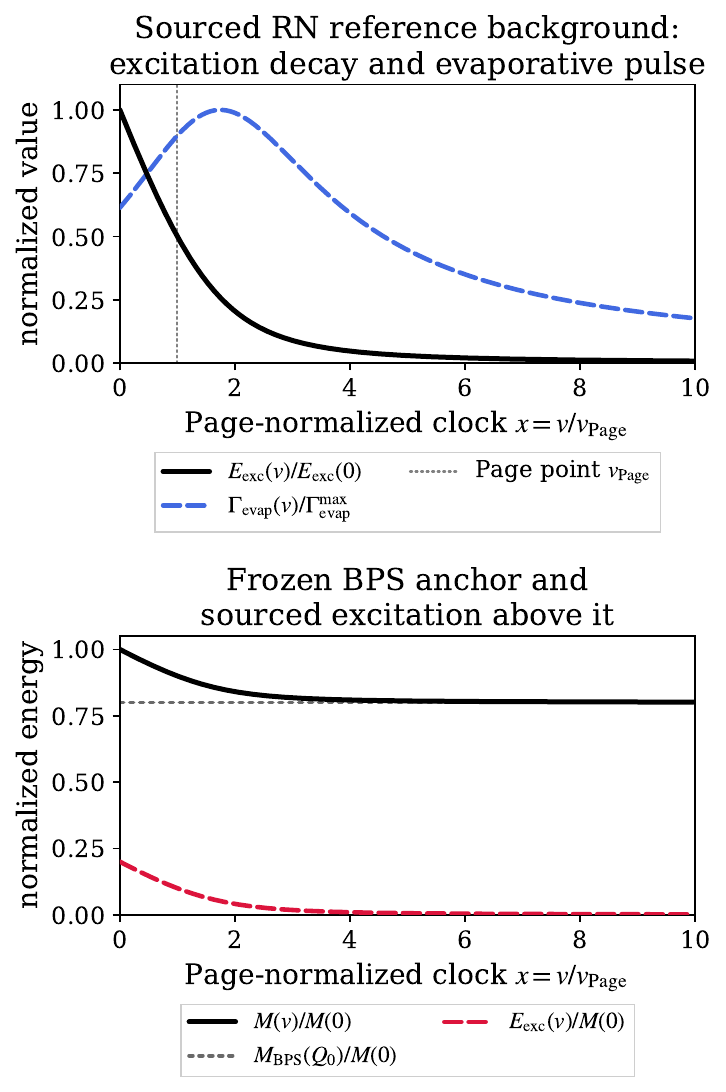}
    \caption{Neutral-source fixed-charge Reissner--Nordstr{\"o}m background. Top: normalized excitation energy and peak-normalized evaporative pulse $\Gamma_{\mathrm{evap}}(v)/\Gamma_{\mathrm{evap}}^{\max}$; the vertical marker is the Page point defined by $\eexc(v_{\mathrm{Page}})=\eexc(0)/2$. Bottom: exact benchmark mass decomposition in units of $M(0)$, with $M(v)=M_{\mathrm{BPS}}(Q_0)+\eexc(v)$, the frozen BPS anchor, and the sourced excitation above it.}
    \label{fig:sourcedbg}
\end{figure}

\section{Emitted-history algebra and pair-polarized Volterra law}
\label{sec:channel}

Projecting the sectorial Albertian evolution onto the Peirce-selected readout gives the reduced history law and the entropy attached to the accumulated channel algebra.
The entropy used below is the entropy of the state restricted to the outgoing ordinary channel selected by the Peirce organization, not of the full exceptional horizon algebra.

The accumulated channel algebra is the algebra of observables that have reached the asymptotic readout by time $v$.
Fix a primitive idempotent $e$ and a slice $u$.
The symbol $\alpha_s$ denotes Heisenberg evolution in the horizon algebra, $\Pi_u$ projects to the ordinary slice selected by $u$, $\pi_Q$ projects from that slice to the radiative block $\Qu$, $\Aone$ is the relay sector that links the hidden complement $\Ku$ to the readout block $\Qu$, and $\Jor\langle\cdots\rangle$ denotes the Jordan algebra generated by the enclosed observables.
With this notation, the minimal accumulated radiation algebra is
\begin{equation}
R_v^{\min}=
\Jor\Big\langle
\pi_Q\Pi_u\alpha_s(X): X\in \Aone,\ 0\le s\le v
\Big\rangle.
\label{eq:Rmin}
\end{equation}
It collects the outgoing observables that have actually reached the ordinary channel by time $v$.

Let $\varpi_v$ denote the evolved horizon state at time $v$.
We define the reduced-channel entropy of this primitive accumulated channel by
\begin{equation}
S_A(v)=S\bigl(\varpi_v\!\upharpoonright R_v^{\min}\bigr).
\label{eq:SAdef}
\end{equation}
This is the entropy later plotted in Fig.~\ref{fig:entropy}: it is the ordinary von Neumann entropy of the state restricted to the associative readout algebra generated inside $\Qu\cong H_2(\C)$, not an entropy of the full non-associative Albert algebra or of the full exterior quantum field theory (QFT).

Let $\mathfrak A_{\mathrm{out}}([0,v])$ denote the Jordan algebra generated by all asymptotic observables that have reached the selected exterior readout sector during the interval $[0,v]$.
On the neutral-source fixed-charge Reissner--Nordstr{\"o}m benchmark, we fix that sector to be the gauge-invariant transverse spin-1 readout carried by the two Maxwell helicities of the ordinary block $\Qu$ \cite{moncrief1974a,moncrief1975,zerilli1974,chandrasekhar1983}.
In this benchmark the readout sector is complete and carries no independent leakage: every gauge-invariant asymptotic observable in that sector is generated by the ordinary channel, while the vanishing of $J_I^{H(5)N}$ together with the fixed-support photonic closure established in Sec.~\ref{sec:state} and Appendix~\ref{app:fixedsupport} excludes an additional gauge-invariant outgoing algebra independent of that channel.
This yields the following benchmark statement.
\begin{proposition}[Outgoing-history algebra of the selected exterior sector]
\label{prop:AoutEqualsRmin}
On the neutral-source fixed-charge Reissner--Nordstr{\"o}m benchmark, the emitted-history algebra of the complete selected exterior readout sector up to time $v$ coincides with the accumulated ordinary channel algebra,
\begin{equation}
\mathfrak A_{\mathrm{out}}([0,v])=R_v^{\min}.
\label{eq:AoutEqualsRmin}
\end{equation}
Consequently,
\begin{equation}
S_A(v)=S\bigl(\varpi_v\!\upharpoonright\mathfrak A_{\mathrm{out}}([0,v])\bigr),
\label{eq:SAout}
\end{equation}
so this reduced-channel entropy is the fine-grained emitted-radiation entropy of that selected exterior readout sector, although not the entropy of the full exceptional horizon system, the full exterior QFT, or a microscopic count of all black-hole states.
Appendix~\ref{app:outgoinghistory} supplies the proof. If additional independent leakage channels are retained, Eq.~\eqref{eq:AoutEqualsRmin} is replaced by the inclusion $R_v^{\min}\subset\mathfrak A_{\mathrm{out}}([0,v])$; equality is the complete selected-readout benchmark.
\end{proposition}
If $v_1\le v_2$, then every generator entering $R_{v_1}^{\min}$ also enters $R_{v_2}^{\min}$, so
\begin{equation}
R_{v_1}^{\min}\subset R_{v_2}^{\min}.
\label{eq:accumulated}
\end{equation}
This inclusion becomes concrete in Sec.~\ref{sec:state}, where the restricted channel state is represented explicitly and its entropy is evaluated.

The reduced law itself is obtained by eliminating the interface sector rather than by postulating a qubit dynamics from the start.
We now pass from the Heisenberg/AQM description of Sec.~\ref{sec:Albertian} to its dual channel-coordinate form.
The Heisenberg law in Eq.~\eqref{eq:heis} evolves observables, while the dual state law in Eq.~\eqref{eq:dualstate} evolves the positive functional $\varpi_v$.

After a Peirce frame has been fixed, we represent the restricted positive functional by Peirce channel coordinates, using the finite trace pairing between the retained Jordan sector and its dual.
Thus $\Psi(v)$ is a state-side density representative of the projected dual evolution; it is not a Heisenberg observable and not a global Hilbert-space wavefunction.
Relative to the refined decomposition
\begin{equation}
J_3(\Oa)=\Atwo\oplus \Aone\oplus \Qu\oplus \Ku,
\label{eq:decomp}
\end{equation}
write this Peirce-coordinate representative as
\begin{equation}
\Psi(v)=\Psi_K(v)\oplus\Psi_1(v)\oplus\Psi_Q(v).
\label{eq:Psidecomp}
\end{equation}
The component $\Psi_K(v)$ represents the hidden-sector part of the restricted functional, $\Psi_1(v)$ represents the relay/interface part, and $\Psi_Q(v)$ represents the ordinary radiative readout part.
The corresponding readout functional is later represented on $\Qu\simeq H_2(\C)$ by the $2\times2$ density matrix $\rho_Q(v)$.

Eliminating $\Psi_1$ in favour of $\Psi_K$ and $\Psi_Q$ by retarded substitution --- the Peirce-adapted analogue of a Nakajima--Zwanzig reduction in open-system language \cite{Nakajima1958,Zwanzig1960,BreuerPetruccione2002} --- gives a Volterra structure \cite{Gripenberg1990}. The full block elimination, displayed in Appendix~\ref{app:hiddenclosure}, also contains self-memory and backflow kernels when $D_{1Q}$ or $D_{KQ}$ are retained. The benchmark studied here is the one-way active-support closure, for which the hidden-to-readout kernel is
\begin{equation}
K_{QK}(v,s)=D_{Q1}(v)\,U_{11}(v,s)\,D_{1K}(s),
\label{eq:kernel}
\end{equation}
with bridge propagator
\begin{equation}
\partial_v U_{11}(v,s)=D_{11}(v)U_{11}(v,s),
\qquad U_{11}(s,s)=1_{\Aone}.
\label{eq:u11eq}
\end{equation}
Here $D_{1K}$ injects the state-side hidden representative from $\Ku$ into the interface, $D_{Q1}$ maps the interface representative into the ordinary radiative block $\Qu$, and $D_{11}$ governs the intrinsic relay propagation through $U_{11}$.
The pair $(v,s)$ therefore carries the memory of the reduced channel problem: the law is Volterra because the outgoing readout representative at time $v$ depends on the earlier hidden and interface history.
Any Markovian or short-memory form is secondary and effective \cite{Feshbach1958,Pazy1983,Gripenberg1990}.
On the neutral-source background, the attractor datum, the channel frame $(e,u)$, and the electric polarization are frozen.
The active bridge therefore collapses to a two-generator support.
Introduce scalar opening and closing amplitudes $a(s)$ and $b(v)$ on that support, a relay frequency $\omega(v)$, and fixed support maps $M_{1K}$ and $M_{Q1}$ that orient the active relay between the hidden block and the radiative block.
With these ingredients the reduced law takes the following pair-polarized form.

\begin{theorem}[Benchmark one-way active-support Volterra closure]
In the source-fixed active-support closure on the neutral-source background, the active bridge closes on the span of $\idtwo$ and a fixed antisymmetric matrix $J$ with $J^2=-\idtwo$.
Equivalently,
\begin{equation}
D_{11}(v)=\omega(v)J,
\label{eq:D11}
\end{equation}
and
\begin{equation}
\begin{aligned}
U_{11}(v,s)&=
\cos\Theta(v,s)\,\idtwo+
\sin\Theta(v,s)\,J,\\
\Theta(v,s)&=\int_s^v\omega(u)\,\dd u.
\end{aligned}
\label{eq:U11}
\end{equation}
If
\begin{equation}
D_{1K}(s)=a(s)M_{1K},
\qquad
D_{Q1}(v)=b(v)M_{Q1},
\label{eq:blockform}
\end{equation}
in the active frame, then
\begin{equation}
K_{QK}(v,s)=a(s)b(v)
\bigl[
\cos\Theta(v,s)\,I+
\sin\Theta(v,s)\,J_{\mathrm{rad}}
\bigr],
\label{eq:pairkernel-general}
\end{equation}
where
\begin{equation}
I:=M_{Q1}M_{1K},
\qquad
J_{\mathrm{rad}}:=M_{Q1}JM_{1K}.
\label{eq:IJdef}
\end{equation}
The induced map $I$ is the identity on the active radiative support, while $J_{\mathrm{rad}}$ is the corresponding fixed rotation acting on that same support.
\end{theorem}

Appendices~\ref{app:projectedblocks}--\ref{app:kernel_source} prove the closure and then make the choice of source-fixed coefficients explicit.
The scalar coefficients have a direct physical meaning in this benchmark.
The amplitudes $a(s)$ and $b(v)$ in Eq.~\eqref{eq:blockform} control how strongly the hidden sector opens into the relay and how strongly the relay closes into the radiative block.
The frequency $\omega(v)$ controls the relay precession on the active support, and the phase
\begin{equation}
\Theta(v,s)=\int_s^v\omega(u)\,\dd u
\label{eq:Theta-again}
\end{equation}
measures the accumulated relay rotation between the emission time $s$ and the readout time $v$.
The fixed maps $M_{1K}$ and $M_{Q1}$ encode the internal orientation of the active relay support; in the neutral-source frame they become the identity, so the relay transports the active two-component support without additional internal mixing.
The remaining active-support identifications now fix the source-fixed representative explicitly.
On the active support one identifies
\begin{equation}
K_{\mathrm{act}}\cong A_{1,\mathrm{act}}\cong Q_{\mathrm{act}}\cong \Cu,
\label{eq:active-support}
\end{equation}
with real basis $\{1,u\}$, and the adapted relay maps are
\begin{equation}
M_{1K}\bigl(\Psi_K(z)\bigr)=\Psi_1(z),
\qquad
M_{Q1}\bigl(\Psi_1(z)\bigr)=\Psi_Q(z).
\label{eq:mapsbg}
\end{equation}
In that basis,
\begin{equation}
M_{1K}=M_{Q1}=\idtwo,
\qquad
J=
\begin{pmatrix}
0 & -1\\[3pt]
1 & 0
\end{pmatrix}.
\label{eq:fixedsupportmatrices}
\end{equation}
Hence $I=\idtwo$ and $J_{\mathrm{rad}}=J$.
Along the neutral-source fixed-charge Reissner--Nordstr{\"o}m branch, the scalar coefficients are fixed as follows. Let
\begin{equation}
E:=\eexc(v),
\qquad
r_+(E)=Q_0+E+\sqrt{E(E+2Q_0)},
\label{eq:rplus}
\end{equation}
and
\begin{equation}
\SBH(E)=\pi r_+(E)^2.
\label{eq:SBH}
\end{equation}
The overlap factorization is derived at WKB level in Appendix~\ref{app:wkb}, while Appendix~\ref{app:bridgeclock} evaluates the corresponding source-fixed Parikh--Wilczek barrier action on the sourced benchmark \cite{ParikhWilczek2000}. The WKB factor is used slice by slice on the adiabatic source-fixed RN trajectory. Within this source-fixed representative, the WKB overlap gives the explicit scalar opening factor used in the Volterra kernel,
\begin{equation}
S_*(E)=\frac12\bigl(\SBH(E)-\SBH(0)\bigr),
\label{eq:Sstar}
\end{equation}
so that
\begin{equation}
\begin{aligned}
a(E)=b(E)
&=e^{-S_*(E)}\\
&=\exp\!\left[-\frac12\bigl(\SBH(E)-\SBH(0)\bigr)\right].
\end{aligned}
\label{eq:abstar}
\end{equation}
The same source-fixed RN entropy fixes the relay frequency used in the active-support kernel:
\begin{equation}
\begin{aligned}
\omega(E)=T_H(E)
&=\left(\frac{\partial\SBH}{\partial E}\right)^{-1}\\
&=\frac{\sqrt{E(E+2Q_0)}}{2\pi r_+(E)^2}.
\end{aligned}
\label{eq:omega}
\end{equation}
This is the controlled benchmark layer. Corrections from $\dot M$, stochastic transport of the Peirce frame, or higher-derivative terms would deform the transfer kernel and define the next benchmark layer; they are not included in the source-fixed representative analyzed here.
Therefore the source-fixed projected blocks are
\begin{align}
D_{1K}(v)&=a\bigl(\eexc(v)\bigr)\,\idtwo,
\label{eq:D1Kexplicit}
\\
D_{Q1}(v)&=a\bigl(\eexc(v)\bigr)\,\idtwo,
\label{eq:DQ1explicit}
\\
D_{11}(v)&=\omega\bigl(\eexc(v)\bigr)\,J.
\label{eq:D11explicit}
\end{align}
These source-fixed expressions are the active-support specialization of the factorized overlaps derived in Appendix~\ref{app:projectedblocks}--\ref{app:wkb}; on that support the identity matrices are the adapted canonical choice recorded in Eq.~\eqref{eq:explicitM} and fixed explicitly in Appendix~\ref{app:actsupport}, with any overall scale absorbed into $a$ and $b$.

Using Eqs.~\eqref{eq:U11}, \eqref{eq:abstar}, and \eqref{eq:D11explicit}, the kernel becomes
\begin{equation}
\begin{aligned}
K_{QK}(v,s)
&=a\bigl(\eexc(s)\bigr)a\bigl(\eexc(v)\bigr)\\
&\quad\times\bigl[\cos\Theta(v,s)\,\idtwo+\sin\Theta(v,s)\,J\bigr].
\end{aligned}
\label{eq:kernelmain}
\end{equation}
This is the explicit kernel used throughout the reduced-state analysis.

\section{Reduced qubit representative, spectral gap, and photonic observables}
\label{sec:state}

The reduced history law of Sec.~\ref{sec:channel} now becomes the state seen by the selected ordinary readout, after restriction to the associative block $\Qu$.
The qubit below is not postulated as a fundamental Hawking qubit. It is the ordinary readout representative obtained after the exceptional complement and the relay have been projected and reduced.
The key outputs are the density-matrix representative $\rho_Q(v)$ of the restricted channel state, its spectral gap $\Delta(v)$, and the reduced-channel entropy $S_A(v)$ plotted in Fig.~\ref{fig:entropy}.
The algebraic state is the restricted positive functional $\varpi_v\!\upharpoonright\Qu$; $\rho_Q(v)$ represents it on $\Qu\simeq H_2(\C)$ by
\begin{equation}
\varpi_v(A_Q)=\Tr[\rho_Q(v)A_Q],
\qquad A_Q\in\Qu .
\label{eq:rhoQrepresentative}
\end{equation}

The ordinary radiative block identified in Sec.~\ref{sec:Peirce} is $\Qu\cong H_2(\C)$, so the selected readout is an ordinary two-level system even though the ambient horizon algebra remains exceptional.
The leading emitted-fraction variable $r(v)$ comes from the sourced clock of Sec.~\ref{sec:background}; $\delta r(v)$ is the memory-induced population shift relative to that leading bookkeeping; $c(v)$ is the accumulated off-diagonal scalar coherence of the reduced channel representative; $x(v),y(v),z(v)$ are the Bloch components of the same representative; and $\Delta(v)$ is the spectral gap that controls the entropy.
The binary arc $H_2(r)$ is useful as a dephased comparator, but it is not the entropy of the reduced channel state unless the coherence and population-shift corrections vanish.

The restricted radiative state is therefore represented as a qubit density matrix on $\Qu\cong H_2(\C)$, first in Bloch form and then in the channel basis.
The two presentations describe the same representative: the Bloch variables $x,y,z$ are convenient for the spectral analysis, while the channel variables $r,\delta r,c$ separate the leading emitted fraction from the memory-induced corrections.
The density-matrix representative is written in Bloch form as
\begin{equation}
\rho_Q(v)=\frac12\Bigl(\idtwo+x(v)\sigma_x+y(v)\sigma_y+z(v)\sigma_z\Bigr),
\label{eq:Bloch}
\end{equation}
and in operational form as
\begin{equation}
\rho_Q(v)=
\begin{pmatrix}
1-r(v)-\delta r(v) & c(v)\\[4pt]
c(v)^* & r(v)+\delta r(v)
\end{pmatrix}.
\label{eq:rhoQ}
\end{equation}
The restricted state represented by $\rho_Q(v)$ is induced by the dual Albertian evolution of Eq.~\eqref{eq:dualstate}, written in the Peirce channel coordinates introduced in Sec.~\ref{sec:channel}, after retarded elimination of the interface.
Appendix~\ref{app:hiddenclosure} derives the explicit formulas from the Volterra kernel and closes the hidden-sector history on the sourced background, so that the projected hidden data below are fixed by the sourced clock together with a finite initial hidden datum rather than by an external history prescription.
Let $\Xi_K(s)$ denote the hidden-sector history variable obtained after eliminating the interface.
It packages the part of the evolving state that remains in the exceptional complement $\Ku$ and therefore is not seen directly by the asymptotic detector, even though it continues to influence the ordinary channel through the relay.
To read that influence on the qubit, we project $\Xi_K(s)$ along the three Pauli directions and along the two active relay directions $\idtwo$ and $J$.
This gives six real functions,
\begin{align}
X_I(s)&:=\Tr\!\left[\sigma_x\,\Xi_K(s)\right],
&X_J(s)&:=\Tr\!\left[\sigma_x\,J\,\Xi_K(s)\right],
\nonumber\\
Y_I(s)&:=\Tr\!\left[\sigma_y\,\Xi_K(s)\right],
&Y_J(s)&:=\Tr\!\left[\sigma_y\,J\,\Xi_K(s)\right],
\nonumber\\
Z_I(s)&:=\Tr\!\left[\sigma_z\,\Xi_K(s)\right],
&Z_J(s)&:=\Tr\!\left[\sigma_z\,J\,\Xi_K(s)\right].
\label{eq:hiddenprojections}
\end{align}
These six functions are projected components of the same hidden history. In the source-fixed benchmark, Appendix~\ref{app:hiddenclosure} reduces them to finite initial hidden data transported by the two source-fixed quadratures and subject to the positivity criterion of Appendix~\ref{app:admissibility}.
The pair $(X_I,X_J)$ controls the hidden contribution to the $x$-quadrature, $(Y_I,Y_J)$ controls the hidden contribution to the $y$-quadrature, and $(Z_I,Z_J)$ controls the hidden contribution to the population imbalance.
On the sourced background, these six functions are not free inputs: Appendix~\ref{app:hiddenclosure} shows that they are generated by the hidden-sector Volterra equation and reduce to explicit source-fixed quadratures.
Then the Bloch components are
\begin{align}
x(v)&=a\bigl(\eexc(v)\bigr)
\int_0^v a\bigl(\eexc(s)\bigr)
\nonumber\\
&\quad\times\Bigl[\cos\Theta(v,s)\,X_I(s)+\sin\Theta(v,s)\,X_J(s)\Bigr]\,\dd s,
\label{eq:xv}
\\
y(v)&=a\bigl(\eexc(v)\bigr)
\int_0^v a\bigl(\eexc(s)\bigr)
\nonumber\\
&\quad\times\Bigl[\cos\Theta(v,s)\,Y_I(s)+\sin\Theta(v,s)\,Y_J(s)\Bigr]\,\dd s,
\label{eq:yv}
\\
z(v)&=1-2r(v)+a\bigl(\eexc(v)\bigr)
\int_0^v a\bigl(\eexc(s)\bigr)
\nonumber\\
&\quad\times\Bigl[\cos\Theta(v,s)\,Z_I(s)+\sin\Theta(v,s)\,Z_J(s)\Bigr]\,\dd s.
\label{eq:zv}
\end{align}
These formulas are the nonperturbative Volterra solution for the reduced radiative state.
Each Bloch component at the readout time $v$ is therefore a memory integral over the full earlier history $0\le s\le v$: the factors $a(\eexc(v))a(\eexc(s))$ weight how strongly the barrier is opened at the emission and readout endpoints, while $\cos\Theta(v,s)$ and $\sin\Theta(v,s)$ encode the accumulated relay interference between those two times. We now translate the Bloch presentation~\eqref{eq:Bloch} into the operational variables $(r,\delta r,c)$ of Eq.~\eqref{eq:rhoQ}. The translation is algebraic: the diagonal of~\eqref{eq:Bloch} fixes $z$ and hence $\delta r$, while the off-diagonal fixes the coherence $c$.
\begin{theorem}[Canonical readout theorem]
For the qubit state \eqref{eq:rhoQ}, the entropically relevant corrected data are determined by
\begin{equation}
\begin{aligned}
\delta r(v)&=\frac{1-2r(v)-z(v)}{2},\\
c(v)&=\frac{x(v)-iy(v)}{2},\\
|c(v)|&=\frac12\sqrt{x(v)^2+y(v)^2}.
\end{aligned}
\label{eq:readout}
\end{equation}
\end{theorem}

Direct comparison of Eqs.~\eqref{eq:Bloch} and \eqref{eq:rhoQ} gives Eqs.~\eqref{eq:readout}. Substituting Eqs.~\eqref{eq:xv}--\eqref{eq:zv} then gives the long-memory formulas below.
Substituting Eqs.~\eqref{eq:xv}--\eqref{eq:zv} into Eq.~\eqref{eq:readout} yields the Volterra formulas \eqref{eq:deltar}, \eqref{eq:cI}, and \eqref{eq:cJ}:
\begin{equation}
\begin{aligned}
\delta r(v)
&=-\frac{a\bigl(\eexc(v)\bigr)}{2}
\int_0^v a\bigl(\eexc(s)\bigr)\,\dd s \\
&\qquad\times\Bigl[\cos\Theta(v,s)\,Z_I(s)
+\sin\Theta(v,s)\,Z_J(s)\Bigr].
\end{aligned}
\label{eq:deltar}
\end{equation}
The coherence splits as
\begin{equation}
c(v)=c_I(v)+c_J(v).
\label{eq:cv}
\end{equation}
\begin{equation}
\begin{aligned}
c_I(v)
&=\frac{a\bigl(\eexc(v)\bigr)}{2}
\int_0^v \dd s\,
 a\bigl(\eexc(s)\bigr)\cos\Theta(v,s) \\
&\qquad\times [X_I(s)-iY_I(s)] .
\end{aligned}
\label{eq:cI}
\end{equation}
\begin{equation}
\begin{aligned}
c_J(v)
&=\frac{a\bigl(\eexc(v)\bigr)}{2}
\int_0^v \dd s\,
 a\bigl(\eexc(s)\bigr)\sin\Theta(v,s) \\
&\qquad\times [X_J(s)-iY_J(s)] .
\end{aligned}
\label{eq:cJ}
\end{equation}
Thus the population correction is sourced by the $Z$-projection of the hidden sector, while the coherence is sourced by the $X$ and $Y$ quadratures of that same hidden history. The spectral gap $\Delta(v)$ now plays the role the Bloch-vector length plays for an ordinary qubit: it controls the eigenvalues of $\rho_Q(v)$ and therefore its von Neumann entropy.
The spectral gap of the reduced qubit is
\begin{equation}
\Delta(v)=\sqrt{\bigl(1-2r(v)-2\delta r(v)\bigr)^2+4|c(v)|^2},
\label{eq:gap}
\end{equation}
It measures how strongly polarized the reduced channel state remains after the hidden-sector memory has been folded into the readout: $\Delta(v)=1$ corresponds to a pure channel state, whereas smaller values indicate stronger mixing and therefore larger channel entropy.
The corresponding eigenvalues are
\begin{equation}
\lambda_{\pm}(v)=\frac{1\pm\Delta(v)}{2}.
\label{eq:lambdas}
\end{equation}
Define the binary entropy
\begin{equation}
H_2(p):=-p\log_2 p-(1-p)\log_2(1-p),
\qquad 0\le p\le1.
\label{eq:H2}
\end{equation}
Then the entropy defined in Sec.~\ref{sec:channel} is
\begin{equation}
S_A(v)=S\bigl(\rho_Q(v)\bigr)
=H_2\!\left(\frac{1-\Delta(v)}{2}\right).
\label{eq:SAchannel}
\end{equation}
In the strictly dephased limit $\delta r=c=0$, Eq.~\eqref{eq:SAchannel} reduces to the familiar binary arc $H_2(r)$. Appendix~\ref{app:admissibility} gives the full necessary-and-sufficient positivity criterion for the reduced channel and proves that the benchmark photonic representative used later lies inside that admissible domain.
The near-coincidence of $S_A(v)$ with the dephased surrogate $H_2(r(v))$ in Fig.~\ref{fig:entropy} has a simple channel interpretation. On the natural transverse photonic representative one has $\delta r(v)=0$, so the diagonal population bookkeeping of the reduced qubit is tied directly to the sourced clock through $r(v)$. The accumulated coherence $c(v)$ then enters the spectral gap only through $|c(v)|^2$, and its effect on the scalar entropy is correspondingly mild over most of the displayed window. This does not make the coherence trivial. It means that the operational imprint of the exceptional horizon sector sits in the accumulated coherence $c(v)$ and in the two-time covariance that realizes it, not in a large instantaneous distortion of the diagonal entropy curve. In this sense $c(v)$ is closer to a multi-time correlation diagnostic than to a correction of the instantaneous occupation number \cite{AnastopoulosSavvidou2020MultiTime}.
Figure~\ref{fig:entropy} therefore isolates the microscopic channel statement: the reduced-state entropy is close to, but not identical with, the dephased qubit surrogate. 
\begin{figure}[t!]
    \centering
    \includegraphics[width=1.00\linewidth]{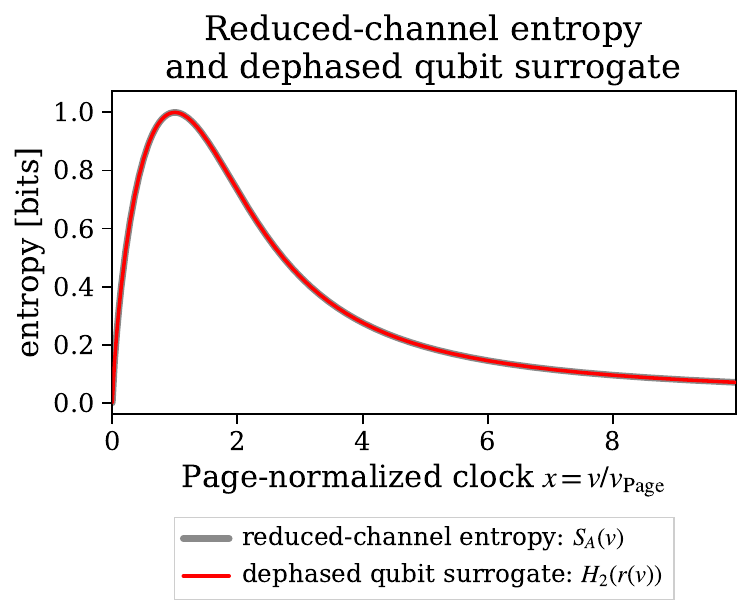}
    \caption{Reduced-channel entropy and dephased qubit surrogate on the neutral-source fixed-charge Reissner--Nordstr{\"o}m benchmark. The solid curve shows $S_A(v)=H_2((1-\Delta(v))/2)$ evaluated on the natural transverse photonic representative of Appendix~\ref{app:fixedsupport}. The red curve shows the dephased surrogate $H_2(r(v))$. Their near-coincidence shows that, on this representative, the scalar entropy is only mildly shifted by coherence over the displayed window even though the emitted-history coherence $c(v)$ remains physically nontrivial.}
    \label{fig:entropy}
\end{figure}

\subsection{Photonic realization of the ordinary channel}

The same two-component support also admits a controlled species-resolved interpretation.
The physical spin-1 sector on adiabatic Reissner--Nordstr{\"o}m slices realizes the ordinary radiative channel $\Qu$ by supplying the transverse helicity support on which the relay kernel acts \cite{moncrief1974a,moncrief1975,zerilli1974,chandrasekhar1983}.
This does not turn the neutral-source background into a pure-photon spacetime.
It identifies the channel support with the physical transverse Maxwell support while keeping the geometry in the effective Einstein--Maxwell-plus-null-fluid class described in Sec.~\ref{sec:background}. The Maxwell problem fixes the physical transverse two-helicity support of the exterior detector. The Peirce relay supplies the memory law carried on that support; the construction does not derive the asymptotic Fock algebra from $J_3(\Oa)$.

Instantaneous polarization and accumulated channel coherence are different observables.
In a parity-even spherically symmetric background the single-mode polarization density matrix of an emitted photon is unpolarized on average, but that is a single-time, single-mode statement.
The channel coherence $c(v)$ entering Eq.~\eqref{eq:rhoQ} is instead the accumulated Volterra-memory quantity attached to the full emitted-history kernel.
Instantaneous mode-resolved polarization data therefore do not compute $c(v)$.
Gauge invariance and transversality fix the physical two-helicity support on which the relay acts; the emitted-history state is then specified by the fixed-support representative recorded in Appendix~\ref{app:fixedsupport}. That representative shows that the natural transverse photonic specialization obeys $\delta r(v)=0$.
The kernel below is the minimal fixed-support representative used in this benchmark. The greybody-filtered occupation fixes the equal-time modulus, the spin-$1$ scattering phase fixes the propagation phase, and the Peirce relay supplies the internal memory rotation. This representative therefore adds no independent diagonal population correction. More general history kernels populate the diagonal correction already present in Eqs.~\eqref{eq:deltar}--\eqref{eq:cJ}, subject to the admissibility criterion of Appendix~\ref{app:admissibility}.
Writing the greybody-filtered Hawking spin-$1$ occupation on each adiabatic RN slice as
\begin{equation}
 n_{\ell\omega}(s)=\frac{\Gamma^{(1)}_{\ell}(\omega;s)}{e^{\omega/T_H(s)}-1},
 \label{eq:nlow-main}
\end{equation}
where \(\Gamma^{(1)}_{\ell}(\omega;s)\) is the spin-$1$ transmission factor and the denominator is the Bose factor at the instantaneous Hawking temperature \(T_H(s)\) \cite{Hawking1975,Page1976I,chandrasekhar1983,hod2016,ngampitipan2013}, the fixed-support history kernel is
\begin{equation}
\begin{aligned}
 g_{\ell\omega m;\lambda\lambda'}(s,s')
 &=\frac12\sqrt{n_{\ell\omega}(s)n_{\ell\omega}(s')}
 e^{i[\delta_{\ell\omega}(s)-\delta_{\ell\omega}(s')]}\\
 &\quad\times
 \left[\cos\Theta(s,s')\,\idtwo+
 \sin\Theta(s,s')\,J\right]_{\lambda\lambda'} .
\end{aligned}
\label{eq:fixed-support-history-kernel-main}
\end{equation}
Here \(\delta_{\ell\omega}(s)\) is the outgoing scattering phase of the same photonic mode, while \(\Theta(s,s')\) is the relay phase already defined in Sec.~\ref{sec:channel}; the difference is that \(\delta_{\ell\omega}\) encodes propagation in the physical spin-$1$ sector, whereas \(\Theta\) encodes the internal channel relay accumulated between the two emission times.  The positive spectral overlap induced by Eq.~\eqref{eq:fixed-support-history-kernel-main} is
\begin{equation}
 W(s,s'):=\sum_{\ell,m}\int_0^\infty d\omega\,
 \sqrt{n_{\ell\omega}(s)n_{\ell\omega}(s')} .
\label{eq:spectral-overlap-main}
\end{equation}
It is the quantitative emitted-history weight used in the photonic forecast. As a one-window scale check, the corresponding photon number flux is obtained by integrating the same greybody-filtered occupation over frequency and angular degeneracy. In the dominant low-frequency electromagnetic channel,
\begin{equation}
\begin{aligned}
 \dot N_\gamma(s)&=\sum_{\ell,m}\int_0^\infty\frac{d\omega}{2\pi}\,
 \frac{\Gamma^{(1)}_\ell(\omega;s)}{e^{\omega/T_H(s)}-1}\\
 &=\frac{36\zeta(5)}{\pi}\,a_1 r_+^4(s)T_H^5(s),
\end{aligned}
\label{eq:photon-number-flux-main}
\end{equation}
Here \(\dot N_\gamma(s)\) is the number of photons emitted per unit time on the adiabatic RN slice, and \(\Gamma_\ell^{(1)}\) is the spin-$1$ greybody transmission factor. The closed form uses the low-frequency law
\(\Gamma_1^{(1)}(\omega;s)=a_1[\omega r_+(s)]^4+O(\omega^6)\),
with the \(\ell=1\) degeneracy included; \(a_1\) fixes the leading greybody normalization. The separable envelope
\(g_\gamma(s)=[\dot N_\gamma(s)/\dot N_\gamma(0)]^{1/2}\)
is useful for scale estimates, while the forecasts in
Figs.~\ref{fig:entropy}--\ref{fig:twotime} use the nonseparable two-time overlap \(W(s,s')\).

After the spectral contraction used for the forecast, the polarimetric cross-helicity covariance is the phase-sensitive observable of this kernel. Up to the common normalization fixed by the selected time windows, its connected part has the structure
\begin{equation}
 C_{+-}(s,s')\propto
 W(s,s')\exp\{2i[\Theta(s)-\Theta(s')]\},
\label{eq:Cpm-main}
\end{equation}
with the scattering phase retained in the mode-resolved form of Eq.~\eqref{eq:fixed-support-history-kernel-main}.  Thermal shot noise can contribute to intensity covariances, but it is not phase-locked to this relay angle.  The falsifiable fixed-support prediction is therefore the joint appearance of the spectral envelope \(W(s,s')\) and the deterministic polarimetric phase \(2[\Theta(s)-\Theta(s')]\).

The accumulated scalar coherence entering the reduced qubit is the prefix-normalized dephasing component.
Equation~\eqref{eq:levelC-c-main} is the fixed-support electromagnetic realization of the same modulus $|c(V)|$ that appears in the Bloch/readout relation of Eq.~\eqref{eq:readout}. It does not introduce an independent coherence variable: $c(V)$ remains the off-diagonal complex scalar coherence of the reduced readout matrix, while the photonic two-time kernel fixes its accumulated magnitude.
\begin{equation}
\begin{aligned}
 |c(V)|&=\frac{1}{4\mathcal N(V)}
 \int_0^V\!ds\int_0^V\!ds'\\
&\quad\times W(s,s')\{1-\cos[2\Delta\Theta(s,s')]\}.
\end{aligned}
\label{eq:levelC-c-main}
\end{equation}
where \(\Delta\Theta(s,s')=\Theta(s)-\Theta(s')\) and
\begin{equation}
 \mathcal N(V)=\int_0^V\!ds\int_0^V\!ds'\,W(s,s').
\label{eq:levelC-normalization-main}
\end{equation}
This fixes the modulus \(|c(V)|\) of the off-diagonal entry used in \(\rho_Q(V)\). Consequently \(\Delta(V)\), \(S_A(V)\), and the coherence-transport rate depend on this accumulated magnitude, while phase-sensitive information is retained in the two-time helicity covariance. The double time integral measures accumulated coherence across the entire emitted history rather than an instantaneous polarization density matrix.
\subsection{Vector-polarization reading and fixed-support photonic channel rates}

The active support determined in Sec.~\ref{sec:channel} carries an immediate vector reading familiar from coherency-matrix and Stokes descriptions of two-component polarization sectors \cite{ParisRehacek2004,Alonso2023PolarizationGeometry,GilOssikovski2022}.
Since Eq.~\eqref{eq:active-support} identifies the active bridge with $\Cu\cong\R^2$, and Eq.~\eqref{eq:fixedsupportmatrices} fixes $J$ as the generator of planar rotations, the reduced qubit representative may also be read as the coherency matrix of an effective two-component radiative mode.
Appendix~\ref{app:blochstokes} records the explicit Bloch--Stokes dictionary, the ladder-operator translation, and the associated polarization observables.
Introduce the ladder operators and number operator on the channel basis,
\begin{equation}
\sigma_+:=|1\rangle\langle 0|,
\qquad
\sigma_-:=|0\rangle\langle 1|,
\qquad
n:=|1\rangle\langle 1|=\sigma_+\sigma_-.
\label{eq:laddermain}
\end{equation}
Then
\begin{equation}
\begin{aligned}
\langle \sigma_-\rangle_v&=c(v),\\
\langle \sigma_+\rangle_v&=c(v)^*,\\
\langle n\rangle_v&=r(v)+\delta r(v).
\end{aligned}
\label{eq:ladderexpectmain}
\end{equation}
Hence $\delta r(v)$ is the memory-induced population shift relative to the leading flux variable $r(v)$, while $|c(v)|$ is the coherence magnitude of the effective two-component mode.
A purely thermal reduced state in the channel basis would be diagonal, so the departures from strict thermality in this setup are carried precisely by the population shift $\delta r(v)$ and the coherence sector $c(v)$.
Appendix~\ref{app:blochstokes} also proves that the effective degree of polarization is the spectral gap itself,
\begin{equation}
P(v)=\Delta(v).
\label{eq:polgap}
\end{equation}
so that the same entropy may equivalently be rewritten as the polarization-entropy identity
\begin{equation}
S_A(v)=H_2\!\left(\frac{1-P(v)}{2}\right).
\label{eq:polentropy}
\end{equation}
This vector and ladder-operator reading is not a second derivation of the species-resolved Maxwell problem.
Its role is narrower: once the photonic realization of $\Qu$ is fixed, it translates the density-matrix representative of the restricted channel state into detector-level observables on the effective two-component support.
The asymptotic readout is therefore organized by the effective occupation $\langle n\rangle_v$, the coherence magnitude $|c(v)|$, and the polarization/spectral-gap variable $P(v)=\Delta(v)$.
The same background data also determine a family of reduced-channel observables that go beyond the entropy curve itself.
Appendix~\ref{app:blochstokes} formulates them in the same source-fixed gauge and places them in the open-system language of reduced channel transport \cite{Nakajima1958,Zwanzig1960,BreuerPetruccione2002}.
From the sourced trajectory one obtains the energy-drain rate
\begin{equation}
\mathcal{R}_E(v):=-\dot{\eexc}(v)=\Gamma_{\mathrm{evap}}(v)\,\eexc(v),
\label{eq:REmain}
\end{equation}
while the source-fixed bridge data determine the relay-opening rate
\begin{equation}
\begin{aligned}
\mathcal{R}_a(v)
&:=\frac{\dd}{\dd v}\log a\bigl(\eexc(v)\bigr)\\
&=\frac{\Gamma_{\mathrm{evap}}(v)\,\eexc(v)}{2\,\omega\bigl(\eexc(v)\bigr)}.
\end{aligned}
\label{eq:Ramain}
\end{equation}
The cumulative memory intensity is
\begin{equation}
\mathcal{M}(v)=a\bigl(\eexc(v)\bigr)\int_0^v a\bigl(\eexc(s)\bigr)\,\dd s,
\label{eq:memorymain}
\end{equation}
The associated rate formulas--effective occupation transfer, coherence transport, spectral-gap flow, and reduced-channel entropy production--are derived in Appendices~\ref{app:ratesclock} and \ref{app:ratesentropy}.

\begin{figure}[t!]
    \centering
    \includegraphics[width=1.00\linewidth]{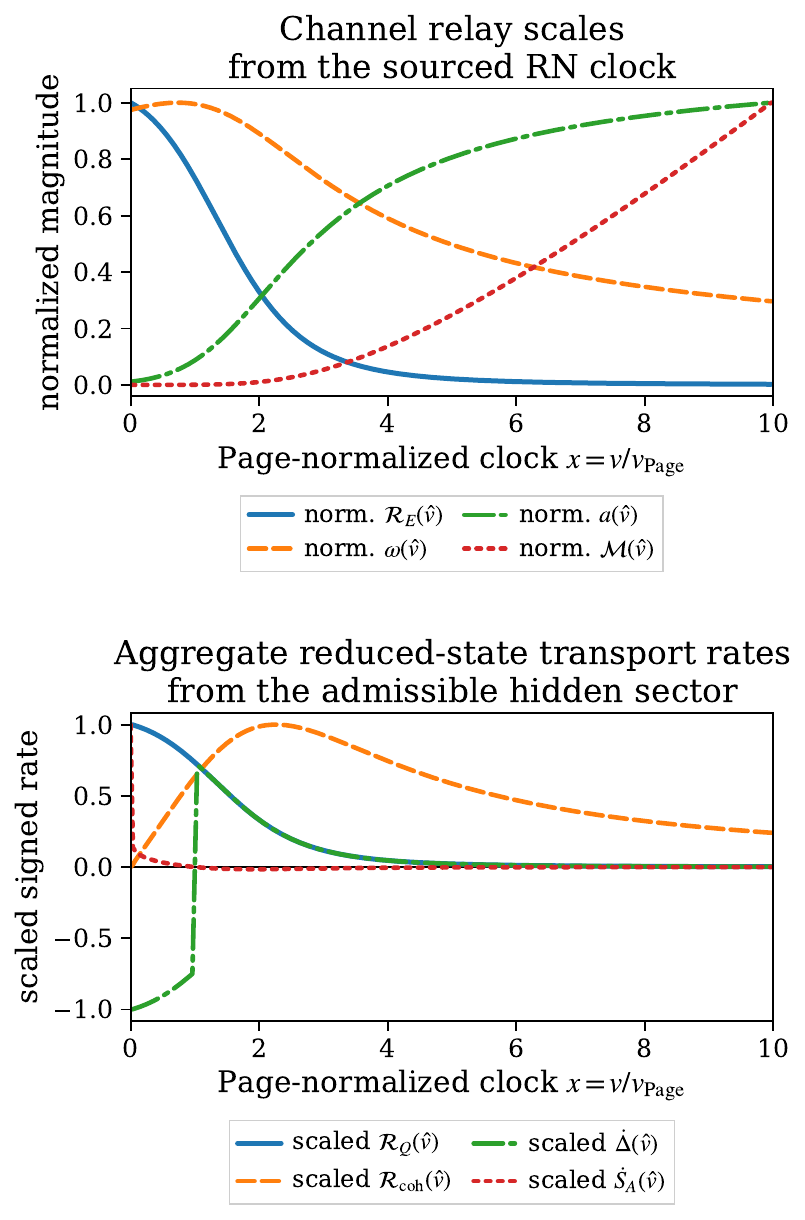}
    \caption{Channel relay scales and reduced-state transport rates on the neutral-source RN background.
Top: display-normalized relay scalars extracted from the same sourced trajectory, namely the energy-drain rate $\mathcal{R}_E(\hat v)=-\dot{\eexc}(\hat v)$, the bridge frequency $\omega(\hat v)$, the relay amplitude $a(\hat v)$, and the cumulative memory variable $\mathcal{M}(\hat v)=a(\hat v)\int_0^{\hat v} a(s)\,\dd s$.
Bottom: display-normalized reduced-state rates evaluated on the natural transverse photonic representative, namely $\mathcal{R}_Q(\hat v)=\frac{\dd}{\dd \hat v}[r(\hat v)+\delta r(\hat v)]$, $\mathcal{R}_{\mathrm{coh}}(\hat v)=\frac{\dd}{\dd \hat v}|c(\hat v)|$, the spectral-gap flow $\dot\Delta(\hat v)$, and the reduced-channel entropy production rate $\dot S_A(\hat v)$. On that representative $\delta r(\hat v)=0$, so $\mathcal{R}_Q(\hat v)=\dot r(\hat v)$. The normalization is for display only.}
    \label{fig:rates}
\end{figure}

The separation in Fig.~\ref{fig:rates} shows that the sourced clock fixes the relay scales, while the fixed-support spectral-overlap kernel controls the coherence-sensitive transport rates. In the natural transverse photonic representative $\delta r(v)=0$, so $\langle n\rangle_v=r(v)$ and $\mathcal R_Q(v)=\dot r(v)$; the nontrivial memory imprint is carried by the emitted-history coherence $c(v)$.

\subsection{Historical coherence and the two-time readout observable}
\label{subsec:historicalcoherence}

The scalar coherence $c(V)$ is the prefix integral of a two-time readout density. The corresponding detector observable is ordinary: it lives in the selected transverse readout block $\Qu$, not in the full retained Albertian sector. Let $\mathcal O_Q^{(\chi)}(v)$ denote a time-windowed quadrature of the two-helicity readout around the retarded time $v$, with analyzer phase $\chi$. In the effective two-state channel it is
\begin{equation}
\mathcal O_Q^{(\chi)}(v)=
 e^{-i\chi}\sigma_-(v)+e^{i\chi}\sigma_+(v),
\label{eq:OQchi}
\end{equation}
and in the transverse photonic realization it is the corresponding helicity coherency observable
\begin{equation}
\begin{aligned}
\mathcal O_Q^{(\chi)}(v)
\longleftrightarrow{}&
 e^{-i\chi}a_+^\dagger(v)a_-(v)\\
&+e^{i\chi}a_-^\dagger(v)a_+(v).
\end{aligned}
\label{eq:OQphoton}
\end{equation}
Here $a_\pm(v)$ are windowed annihilation operators for the two physical transverse helicities. The symbols $\sigma_\pm(v)$ denote the same $Q_u$ ladder/Stokes components used in the one-time Bloch representation of Sec.~\ref{sec:state}, now evaluated as readout observables in two emission windows. They should not be read as new photon creation and annihilation operators; the Maxwell field provides the transverse two-helicity support on which these readout components are realized. Thus $\mathcal O_Q^{(\chi)}(v)$ is the Stokes/coherency quadrature measured by a polarization analyzer in a finite emission window. Its connected two-time readout correlator is
\begin{equation}
\begin{aligned}
G_{\rm conn}^{(\chi)}(v,s)
&=\bigl\langle
\mathcal O_Q^{(\chi)}(v)\mathcal O_Q^{(\chi)}(s)
\bigr\rangle\\
&\quad-
\bigl\langle\mathcal O_Q^{(\chi)}(v)\bigr\rangle
\bigl\langle\mathcal O_Q^{(\chi)}(s)\bigr\rangle .
\end{aligned}
\label{eq:Gconnmain}
\end{equation}
This covariance is the operational two-time observable of the selected readout: it measures the part of the emitted history carried by correlations between two time windows rather than by their separate one-time means. The complex cross-helicity component underlying this real Stokes quadrature is
\begin{equation}
C_{+-}(v,s):=
\bigl\langle\sigma_+(v)\sigma_-(s)\bigr\rangle_{\rm conn}.
\label{eq:CpmGbridge}
\end{equation}
The analyzer phase $\chi$ selects a real quadrature of this complex component. Up to the corresponding one-point subtractions and same-helicity/intensity pieces fixed by the chosen Stokes observable, $G_{\rm conn}^{(\chi)}(v,s)$ contains the projection
\begin{equation}
G_{\rm conn}^{(\chi)}(v,s)\supset
2\,\mathrm{Re}\!\left[e^{2i\chi}C_{+-}(v,s)\right].
\label{eq:GchiCpmProjection}
\end{equation}
Thus $C_{+-}$ carries the phase-sensitive helicity memory, while $G_{\rm conn}^{(\chi)}$ is its directly measured Stokes-quadrature readout. At one time this is consistent with the Bloch--Stokes relation $\mathrm{Tr}[\rho_Q(v)\sigma_-]=c(v)$; at two times it records the emitted-history covariance that is compressed into the accumulated scalar coherence.

This connects the construction with the multi-time Hawking-correlation viewpoint: the one-time spectrum can remain close to its greybody-filtered form, while the emitted history carries information in temporal covariances \cite{AnastopoulosSavvidou2020MultiTime}. Here that covariance is not left abstract; it is the helicity/Stokes correlator selected by $Q_u$ and generated by the Peirce relay. The single-window occupation follows the sourced clock in the natural representative, $\langle n\rangle_v=r(v)$, while the retained-sector imprint appears in the phase-sensitive correlator of Eq.~\eqref{eq:Gconnmain}. This is the correlation-level diagnostic anticipated by the Page/Mathur information-flow lesson \cite{Page1993,Mathur2009SmallCorrections}.

The full Maxwell correlator carries its frequency, angular-momentum, helicity, greybody, and scattering-phase labels.  The fixed-support forecast keeps their positive spectral contraction through the overlap \(W(v,s)\) of Eq.~\eqref{eq:spectral-overlap-main}.  Define
\begin{equation}
\Delta\Theta(v,s):=\Theta(v)-\Theta(s).
\label{eq:DeltaThetaMain}
\end{equation}
The two-time readout separates into the spectral envelope, the phase-locked polarimetric covariance, and the connected dephasing density,
\begin{equation}
\begin{aligned}
D_0(v,s)&=W(v,s),\\
D_\phi(v,s)&=W(v,s)\cos[2\Delta\Theta(v,s)],\\
D_{\rm conn}(v,s)&=W(v,s)\{1-\cos[2\Delta\Theta(v,s)]\}.
\end{aligned}
\label{eq:twotimedensities}
\end{equation}
The envelope \(D_0\) gives the spectral pair weight of two emission windows. The phase density \(D_\phi\) records the relay-induced polarimetric phase alignment of those windows. The connected density \(D_{\rm conn}=D_0-D_\phi\) records the positive dephasing part associated with separated windows; it vanishes on the equal-time diagonal and grows where the relay phase distinguishes two portions of the emitted history. With the spectral-overlap normalization \(\mathcal N(V)\) of Eq.~\eqref{eq:levelC-normalization-main}, the connected density fixes the real positive magnitude of the connected contribution,
\begin{equation}
\begin{aligned}
|c_{\rm conn}(V)|&=
 \frac{1}{4\mathcal N(V)}
 \int_0^V\!\dd s\int_0^V\!\dd s'\,
 D_{\rm conn}(s,s').
\end{aligned}
\label{eq:cconnmain}
\end{equation}
In the helicity basis used for the fixed-support representative, this contribution is inserted with the conventional phase of that basis. The scalar readout quantities depend on \(|c(V)|\) and on the spectral gap. Appendix~\ref{app:fixedsupport} gives the fixed-support emitted-history kernel from which \(W(v,s)\) and the polarimetric forecast are obtained.

\begin{figure*}[t!]
    \centering
    \includegraphics[width=0.98\textwidth]{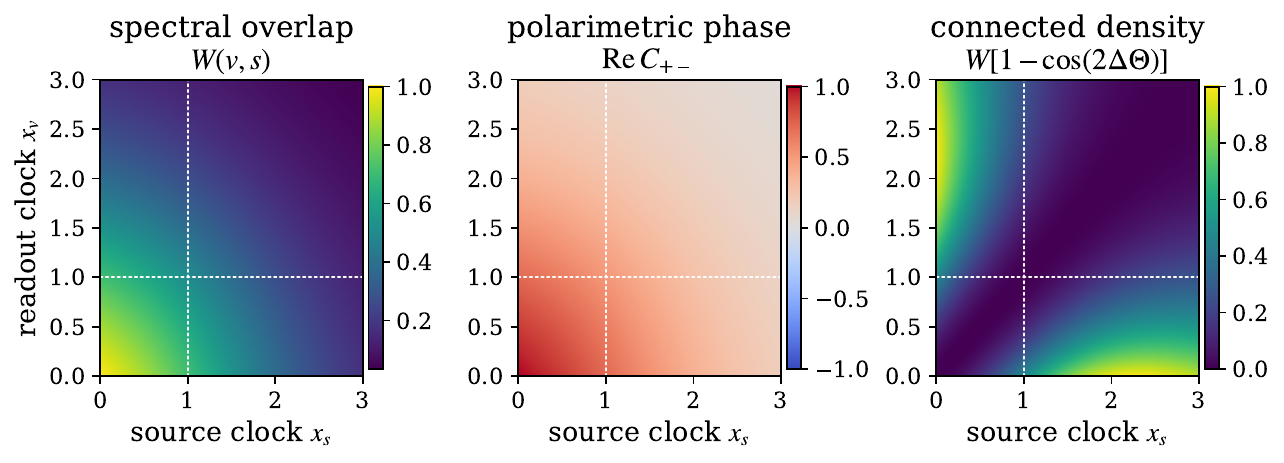}
    \caption{Fixed-support photonic two-time forecast on the sourced RN clock. Left: the spectral overlap $W(v,s)$ induced by the mode-resolved history kernel. Middle: the phase-locked polarimetric covariance proportional to $\mathrm{Re}\,C_{+-}(v,s)=W(v,s)\cos[2(\Theta(v)-\Theta(s))]$. Right: the connected density $D_{\rm conn}(v,s)=W(v,s)\{1-\cos[2(\Theta(v)-\Theta(s))]\}$, whose prefix double integral gives the connected contribution to the accumulated coherence magnitude $|c(V)|$. Dotted lines mark the Page time in both source and readout directions. The figure displays the Lorentzian memory carried by correlations between separated emission windows of the selected transverse readout.}
    \label{fig:twotime}
\end{figure*}

Figure~\ref{fig:twotime} unfolds the scalar coherence into the two-time plane. The left panel shows the spectrally integrated support of the emitted history; the middle panel shows the phase-locked polarimetric covariance; the right panel shows the connected density whose prefix integral contributes to the off-diagonal readout coherence. The selected transverse channel therefore carries the retained-sector imprint as historical coherence: the leading occupation is tied to the sourced clock, and the phase-sensitive two-time sector records the relay memory between separated emission windows.

\section{Euclidean transfer and superstatistical branch canonization of the Page envelope}
\label{sec:superstatistics}

Up to this point the reduced Volterra kernel has been used in real time \cite{Gripenberg1990}. That Lorentzian reading gives the exterior detector state, the accumulated photonic coherence, the transport observables, and the two-time density displayed in Fig.~\ref{fig:twotime}. The Peirce--Volterra kernel also has a Euclidean branch reading. Once the RN clock writes the retarded memory law as an ordered energy transfer, Euclidean continuation turns the relay phase into a branch-transfer weight. Equivalently, one diagonalizes the active two-dimensional relay support into its circular components and keeps the decaying Euclidean transfer component on each branch. The full Lorentzian rotation is not itself a positive measure; positivity belongs to the branchwise Euclidean transfer representation. Summing this transfer over an endpoint-selected branch-history domain canonizes that branch measure and produces a positive superstatistical Laplace form \cite{BeckCohen2003,Beck2007,Bernstein1929,Widder1941}. The construction is therefore branchwise from the outset. The regular-opening endpoint and the near-extremal endpoint do not define two terms of one global positive Laplace transform; they define two positive transfer representations on two physical branch-history domains. We now derive the two endpoint branches used in Fig.~\ref{fig:samebgbranch}: regular opening as $r\to0$, which gives the Tsallis/Lomax scale ensemble, and near-extremal Reissner--Nordstr{\"o}m residence as $r\to1$, which gives the shifted-L{\'e}vy residence ensemble.

The retarded elimination of the relay, derived in Appendix~\ref{app:hiddenclosure}, gives the source-fixed active kernel
\begin{equation}
K_{\rm Lor}(v,s)=a(E_v)a(E_s)\exp[J\Theta(v,s)],
\label{eq:KLorVII}
\end{equation}
where $E_v=\eexc(v)$, $a(E)=\exp[-S_*(E)]$, and $J^2=-\idtwo$ on the active two-dimensional support.  The WKB barrier action $S_*(E)$ and the source-fixed clock are derived in Appendices~\ref{app:wkb}--\ref{app:sourcedclockdetail}, while the explicit source-fixed kernel is collected in Appendix~\ref{app:kernel_source}.  Introduce the monotone sourced energy clock
\begin{equation}
E=\eexc(v),\qquad \Gamma(E)=-\frac{\dd E}{\dd v}>0 .
\label{eq:GammaVII}
\end{equation}
A Volterra segment from an earlier point $s$ to a later point $v$ becomes an ordered energy-history segment
\begin{equation}
h=(E,E'),\qquad 0\le E\le E'\le E_0,
\end{equation}
with
\begin{equation}
\begin{aligned}
\dd s&=-\frac{\dd E'}{\Gamma(E')},\\
\Theta(v,s)&=\int_E^{E'}\frac{\omega(\xi)}{\Gamma(\xi)}\,\dd\xi .
\end{aligned}
\label{eq:ThetaEnergyVII}
\end{equation}
Thus the sourced Reissner--Nordstr{\"o}m clock turns the retarded Volterra memory into an energy-ordered transfer problem.

The active relay has the two circular components selected by the ordinary two-helicity readout.  Since $J^2=-\idtwo$,
\begin{equation}
\begin{aligned}
\exp[J\Theta]&=e^{+i\Theta}P_+ + e^{-i\Theta}P_-,\\
P_\pm&=\frac12(\idtwo\mp iJ).
\end{aligned}
\end{equation}
Euclidean continuation of the stable transfer component gives the positive relay weight. The growing Euclidean component is excluded by the bounded-transfer/admissibility prescription on each branch; after projection to the positive branch weight, the two circular components contribute the same scalar factor.
\begin{equation}
\begin{aligned}
e^{\pm i\Theta}&\longmapsto e^{-\tau(E,E')},\\
\tau(E,E')&=\int_E^{E'}\frac{\omega(\xi)}{\Gamma(\xi)}\,\dd\xi\ge0 .
\end{aligned}
\label{eq:tauVII}
\end{equation}
The scalar Euclidean transfer kernel is therefore
\begin{equation}
K_E(E,E')=
\frac{
\exp[-S_*(E)-S_*(E')-\tau(E,E')]
}{\Gamma(E')}.
\label{eq:KEVII}
\end{equation}
Equivalently,
\begin{equation}
\begin{aligned}
K_E(E,E')\,\dd E'
&=W(E)W(E')e^{-\tau(E,E')}\,
\dd\nu_{\rm clock}(E'),\\
W(E)&=e^{-S_*(E)},\qquad
\dd\nu_{\rm clock}(E')=\frac{\dd E'}{\Gamma(E')} .
\end{aligned}
\label{eq:KEfactorVII}
\end{equation}
Equation~\eqref{eq:KEfactorVII} is the physical content of the Euclidean transfer: barrier weight at the two endpoints, Euclidean relay time between them, and the measure induced by the RN clock.  The transfer kernel is the analogue of a Euclidean matrix element.  It becomes a branch partition weight after being summed over a branch domain of histories.

Let $b$ denote one of the two endpoint branches.  A branch history $h=(E,E')\in H_b$ carries the positive barrier-clock measure $W(E)W(E')\dd\nu_{\rm clock}(E,E')$.  The branch scale map
\begin{equation}
\lambda_b:H_b\longrightarrow\mathbb R_+
\end{equation}
assigns to each history its effective intensive transfer scale.  The variable $\lambda_b$ is conjugate to the macroscopic branch coordinate $y_b$, just as inverse temperature is conjugate to energy in the canonical ensemble.  For the opening endpoint $y_H$ is the emitted fraction $r$, expressed below on the branch clock used in Fig.~\ref{fig:samebgbranch}; for the late endpoint $y_I$ is the distance from extremality $1-r$, again expressed on its branch clock.  The branch transfer function is
\begin{equation}
\begin{aligned}
G_b(y_b)=\int_{H_b}& W(E)W(E')
\exp[-\lambda_b(E,E')y_b] \\
&\times\dd\nu_{\rm clock}(E,E').
\end{aligned}
\label{eq:GbHistVII}
\end{equation}
Pushing the positive history measure forward by $\lambda_b$ gives
\begin{equation}
\dd\mu_b(\lambda)=(\lambda_b)_*
\left[W(E)W(E')\dd\nu_{\rm clock}(E,E')\right],
\label{eq:pushVII}
\end{equation}
so that
\begin{equation}
G_b(y_b)=\int_0^\infty e^{-\lambda y_b}\,\dd\mu_b(\lambda).
\label{eq:superstatVII}
\end{equation}
This is the canonical branch sum.  The measure $\dd\mu_b$ counts Euclidean transfer histories by their effective scale $\lambda$, while $e^{-\lambda y_b}$ is the ordinary canonical weight conjugate to the branch coordinate.  In this precise sense, the superstatistical representation used here is a positive Laplace superposition of local exponential transfer weights \cite{BeckCohen2003,Bernstein1929,Widder1941}.  If locally $\dd\mu_b(\lambda)=e^{\Sigma_b(\lambda)}\dd\lambda$, then the branch free action
\begin{equation}
A_b(y_b):=-\log G_b(y_b)
\end{equation}
has the saddle form
\begin{equation}
A_b(y_b)\simeq
\min_\lambda\{\lambda y_b-\Sigma_b(\lambda)\}.
\label{eq:branchfreeVII}
\end{equation}
The branch measure also carries a replica/Mellin reading. Since $e^{-\lambda y_b}$ is the positive local transfer weight, replica moments are its multiplicative moments: the shared-scale moment is $G_b(ny_b)$, whereas statistically independent replicas give $G_b(y_b)^n$. The resulting non-factorizing ratio is the channel analogue of the connected replica saddle in Island calculations. Appendix~\ref{app:replica_mellin} gives the corresponding Mellin construction and connected branch action.

Thus the Laplace representation canonizes the scale-resolved entropy of Euclidean branch histories.  The subscript $b$ is essential. The opening and near-extremal measures live on different endpoint-selected history domains. They are not restrictions of one positive global measure, and the branch exchange below is not a soft global partition sum. Appendix~\ref{app:branch_canonization} records the pushforward construction and the canonical-action statement, while Appendix~\ref{app:noglobal} proves that a single positive global Laplace reconstruction is generically incompatible with a genuine analytic branch switch.

We now evaluate the two endpoint limits of the same transfer kernel.  Define
\begin{equation}
\begin{gathered}
r=1-\frac{E}{E_0},\qquad E=E_0(1-r),\\
E'=E_0(1-u),\qquad 0\le u\le r .
\end{gathered}
\label{eq:rEndpointDefsVII}
\end{equation}
At the opening endpoint $r\to0$, the sourced clock and barrier are regular:
\begin{equation}
\begin{aligned}
S_*(E)&=S_*(E_0)+O(r),\\
\Gamma(E)&=\Gamma_0+O(r),\\
\tau(E,E')&=O(r).
\end{aligned}
\label{eq:openingRegularVII}
\end{equation}
Consequently
\begin{equation}
K_E(r,u)=C_0[1+O(r)].
\label{eq:KEregularVII}
\end{equation}
For a regular branch preparation $X_H(u)=X_0+O(u)$, the opening weight satisfies
\begin{equation}
\begin{aligned}
Z_H(r)&\propto\int_0^rK_E(r,u)X_H(u)\,\dd u\\
&=C_0X_0r+O(r^2).
\end{aligned}
\label{eq:ZHregularVII}
\end{equation}
Here $Z_H(r)$ is the opening-domain weight. The endpoint action is the binary resolution entropy of this opening event, so $Z_H(r)=\kappa r+O(r^2)$ gives the universal onset after branch normalization:
\begin{equation}
A_H(r)=-r\log_2 r+O(r).
\label{eq:AHonsetVII}
\end{equation}
A positive canonical scale measure producing this onset has high-scale tail $w_H(\lambda)\sim\lambda^{-2}$.  The canonical one-scale representative is
\begin{equation}
w_H^{\rm can}(\lambda)=
\frac{\lambda_c}{(\lambda+\lambda_c)^2},
\qquad \lambda_c=\frac{1}{\ln2}.
\label{eq:LomaxVII}
\end{equation}
Thus the regular-opening endpoint canonizes to the Lomax/Pareto-II scale ensemble, equivalently the one-sided Tsallis/Lomax branch \cite{Tsallis1988,Tsallis2009,Shalizi2007,Lomax1954}.  Appendix~\ref{app:opening_lomax} gives the endpoint expansion and the scale-density calculation.

For the display comparator in Fig.~\ref{fig:samebgbranch}, the Bekenstein branch is the normalized source-fixed RN entropy difference
\begin{equation}
\begin{aligned}
s_B(x)&=\frac{\Delta S_{\rm BH}^{\rm RN}(E(x))}
{\Delta S_{\rm BH}^{\rm RN}(E(0))},\\
\Delta S_{\rm BH}^{\rm RN}(E)&:=S_{\rm BH}^{\rm RN}(E)-S_{\rm BH}^{\rm RN}(0).
\end{aligned}
\end{equation}
the Hawking branch is $s_H(x)=1-s_B(x)$, and the Page comparator is $s_{\rm Page}(x)=\min\{s_H(x),s_B(x)\}$. These comparator branches fix the common display target; the endpoint classes themselves are the Tsallis/Lomax and shifted-L{\'e}vy classes derived in this section.

\begin{figure}[!t]
    \centering
    \includegraphics[width=0.98\columnwidth]{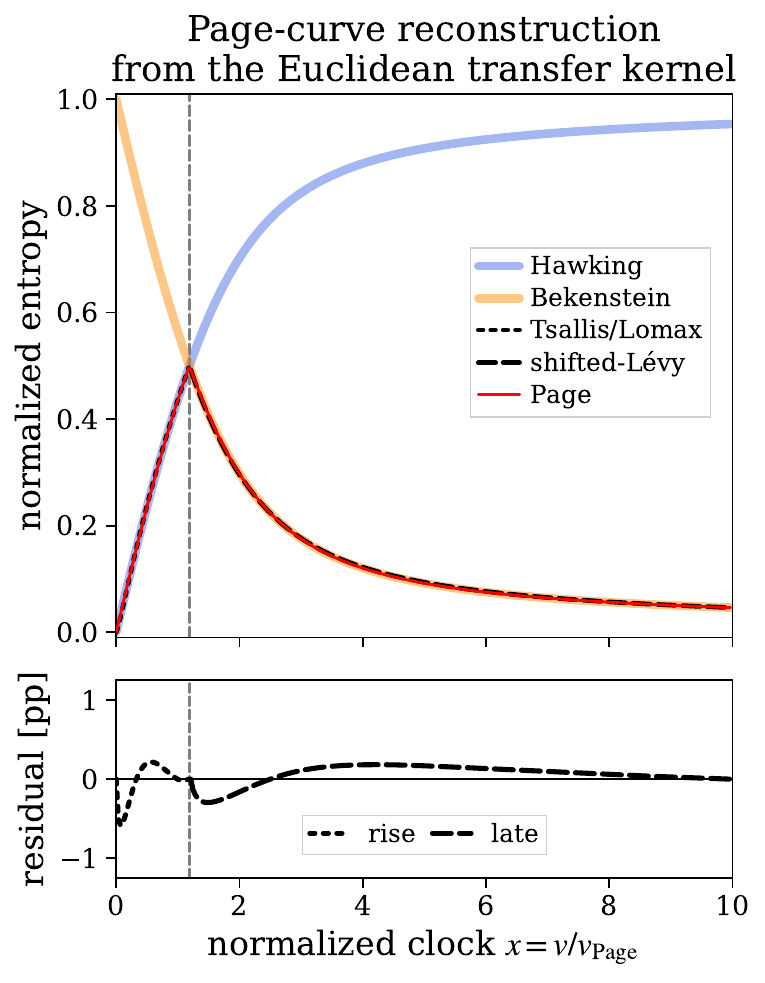}
    \caption{Source-fixed branch-envelope reconstruction from the Euclidean endpoint actions. The upper panel shows the Hawking and Bekenstein branch actions, their branch-admissible Page envelope, and the two endpoint-canonized branches: Tsallis/Lomax for regular opening and shifted-L{\'e}vy for near-extremal residence. The lower panel shows the branch residuals in percentage points on the common Page-normalized clock. The envelope is a hard branch selection, not a global soft two-branch partition sum; Appendix~\ref{app:noglobal} gives the no-global-positive-transform argument.}
    \label{fig:samebgbranch}
\end{figure}

The late endpoint is controlled by the near-extremal Reissner--Nordstr{\"o}m barrier.  Set
\begin{equation}
\epsilon=1-r=\frac{E}{E_0}\to0.
\end{equation}
The source-fixed action has the expansion
\begin{equation}
S_*(E)=\mu E+\nu\sqrt E+O(E^{3/2}).
\label{eq:SstarLateVII}
\end{equation}
Hence the late branch weight has the representative form
\begin{equation}
\begin{aligned}
G_I(r)&=\exp[-\mu_e(1-r)-\nu_e\sqrt{1-r}],\\
\mu_e&=\mu E_0,
\qquad \nu_e=\nu\sqrt{E_0}.
\end{aligned}
\label{eq:GIlateVII}
\end{equation}
It admits the shifted one-sided L{\'e}vy representation
\begin{equation}
G_I(r)=\int_{\mu_e}^{\infty}
e^{-\beta(1-r)}\,\dd\mu_I(\beta),
\label{eq:GILevyVII}
\end{equation}
with density
\begin{equation}
\begin{aligned}
\frac{\dd\mu_I}{\dd\beta}
&=\frac{\nu_e}{2\sqrt\pi}
(\beta-\mu_e)^{-3/2} \\
&\quad\times
\exp\!\left[-\frac{\nu_e^2}{4(\beta-\mu_e)}\right].
\end{aligned}
\label{eq:LevyDensityVII}
\end{equation}
The square-root term in Eq.~\eqref{eq:SstarLateVII} is the near-extremal residence signature.  Its positive scale measure is the shifted-L{\'e}vy residence ensemble, the standard canonical form associated with broad residence and first-passage statistics \cite{Feller1971,Redner2001}.  Appendix~\ref{app:late_levy} derives the expansion of $S_*(E)$ and the shifted-L{\'e}vy identity.

The endpoint families are fixed; only the monotone branch-clock placements are calibrated to the common RN/Page display clock. The endpoint actions are now placed on $x=v/v_{\rm Page}$ through their branch clocks $y_H(x)$ and $y_I(x)$:
\begin{equation}
A_H(x):=A_H(y_H(x)),\qquad
A_I(x):=A_I(y_I(x)).
\label{eq:branchclockactionsVII}
\end{equation}
The Page reconstruction is the branch-admissible lower envelope
\begin{equation}
A_{\rm br}(x)=\min\{A_H(x),A_I(x)\}.
\label{eq:hardEnvelopeVII}
\end{equation}
Equivalently, on the two physical branch domains,
\begin{equation}
A_{\rm br}(x)=
\begin{cases}
A_H(x), & x\in I_H,\\
A_I(x), & x\in I_I,
\end{cases}
\label{eq:hardEnvelopePiecewiseVII}
\end{equation}
with the exchange point fixed by the common RN clock normalization. The hard-envelope equation is not an unannounced approximation to a soft partition sum. A soft expression $-\log(e^{-A_H}+e^{-A_I})$ would require an additional global two-branch ensemble in which both endpoint measures coexist on one common history space. That ensemble is not part of the source-fixed branch construction. Appendix~\ref{app:noglobal} proves that such a single positive global reconstruction is generically incompatible with a genuine analytic branch switch.

Figure~\ref{fig:samebgbranch} displays the two endpoint waiting-time regimes selected by the Euclidean transfer kernel on the common source-fixed RN display clock. The endpoint classes are fixed: the regular-opening endpoint canonizes to the Tsallis/Lomax class with $q=2$, while the near-extremal residence endpoint canonizes to the shifted-L{\'e}vy class with square-root exponent $\gamma=1/2$. The common-clock placement fixes only the branch clocks $y_H(x)$ and $y_I(x)$ used to display both endpoint actions on $x=v/v_{\rm Page}$. It changes the parametrization of the endpoint variables, not the endpoint scale ensembles. The Page-curve envelope is the lower admissible branch envelope selected by these two branch actions on the common RN clock, in analogy with saddle/envelope dominance but without invoking an Island prescription \cite{Page1993,Almheiri2020,Penington2020,AlmheiriRMP2021}.

\section{Conclusion}
\label{sec:conclusion}

The question posed in the Introduction was whether the globally associative Hilbert-factor stage used by AMPS is the correct observable stage for the relevant horizon degrees of freedom. In the retained octonionic-magical sector studied here, the answer is sectorial. The retained horizon-side observable algebra is Albertian, while the exterior detector remains associative because the Peirce/readout projection selects the two-helicity block $Q_u\simeq H_2(\C)$. Local two-generated subalgebras remain special, as guaranteed by the Artin--Cohn/Jacobson--Paige two-generator specialness results \cite{Cohn1954,JacobsonPaige1957}. The infalling description is therefore governed by the usual local effective reasoning on those local algebras, while the full retained horizon algebra is exceptional. The AMPS tripartite tensor-product stage is replaced at the retained-horizon level by this Albertian observable algebra. The compositional obstruction is a statement about the retained Jordan-AQM sector, not a theorem about the full operator algebra of quantum fields in the spacetime \cite{AMPS2013,Harlow2016,HancheOlsen1983,BarnumGraydonWilce2020}.

This reframes several lessons from the information-problem literature. From AMPS it keeps the centrality of the factorization question: the conflict is sharpened only after one assumes an ordinary global tensor-product stage \cite{AMPS2013,Harlow2016}. From Page and the Island developments it keeps the central role of the fine-grained radiation history and of saddle/envelope dominance \cite{Page1993,Almheiri2020,Penington2020,AlmheiriRMP2021}. From Mathur's theorem it keeps the lesson that purification is carried by a structured mechanism in the radiation history, not by arbitrary small local changes to individual Hawking pairs \cite{Mathur2009SmallCorrections}. The present construction implements these lessons inside a selected RN readout: the exterior detector is ordinary, while the retained horizon algebra that feeds it is Albertian.

In Lorentzian time, the closest operational analogue is the multi-time viewpoint on Hawking radiation. That literature identifies history-sensitive temporal correlators as the natural diagnostic when one-time spectra remain close to thermal \cite{AnastopoulosSavvidou2020MultiTime}. Here that diagnostic is realized in a specific channel. The observable is the connected covariance $G_{\rm conn}^{(\chi)}(v,s)$ of windowed helicity/Stokes observables in the selected two-helicity readout. The fixed-support kernel supplies a spectral overlap $W(v,s)$ and a polarimetric phase locked to $2[\Theta(v)-\Theta(s)]$; its connected density prefix-integrates to the scalar coherence $c(V)$ entering the reduced qubit. Thus the paper supplies an Albertian-channel origin for the multi-time diagnostic: the hidden Peirce complement and interface relay determine the long-range temporal structure seen by the ordinary exterior readout.

The same Peirce--Volterra memory has a Euclidean transfer reading \cite{Gripenberg1990}. Its branchwise Euclidean continuation gives a positive branch-transfer kernel. Its regular-opening endpoint gives the Tsallis/Lomax branch, the waiting-time class associated with activation of the readout channel, while its near-extremal endpoint gives the shifted-L{\'e}vy residence branch, the waiting-time class associated with long residence and escape near extremality \cite{BeckCohen2003,Feller1971,Redner2001}. These positive endpoint measures live on distinct branch-history domains. When placed on the common Page-normalized RN clock, their lower admissible envelope reconstructs the Page-curve envelope of the selected source-fixed readout, Eqs.~\eqref{eq:hardEnvelopeVII}.

The Euclidean branch measures also have a replica/Mellin reading. The positive branch measure that produces the superstatistical transfer function also carries a canonical Mellin tower: the shared-scale replica moment is $G_b(ny_b)$, whereas statistically independent replicas give $G_b(y_b)^n$. Their ratio defines the connected branch action $A_b(ny_b)-nA_b(y_b)$. Thus the Peirce--Volterra branch measure supplies the channel counterpart of the replica non-factorization: the replicated histories are correlated by a common transfer scale, and their dominant effective actions organize the same saddle/envelope structure used in Island calculations \cite{Almheiri2020,Penington2020,AlmheiriRMP2021}.

The no-global-positive-transform statement follows from the analyticity properties of positive Laplace transforms \cite{Bernstein1929,Widder1941} and is recorded in Appendix~\ref{app:noglobal}. The comparison is therefore a channel-level reconstruction of the Page envelope for the selected source-fixed readout, with the saddle/envelope logic made familiar by Island calculations but without identifying the mechanism with an Island prescription \cite{Almheiri2020,Penington2020,AlmheiriRMP2021}.

The open-system aspect also has a definite role. The relay elimination is a retarded memory reduction of the Nakajima--Zwanzig type \cite{Nakajima1958,Zwanzig1960,BreuerPetruccione2002}, but the eliminated variables are not a generic bath: they are the hidden Peirce complement and interface relay of the retained Albertian horizon algebra. In real time this reduction produces the measurable two-time readout covariance. In Euclidean time it produces the two endpoint waiting-time ensembles whose admissible lower envelope reconstructs the Page-curve envelope for the selected source-fixed readout. The retained horizon algebra is Albertian; the exterior detector is ordinary; the bridge between them is memory.
\begin{acknowledgments}

R.B.F. acknowledges partial financial support from the Coordena\c{c}\~ao de Aperfei\c{c}oamento de Pessoal de N\'ivel Superior (CAPES), Brazil, Finance Code 001.
\end{acknowledgments}

\bibliographystyle{apsrev4-2}
\bibliography{bibliography}

@article{Hawking1975,
  author  = {Hawking, S. W.},
  title   = {Particle Creation by Black Holes},
  journal = {Communications in Mathematical Physics},
  volume  = {43},
  number  = {3},
  pages   = {199--220},
  year    = {1975},
  doi     = {10.1007/BF02345020}
}

@article{Page1993,
  author  = {Page, Don N.},
  title   = {Information in Black Hole Radiation},
  journal = {Physical Review Letters},
  volume  = {71},
  number  = {23},
  pages   = {3743--3746},
  year    = {1993},
  doi     = {10.1103/PhysRevLett.71.3743}
}

@article{AMPS2013,
  author  = {Almheiri, Ahmed and Marolf, Donald and Polchinski, Joseph and Sully, James},
  title   = {Black Holes: Complementarity or Firewalls?},
  journal = {JHEP},
  volume  = {2013},
  number  = {2},
  pages   = {062},
  year    = {2013},
  doi     = {10.1007/JHEP02(2013)062}
}

@article{Harlow2016,
  author  = {Harlow, Daniel},
  title   = {Jerusalem Lectures on Black Holes and Quantum Information},
  journal = {Reviews of Modern Physics},
  volume  = {88},
  number  = {1},
  pages   = {015002},
  year    = {2016},
  doi     = {10.1103/RevModPhys.88.015002}
}

@article{Jordan1934,
  author  = {Jordan, Pascual and von Neumann, John and Wigner, Eugene P.},
  title   = {On an Algebraic Generalization of the Quantum Mechanical Formalism},
  journal = {Annals of Mathematics},
  volume  = {35},
  number  = {1},
  pages   = {29--64},
  year    = {1934},
  doi     = {10.2307/1968117}
}

@book{McCrimmon2004,
  author    = {McCrimmon, Kevin},
  title     = {A Taste of Jordan Algebras},
  publisher = {Springer},
  address   = {New York},
  year      = {2004},
  doi       = {10.1007/b97489}
}

@article{Cohn1954,
  author  = {Cohn, P. M.},
  title   = {On Homomorphic Images of Special Jordan Algebras},
  journal = {Canadian Journal of Mathematics},
  volume  = {6},
  pages   = {253--264},
  year    = {1954},
  doi     = {10.4153/CJM-1954-026-9}
}

@article{JacobsonPaige1957,
  author  = {Jacobson, Nathan and Paige, Lowell J.},
  title   = {On Jordan Algebras with Two Generators},
  journal = {Journal of Mathematics and Mechanics},
  volume  = {6},
  number  = {6},
  pages   = {895--906},
  year    = {1957}
}

@article{HancheOlsen1983,
  author  = {Hanche-Olsen, Harald},
  title   = {On the Structure and Tensor Products of {JC}-Algebras},
  journal = {Canadian Journal of Mathematics},
  volume  = {35},
  number  = {6},
  pages   = {1059--1074},
  year    = {1983},
  doi     = {10.4153/CJM-1983-059-8}
}

@article{BarnumGraydonWilce2020,
  author  = {Barnum, Howard and Graydon, Matthew A. and Wilce, Alexander},
  title   = {Composites and Categories of Euclidean Jordan Algebras},
  journal = {Quantum},
  volume  = {4},
  pages   = {359},
  year    = {2020},
  doi     = {10.22331/q-2020-11-08-359}
}

@article{GunaydinSierraTownsend1983,
  author  = {G{\"u}naydin, Murat and Sierra, G. and Townsend, P. K.},
  title   = {Exceptional Supergravity Theories and the Magic Square},
  journal = {Physics Letters B},
  volume  = {133},
  number  = {1-2},
  pages   = {72--76},
  year    = {1983},
  doi     = {10.1016/0370-2693(83)90108-9}
}

@article{GunaydinSierraTownsend1984,
  author  = {G{\"u}naydin, Murat and Sierra, G. and Townsend, P. K.},
  title   = {The Geometry of {$N=2$} Maxwell--Einstein Supergravity and Jordan Algebras},
  journal = {Nuclear Physics B},
  volume  = {242},
  number  = {1},
  pages   = {244--268},
  year    = {1984},
  doi     = {10.1016/0550-3213(84)90142-1}
}

@article{FerraraGunaydin2006,
  author  = {Ferrara, Sergio and G{\"u}naydin, Murat},
  title   = {Orbits and Attractors for {$N=2$} Maxwell--Einstein Supergravity Theories in Five Dimensions},
  journal = {Nuclear Physics B},
  volume  = {759},
  number  = {1-2},
  pages   = {1--19},
  year    = {2006},
  doi     = {10.1016/j.nuclphysb.2006.09.016}
}

@article{BianchiFerrara2008,
  author  = {Bianchi, Massimo and Ferrara, Sergio},
  title   = {Enriques and Octonionic Magic Supergravity Models},
  journal = {JHEP},
  volume  = {2008},
  number  = {2},
  pages   = {054},
  year    = {2008},
  doi     = {10.1088/1126-6708/2008/02/054}
}

@article{FerraraKalloshStrominger1995,
  author  = {Ferrara, Sergio and Kallosh, Renata and Strominger, Andrew},
  title   = {{$N=2$} Extremal Black Holes},
  journal = {Physical Review D},
  volume  = {52},
  number  = {10},
  pages   = {R5412--R5416},
  year    = {1995},
  doi     = {10.1103/PhysRevD.52.R5412}
}

@article{FerraraKallosh1996,
  author  = {Ferrara, Sergio and Kallosh, Renata},
  title   = {Supersymmetry and Attractors},
  journal = {Physical Review D},
  volume  = {54},
  number  = {2},
  pages   = {1514--1524},
  year    = {1996},
  doi     = {10.1103/PhysRevD.54.1514}
}

@article{FerraraGibbonsKallosh1997,
  author  = {Ferrara, Sergio and Gibbons, Gary W. and Kallosh, Renata},
  title   = {Black Holes and Critical Points in Moduli Space},
  journal = {Nuclear Physics B},
  volume  = {500},
  number  = {1-3},
  pages   = {75--93},
  year    = {1997},
  doi     = {10.1016/S0550-3213(97)00324-6}
}

@article{MeessenOrtinPerzShahbazi2012d5,
  author  = {Meessen, Patrick and Ort{\'i}n, Tom{\'a}s and Perz, Jan and Shahbazi, C. S.},
  title   = {Black Holes and Black Strings of {$N=2$, $D=5$} Supergravity in the H-FGK Formalism},
  journal = {JHEP},
  volume  = {2012},
  number  = {9},
  pages   = {001},
  year    = {2012},
  doi     = {10.1007/JHEP09(2012)001}
}

@article{Petersson2019,
  author = {Petersson, Holger P.},
  title = {A Survey on Albert Algebras},
  journal = {Transform. Groups},
  volume = {24},
  number = {1},
  pages = {219--278},
  year = {2019},
  doi = {10.1007/s00031-017-9471-4},
  url = {https://doi.org/10.1007/s00031-017-9471-4}
}

@article{GunaydinKidambi2022,
  author = {G{\"u}naydin, Murat and Kidambi, Abhiram},
  title = {Octonionic Magical Supergravity, Niemeier Lattices, and Exceptional \& Hilbert Modular Forms},
  journal = {Fortsch. Phys.},
  volume = {72},
  number = {2},
  pages = {2300242},
  year = {2024},
  doi = {10.1002/prop.202300242},
  eprint = {2209.05004},
  archivePrefix = {arXiv},
  primaryClass = {hep-th}
}

@article{DonnellyFreidel2016,
  author = {Donnelly, William and Freidel, Laurent},
  title = {Local subsystems in gauge theory and gravity},
  journal = {JHEP},
  volume = {2016},
  number = {09},
  pages = {102},
  year = {2016},
  doi = {10.1007/JHEP09(2016)102},
  eprint = {1601.04744},
  archivePrefix = {arXiv},
  primaryClass = {hep-th}
}

@article{DonnellyGiddings2017,
  author = {Donnelly, William and Giddings, Steven B.},
  title = {How is quantum information localized in gravity?},
  journal = {Phys. Rev. D},
  volume = {96},
  number = {8},
  pages = {086013},
  year = {2017},
  doi = {10.1103/PhysRevD.96.086013},
  eprint = {1706.03104},
  archivePrefix = {arXiv},
  primaryClass = {hep-th}
}

@article{Raju2022,
  author = {Raju, Suvrat},
  title = {Lessons from the Information Paradox},
  journal = {Phys. Rept.},
  volume = {943},
  pages = {1--80},
  year = {2022},
  doi = {10.1016/j.physrep.2021.10.001},
  url = {https://doi.org/10.1016/j.physrep.2021.10.001},
  eprint = {2012.05770},
  archivePrefix = {arXiv},
  primaryClass = {hep-th}
}

@article{Almheiri2020,
  author = {Almheiri, Ahmed and Hartman, Thomas and Maldacena, Juan and Shaghoulian, Edgar and Tajdini, Amirhossein},
  title = {Replica Wormholes and the Entropy of Hawking Radiation},
  journal = {JHEP},
  volume = {2020},
  number = {05},
  pages = {013},
  year = {2020},
  doi = {10.1007/JHEP05(2020)013},
  url = {https://doi.org/10.1007/JHEP05(2020)013},
  eprint = {1911.12333},
  archivePrefix = {arXiv},
  primaryClass = {hep-th}
}

@book{Schafer1966,
  author = {Schafer, Richard D.},
  title = {An Introduction to Nonassociative Algebras},
  publisher = {Academic Press},
  address = {New York},
  year = {1966}
}

@article{Baez2002,
  author = {Baez, John C.},
  title = {The Octonions},
  journal = {Bull. Amer. Math. Soc.},
  volume = {39},
  pages = {145--205},
  year = {2002},
  doi = {10.1090/S0273-0979-01-00934-X},
  url = {https://doi.org/10.1090/S0273-0979-01-00934-X}
}

@book{Haag1992,
  author = {Haag, Rudolf},
  title = {Local Quantum Physics: Fields, Particles, Algebras},
  publisher = {Springer},
  address = {Berlin},
  year = {1992},
  doi = {10.1007/978-3-642-61458-3},
  url = {https://doi.org/10.1007/978-3-642-61458-3}
}

@article{Barnum2014,
  author = {Barnum, Howard and Wilce, Alexander},
  title = {Local Tomography and the Jordan Structure of Quantum Theory},
  journal = {Found. Phys.},
  volume = {44},
  pages = {192--212},
  year = {2014},
  doi = {10.1007/s10701-014-9777-1},
  url = {https://doi.org/10.1007/s10701-014-9777-1}
}

@article{Petersson2004,
  author = {Petersson, Holger P.},
  title = {Structure Theorems for Jordan Algebras of Degree Three over Fields of Arbitrary Characteristic},
  journal = {Comm. Algebra},
  volume = {32},
  number = {3},
  pages = {1019--1049},
  year = {2004},
  doi = {10.1081/AGB-120027965},
  url = {https://doi.org/10.1081/AGB-120027965}
}

@article{Niestegge2015,
  author = {Niestegge, Gerd},
  title = {Dynamical Correspondence in a Generalized Quantum Theory},
  journal = {Found. Phys.},
  volume = {45},
  number = {5},
  pages = {525--534},
  year = {2015},
  doi = {10.1007/s10701-015-9881-x},
  url = {https://doi.org/10.1007/s10701-015-9881-x},
  eprint = {1402.0158},
  archivePrefix = {arXiv},
  primaryClass = {math-ph}
}

@article{Feshbach1958,
  author  = {Feshbach, Herman},
  title   = {Unified Theory of Nuclear Reactions},
  journal = {Annals of Physics},
  volume  = {5},
  number  = {4},
  pages   = {357--390},
  year    = {1958},
  doi     = {10.1016/0003-4916(58)90007-1}
}

@book{Pazy1983,
  author    = {Pazy, Amnon},
  title     = {Semigroups of Linear Operators and Applications to Partial Differential Equations},
  publisher = {Springer},
  address   = {New York},
  year      = {1983},
  doi       = {10.1007/978-1-4612-5561-1}
}

@book{Gripenberg1990,
  author    = {Gripenberg, Gustaf and Londen, Stig-Olof and Staffans, Olof},
  title     = {Volterra Integral and Functional Equations},
  publisher = {Cambridge University Press},
  address   = {Cambridge},
  year      = {1990}  
}

@article{Nakajima1958,
  author  = {Nakajima, Sadao},
  title   = {On Quantum Theory of Transport Phenomena: Steady Diffusion},
  journal = {Progress of Theoretical Physics},
  volume  = {20},
  number  = {6},
  pages   = {948--959},
  year    = {1958},
  doi     = {10.1143/PTP.20.948}
}

@article{Zwanzig1960,
  author  = {Zwanzig, Robert},
  title   = {Ensemble Method in the Theory of Irreversibility},
  journal = {The Journal of Chemical Physics},
  volume  = {33},
  number  = {5},
  pages   = {1338--1341},
  year    = {1960},
  doi     = {10.1063/1.1731409}
}

@book{BreuerPetruccione2002,
  author    = {Breuer, Heinz-Peter and Petruccione, Francesco},
  title     = {The Theory of Open Quantum Systems},
  publisher = {Oxford University Press},
  address   = {Oxford},
  year      = {2002},
  doi       = {10.1093/acprof:oso/9780199213900.001.0001}
}

@article{HiscockWeems1990,
  author  = {Hiscock, William A. and Weems, Larissa D.},
  title   = {Evolution of Charged Evaporating Black Holes},
  journal = {Physical Review D},
  volume  = {41},
  number  = {4},
  pages   = {1142--1151},
  year    = {1990},
  doi     = {10.1103/PhysRevD.41.1142}
}

@article{ParisRehacek2004,
  author  = {Paris, Matteo G. A. and Reh{\'a}{\v{c}}ek, Jaroslav},
  title   = {Quantum State Estimation},
  journal = {Lecture Notes in Physics},
  volume  = {649},
  pages   = {1--455},
  year    = {2004},
  doi     = {10.1007/b98673}
}

@article{AlmheiriRMP2021,
  author  = {Almheiri, Ahmed and Hartman, Thomas and Maldacena, Juan and Shaghoulian, Edgar and Tajdini, Amirhossein},
  title   = {The Entropy of Hawking Radiation},
  journal = {Reviews of Modern Physics},
  volume  = {93},
  number  = {3},
  pages   = {035002},
  year    = {2021},
  doi     = {10.1103/RevModPhys.93.035002}
}

@article{Alonso2023PolarizationGeometry,
  author  = {Alonso, Miguel A.},
  title   = {Geometric Descriptions for the Polarization of Nonparaxial Light: A Tutorial},
  journal = {Advances in Optics and Photonics},
  volume  = {15},
  number  = {1},
  pages   = {176--235},
  year    = {2023},
  doi     = {10.1364/AOP.475491}
}

@article{Beck2007,
  author  = {Beck, Christian},
  title   = {Superstatistics: Theory and Applications},
  journal = {Continuum Mechanics and Thermodynamics},
  volume  = {16},
  number  = {3},
  pages   = {293--304},
  year    = {2004},
  doi     = {10.1007/s00161-003-0145-1}
}

@article{Beck2009HEP,
  author  = {Beck, Christian},
  title   = {Generalised Information and Entropy Measures in Physics},
  journal = {Contemporary Physics},
  volume  = {50},
  number  = {4},
  pages   = {495--510},
  year    = {2009},
  doi     = {10.1080/00107510902823517}
}

@article{BeckCohen2003,
  author  = {Beck, Christian and Cohen, E. G. D.},
  title   = {Superstatistics},
  journal = {Physica A: Statistical Mechanics and its Applications},
  volume  = {322},
  pages   = {267--275},
  year    = {2003},
  doi     = {10.1016/S0378-4371(03)00019-0}
}

@article{Bernstein1929,
  author  = {Bernstein, Serge},
  title   = {Sur les fonctions absolument monotones},
  journal = {Acta Mathematica},
  volume  = {52},
  pages   = {1--66},
  year    = {1929},
  doi     = {10.1007/BF02592679}
}

@article{CeresoleFerraraMarrani2007,
  author  = {Ceresole, Anna and Ferrara, Sergio and Marrani, Alessio},
  title   = {4d/5d Correspondence for the Black Hole Potential and its Critical Points},
  journal = {Classical and Quantum Gravity},
  volume  = {24},
  number  = {22},
  pages   = {5651--5666},
  year    = {2007},
  doi     = {10.1088/0264-9381/24/22/023},
  eprint  = {0707.0964},
  archivePrefix = {arXiv},
  primaryClass  = {hep-th}
}

@article{FaulknerSperanza2024,
  author       = {Thomas Faulkner and Antony J. Speranza},
  title        = {Gravitational algebras and the generalized second law},
  journal      = {J. High Energy Phys.},
  year         = {2024},
  volume       = {2024},
  number       = {11},
  pages        = {099},
  doi          = {10.1007/JHEP11(2024)099},
  eprint       = {2405.00847},
  archivePrefix= {arXiv},
  primaryClass = {hep-th}
}

@article{Frigori2014NonextensiveLGT,
  title={Nonextensive lattice gauge theories: Algorithms and methods},
  author={Frigori, Rafael B},
  journal={Computer Physics Communications},
  volume={185},
  number={8},
  pages={2232--2239},
  year={2014},
  publisher={Elsevier},
  doi     = {10.1016/j.cpc.2014.04.016}
}

@book{GilOssikovski2022,
  author    = {Gil, Jos{\'e} Jorge and Ossikovski, Razvigor},
  title     = {Polarized Light and the Mueller Matrix Approach},
  edition   = {2},
  publisher = {CRC Press},
  address   = {Boca Raton},
  year      = {2022},
  doi       = {10.1201/9780367815578}
}

@article{Harlow2014PR,
  author       = {Daniel Harlow},
  title        = {Aspects of the {P}apadodimas--{R}aju proposal for the black hole interior},
  journal      = {J. High Energy Phys.},
  year         = {2014},
  volume       = {2014},
  number       = {11},
  pages        = {055},
  doi          = {10.1007/JHEP11(2014)055},
  eprint       = {1405.1995},
  archivePrefix= {arXiv},
  primaryClass = {hep-th}
}

@article{JensenSorceSperanza2023,
  author       = {Kristan Jensen and Jonathan Sorce and Antony J. Speranza},
  title        = {Generalized entropy for general subregions in quantum gravity},
  journal      = {J. High Energy Phys.},
  year         = {2023},
  volume       = {2023},
  number       = {12},
  pages        = {020},
  doi          = {10.1007/JHEP12(2023)020},
  eprint       = {2306.01837},
  archivePrefix= {arXiv},
  primaryClass = {hep-th}
}

@article{Page1976I,
  author  = {Page, Don N.},
  title   = {Particle Emission Rates from a Black Hole: Massless Particles from an Uncharged, Nonrotating Hole},
  journal = {Physical Review D},
  volume  = {13},
  number  = {2},
  pages   = {198--206},
  year    = {1976},
  doi     = {10.1103/PhysRevD.13.198}
}

@article{Page1976II,
  author  = {Page, Don N.},
  title   = {Particle Emission Rates from a Black Hole. II. Massless Particles from a Rotating Hole},
  journal = {Physical Review D},
  volume  = {14},
  number  = {12},
  pages   = {3260--3273},
  year    = {1976},
  doi     = {10.1103/PhysRevD.14.3260}
}

@article{PapadodimasRaju2013,
  author       = {Kyriakos Papadodimas and Suvrat Raju},
  title        = {An infalling observer in {AdS}/{CFT}},
  journal      = {J. High Energy Phys.},
  year         = {2013},
  volume       = {2013},
  number       = {10},
  pages        = {212},
  doi          = {10.1007/JHEP10(2013)212},
  eprint       = {1211.6767},
  archivePrefix= {arXiv},
  primaryClass = {hep-th}
}

@article{PapadodimasRaju2016,
  author       = {Kyriakos Papadodimas and Suvrat Raju},
  title        = {Remarks on the necessity and implications of state-dependence in the black hole interior},
  journal      = {Phys. Rev. D},
  year         = {2016},
  volume       = {93},
  pages        = {084049},
  doi          = {10.1103/PhysRevD.93.084049},
  eprint       = {1503.08825},
  archivePrefix= {arXiv},
  primaryClass = {hep-th}
}

@article{ParikhWilczek2000,
  author  = {Parikh, Maulik K. and Wilczek, Frank},
  title   = {Hawking Radiation as Tunneling},
  journal = {Physical Review Letters},
  volume  = {85},
  pages   = {5042--5045},
  year    = {2000},
  doi     = {10.1103/PhysRevLett.85.5042},
  eprint  = {hep-th/9907001},
  archivePrefix = {arXiv},
  primaryClass = {hep-th}
}

@article{Penington2020,
  author  = {Penington, Geoffrey},
  title   = {Entanglement Wedge Reconstruction and the Information Paradox},
  journal = {JHEP},
  volume  = {2020},
  number  = {9},
  pages   = {002},
  year    = {2020},
  doi     = {10.1007/JHEP09(2020)002}
}

@article{Shalizi2007,
  author  = {Shalizi, Cosma Rohilla},
  title   = {Maximum Likelihood Estimation for q-Exponential (Tsallis) Distributions},
  journal = {arXiv e-prints},
  eprint  = {math/0701854},
  archivePrefix = {arXiv},
  primaryClass  = {math.ST},
  year    = {2007}
}

@article{Tsallis1988,
  author  = {Tsallis, Constantino},
  title   = {Possible Generalization of Boltzmann--Gibbs Statistics},
  journal = {Journal of Statistical Physics},
  volume  = {52},
  pages   = {479--487},
  year    = {1988},
  doi     = {10.1007/BF01016429}
}

@book{Tsallis2009,
  author    = {Tsallis, Constantino},
  title     = {Introduction to Nonextensive Statistical Mechanics: Approaching a Complex World},
  publisher = {Springer},
  address   = {New York},
  year      = {2009},
  doi       = {10.1007/978-0-387-85359-8}
}

@article{Vaidya1943,
  author  = {Vaidya, P. C.},
  title   = {The External Field of a Radiating Star in General Relativity},
  journal = {Current Science},
  volume  = {12},
  pages   = {183--184},
  year    = {1943}
}

@book{Widder1941,
  author    = {Widder, David Vernon},
  title     = {The Laplace Transform},
  publisher = {Princeton University Press},
  address   = {Princeton},
  year      = {1941}
}

@article{Witten2022Crossed,
  author       = {Edward Witten},
  title        = {Gravity and the crossed product},
  journal      = {J. High Energy Phys.},
  year         = {2022},
  volume       = {2022},
  number       = {10},
  pages        = {008},
  doi          = {10.1007/JHEP10(2022)008},
  eprint       = {2112.12828},
  archivePrefix= {arXiv},
  primaryClass = {hep-th}
}

@article{ZaburdaevDenisovKlafter2015,
  author  = {Zaburdaev, V. and Denisov, S. and Klafter, J.},
  title   = {L{\'e}vy Walks},
  journal = {Reviews of Modern Physics},
  volume  = {87},
  number  = {2},
  pages   = {483--530},
  year    = {2015},
  doi     = {10.1103/RevModPhys.87.483}
}

@article{bonnorvaidya1970,
  author = {Bonnor, W. B. and Vaidya, P. C.},
  title = {Spherically symmetric radiation of charge in Einstein-Maxwell theory},
  journal = {Gen. Relativ. Gravit.},
  volume = {1},
  pages = {127--130},
  year = {1970},
  doi = {10.1007/BF00756891},
  url = {https://doi.org/10.1007/BF00756891}
}

@book{chandrasekhar1983,
  author = {Chandrasekhar, S.},
  title = {The Mathematical Theory of Black Holes},
  publisher = {Oxford University Press},
  address = {Oxford},
  year = {1983},
  isbn = {9780198503705}
}

@article{hod2016,
  author = {Hod, Shahar},
  title = {Entropy emission properties of near-extremal Reissner-Nordstr{\"o}m black holes},
  journal = {Phys. Rev. D},
  volume = {93},
  pages = {104027},
  year = {2016},
  doi = {10.1103/PhysRevD.93.104027},
  url = {https://doi.org/10.1103/PhysRevD.93.104027}
}

@article{lindquist1965,
  author = {Lindquist, R. W. and Schwartz, R. A. and Misner, C. W.},
  title = {Vaidya's Radiating Schwarzschild Metric},
  journal = {Phys. Rev.},
  volume = {137},
  pages = {B1364--B1368},
  year = {1965},
  doi = {10.1103/PhysRev.137.B1364},
  url = {https://doi.org/10.1103/PhysRev.137.B1364}
}

@article{moncrief1974a,
  author = {Moncrief, Vincent},
  title = {Stability of Reissner-Nordstr{\"o}m black holes},
  journal = {Phys. Rev. D},
  volume = {10},
  pages = {1057--1059},
  year = {1974},
  doi = {10.1103/PhysRevD.10.1057},
  url = {https://doi.org/10.1103/PhysRevD.10.1057}
}

@article{moncrief1975,
  author = {Moncrief, Vincent},
  title = {Gauge-Invariant Perturbations of Reissner-Nordstr{\"o}m Black Holes. II},
  journal = {Phys. Rev. D},
  volume = {12},
  pages = {1526--1537},
  year = {1975},
  doi = {10.1103/PhysRevD.12.1526},
  url = {https://doi.org/10.1103/PhysRevD.12.1526}
}

@article{ngampitipan2013,
  author = {Ngampitipan, Tritos and Boonserm, Petarpa},
  title = {Bounding the greybody factors for the Reissner-Nordstr{\"o}m black holes},
  journal = {J. Phys.: Conf. Ser.},
  volume = {435},
  pages = {012027},
  year = {2013},
  doi = {10.1088/1742-6596/435/1/012027},
  url = {https://doi.org/10.1088/1742-6596/435/1/012027}
}

@article{zerilli1974,
  author = {Zerilli, Frank J.},
  title = {Perturbation analysis for gravitational and electromagnetic radiation in a Reissner-Nordstr{\"o}m geometry},
  journal = {Phys. Rev. D},
  volume = {9},
  pages = {860--868},
  year = {1974},
  doi = {10.1103/PhysRevD.9.860},
  url = {https://doi.org/10.1103/PhysRevD.9.860}
}

@article{Segal1947,
  author = {Segal, Irving E.},
  title = {Postulates for General Quantum Mechanics},
  journal = {Annals of Mathematics},
  series = {Second Series},
  volume = {48},
  number = {4},
  pages = {930--948},
  year = {1947},
  doi = {10.2307/1969387},
  url = {https://doi.org/10.2307/1969387}
}

@article{HaagKastler1964,
  author = {Haag, Rudolf and Kastler, Daniel},
  title = {An Algebraic Approach to Quantum Field Theory},
  journal = {Journal of Mathematical Physics},
  volume = {5},
  number = {7},
  pages = {848--861},
  year = {1964},
  doi = {10.1063/1.1704187},
  url = {https://doi.org/10.1063/1.1704187}
}

@article{GunaydinPironRuegg1978,
  author = {G{\"u}naydin, Murat and Piron, Constantin and Ruegg, Henri},
  title = {Moufang plane and octonionic Quantum Mechanics},
  journal = {Communications in Mathematical Physics},
  volume = {61},
  number = {1},
  pages = {69--85},
  year = {1978},
  doi = {10.1007/BF01609468},
  url = {https://doi.org/10.1007/BF01609468}
}

@article{Albert1934,
  author  = {Albert, A. Adrian},
  title   = {On a Certain Algebra of Quantum Mechanics},
  journal = {Annals of Mathematics},
  volume  = {35},
  number  = {1},
  pages   = {65--73},
  year    = {1934},
  doi     = {10.2307/1968118}
}

@book{Feller1971,
  author    = {Feller, William},
  title     = {An Introduction to Probability Theory and Its Applications, Volume II},
  edition   = {2},
  publisher = {Wiley},
  address   = {New York},
  year      = {1971},
  isbn      = {9780471257097}
}

@book{Redner2001,
  author    = {Redner, Sidney},
  title     = {A Guide to First-Passage Processes},
  publisher = {Cambridge University Press},
  address   = {Cambridge},
  year      = {2001},
  doi       = {10.1017/CBO9780511606014},
  isbn      = {9780521652483}
}

@book{FarautKoranyi1994SymmetricCones,
  author    = {Faraut, Jacques and Kor{\'a}nyi, Adam},
  title     = {Analysis on Symmetric Cones},
  publisher = {Oxford University Press},
  address   = {Oxford},
  year      = {1994}
}

@article{Krutelevich2007Jordan,
  author  = {Krutelevich, Sergei},
  title   = {Jordan Algebras, Exceptional Groups, and {Bhargava} Composition},
  journal = {Journal of Algebra},
  volume  = {314},
  number  = {2},
  pages   = {924--977},
  year    = {2007},
  doi     = {10.1016/j.jalgebra.2007.02.060}
}

@article{Mathur2009SmallCorrections,
  author  = {Mathur, Samir D.},
  title   = {The Information Paradox: A Pedagogical Introduction},
  journal = {Classical and Quantum Gravity},
  volume  = {26},
  pages   = {224001},
  year    = {2009},
  doi     = {10.1088/0264-9381/26/22/224001},
  eprint  = {0909.1038},
  archivePrefix = {arXiv},
  primaryClass = {hep-th}
}

@article{AnastopoulosSavvidou2020MultiTime,
  author  = {Anastopoulos, Charis and Savvidou, Ntina},
  title   = {Multi-Time Measurements in Hawking Radiation: Information at Higher-Order Correlations},
  journal = {Classical and Quantum Gravity},
  volume  = {37},
  pages   = {025015},
  year    = {2020},
  doi     = {10.1088/1361-6382/ab5eb2},
  eprint  = {1909.00438},
  archivePrefix = {arXiv},
  primaryClass = {gr-qc}
}

@article{Giddings2012Nonviolent,
  author  = {Giddings, Steven B.},
  title   = {Nonviolent Nonlocality},
  journal = {Physical Review D},
  volume  = {88},
  pages   = {064023},
  year    = {2013},
  doi     = {10.1103/PhysRevD.88.064023},
  eprint  = {1211.7070},
  archivePrefix = {arXiv},
  primaryClass = {hep-th}
}

@article{HawkingPerryStrominger2016SoftHair,
  author  = {Hawking, Stephen W. and Perry, Malcolm J. and Strominger, Andrew},
  title   = {Soft Hair on Black Holes},
  journal = {Physical Review Letters},
  volume  = {116},
  pages   = {231301},
  year    = {2016},
  doi     = {10.1103/PhysRevLett.116.231301},
  eprint  = {1601.00921},
  archivePrefix = {arXiv},
  primaryClass = {hep-th}
}

@article{Koecher1957Positivitaetsbereiche,
  author  = {Koecher, Max},
  title   = {Positivit{\"a}tsbereiche im \(R^n\)},
  journal = {American Journal of Mathematics},
  volume  = {79},
  number  = {3},
  pages   = {575--596},
  year    = {1957},
  doi     = {10.2307/2372563}
}

@article{Vinberg1963HomogeneousCones,
  author  = {Vinberg, E. B.},
  title   = {The theory of convex homogeneous cones},
  journal = {Transactions of the Moscow Mathematical Society},
  volume  = {12},
  pages   = {340--403},
  year    = {1963},
  note    = {English translation of Trudy Moskov. Mat. Obshch. 12, 303--358}
}

@article{BellissardIochum1978,
  author  = {Bellissard, Jean and Iochum, B.},
  title   = {Homogeneous self dual cones versus {Jordan} algebras. {The} theory revisited},
  journal = {Annales de l'Institut Fourier},
  volume  = {28},
  number  = {1},
  pages   = {27--67},
  year    = {1978},
  doi     = {10.5802/aif.680}
}

@article{Lomax1954,
  author  = {Lomax, K. S.},
  title   = {Business Failures: Another Example of the Analysis of Failure Data},
  journal = {Journal of the American Statistical Association},
  volume  = {49},
  number  = {268},
  pages   = {847--852},
  year    = {1954},
  doi     = {10.1080/01621459.1954.10501239}
}

\appendix
\small

\section{Albertian preliminaries and sourced benchmark-canonical kernel}
\label{app:kernel}

This appendix first fixes the Albertian conventions used throughout the paper and then collects the sourced benchmark derivations needed to turn the near-horizon structure into the benchmark kernel used in the reduced-channel problem.

\subsection{Albertian AQM preliminaries}
\label{app:preliminaries}

The Albert algebra is the Jordan algebra of $3\times 3$ Hermitian octonionic matrices,
\begin{equation}
X=
\begin{pmatrix}
\xi_1 & x_3 & \bar x_2\\
\bar x_3 & \xi_2 & x_1\\
x_2 & \bar x_1 & \xi_3
\end{pmatrix},
\qquad
\xi_i\in\R,
\quad x_i\in\Oa,
\label{eq:AlbertMatrix}
\end{equation}
with Jordan product
\begin{equation}
X\circ Y=\frac12(XY+YX).
\label{eq:JordanProduct}
\end{equation}
It is simple, formally real, and exceptional \cite{Jordan1934,Albert1934,McCrimmon2004,GunaydinPironRuegg1978}. Its cubic norm is
\begin{equation}
N(X)=\xi_1\xi_2\xi_3-\sum_i \xi_i |x_i|^2+2\,\mathrm{Re}(x_1x_2x_3),
\label{eq:AlbertNorm}
\end{equation}
and its quadratic adjoint $X^\#$ is defined by
\begin{equation}
X\circ X^\#=N(X)\,1,
\qquad
(X^\#)^\#=N(X)\,X.
\label{eq:AlbertAdjoint}
\end{equation}
In the notation used later in the paper, the cubic invariant is also written as
\begin{equation}
I_3=N.
\label{eq:I3equalsN}
\end{equation}
Primitive idempotents form the Cayley plane,
\begin{equation}
\Oa P^2\cong F_4/\mathrm{Spin}(9),
\label{eq:appendixCayley}
\end{equation}
and the derivation algebra is
\begin{equation}
\mathrm{Der}(J_3(\Oa))\cong \mathfrak f_4.
\label{eq:appendixf4}
\end{equation}
The Albertian state space is the compact convex set of normalized positive linear functionals,
\begin{equation}
\Omega(J_3(\Oa))=
\bigl\{\omega:J_3(\Oa)\to\R\ \big|\ \omega(X\circ X)\ge 0,\ \omega(1)=1\bigr\}.
\label{eq:AlbertStates}
\end{equation}
The quartic invariant $I_4$ used later in the paper belongs to the Freudenthal envelope of $J_3(\Oa)$,
\begin{equation}
\bigl(\Omega(J_3(\Oa)),\,I_3=N,\,\FTS(J_3(\Oa)),\,I_4\bigr),
\label{eq:appendixPackage}
\end{equation}
and enters explicitly only after the $5$D$\to4$D electric descent. Thus $N$ and $I_3$ denote the same cubic Albertian invariant, while $I_4$ is the associated quartic Freudenthal invariant used in the four-dimensional slice.

\noindent\textit{Symmetric-cone and symbol viewpoint.}
The ordered sector used in the main text is a retained real symbol sector, not the full Hilbert-space operator algebra of an ordinary geometric quantization. On the electric Freudenthal slice \(x_e=(\alpha_{\rm FTS},A;0,0)\), the retained block is \(A\in J_3(\Oa)\), and \(I_4(x_e)=-4\alpha_{\rm FTS}N(A)\). The admissible BPS orientation selects the exceptional positive cone; the Koecher--Vinberg reconstruction of that cone gives the Albert product \cite{FarautKoranyi1994SymmetricCones,Koecher1957Positivitaetsbereiche,Vinberg1963HomogeneousCones,BellissardIochum1978,Krutelevich2007Jordan}. This paragraph fixes the notation used below; the scope of the retained-symbol claim is established in Sec.~\ref{sec:Albertian}.

For later use in the five-dimensional attractor equations, we write the attractor polarization as
\begin{equation}
X_*=h^I e_I\in J_3(\Oa),
\label{eq:appendixXstar}
\end{equation}
where $e_I$ is a Jordan basis (or Jordan frame basis), $h^I$ are the very-special scalar coordinates, $h_I$ are their cubic duals, $q_I$ are the electric charge components in that basis, and $Z_*$ is the attractor value of the BPS central charge. The normalization
\begin{equation}
N(X_*)=1
\label{eq:appendixNXstar}
\end{equation}
means that $X_*$ lies on the unit cubic slice of the Albert cone. Here $\alpha_{\mathrm{FTS}}$ denotes the Freudenthal Triple System normalization constant used in the electric slice of the four-dimensional descent.

\subsection{Attractor anchor, ordinary slice, and source-fixed electric descent}
\label{app:attractor}

For a five-dimensional BPS black hole with charge element $Q\in J_3(\Oa)$, the attractor equations may be written in very-special form as
\begin{equation}
q_I=Z_* h_I,
\label{eq:app_qI}
\end{equation}
where $q_I$ are the electric charge components, $h^I$ are the very-special scalar coordinates, $h_I$ are their cubic duals, and $Z_*$ is the attractor value of the BPS central charge. Writing the attractor polarization as
\begin{equation}
X_*=h^I e_I\in J_3(\Oa),
\qquad
N(X_*)=1,
\label{eq:app_Xstar_basis}
\end{equation}
with $e_I$ a Jordan basis, one obtains the equivalent base-free relation
\begin{equation}
Q=Z_* X_*^\#.
\label{eq:app_basefree}
\end{equation}
Applying the adjoint again and using Eq.~\eqref{eq:AlbertAdjoint} gives
\begin{equation}
Q^\#=Z_*^{\,2} X_*,
\qquad
N(Q)=Z_*^{\,3},
\label{eq:app_Qsharp}
\end{equation}
hence
\begin{equation}
Z_*=N(Q)^{1/3},
\qquad
X_*(Q)=\frac{Q^\#}{N(Q)^{2/3}}.
\label{eq:app_exact_attractor}
\end{equation}

On the charge-ray-preserving sourced class,
\begin{equation}
Q(v)=s(v)Q_0,
\label{eq:app_charge_ray}
\end{equation}
one has
\begin{equation}
(s(v)Q_0)^\#=s(v)^2Q_0^\#,
\qquad
N(s(v)Q_0)=s(v)^3N(Q_0).
\label{eq:app_ray_homogeneity}
\end{equation}
Hence
\begin{align}
X_*(Q(v))
&=
\frac{(s(v)Q_0)^\#}{N(s(v)Q_0)^{2/3}}
=
\frac{s(v)^2Q_0^\#}{(s(v)^3N(Q_0))^{2/3}}
\nonumber\\
&=
\frac{Q_0^\#}{N(Q_0)^{2/3}}
=
X_*(Q_0),
\label{eq:app_ray_freeze}
\end{align}
which proves Eq.~\eqref{eq:survivalid}. The attractor-adapted horizon frame therefore remains frozen on the exact charge-ray-preserving class.

For the canonical primitive idempotent \(e=\mathrm{diag}(1,0,0)\), Peirce theory gives \(J_3(\Oa)=\Atwo\oplus\Aone\oplus\Azero\), with \(\Atwo=\R e\). The explicit matrix realization and multiplication rules are stated in Sec.~\ref{sec:Peirce}; here we only record the support used later. Choosing a unit imaginary octonion \(u\) and the associated associative copy \(\Cu\subset\Oa\) selects the ordinary detector block
\begin{equation}
H_3(\Cu)\hookrightarrow J_3(\Oa),
\qquad
\Qu\cong H_2(\C),
\qquad
\Azero=\Qu\oplus\Ku.
\label{eq:app_Cslice}
\end{equation}
Thus \(Q_u\) is the \(\Cu\)-valued Hermitian lower block inside \(A_0(e)\), while \(K_u\) is the orthogonal exceptional remainder. This is the fixed ordinary support used in the channel reduction.

For the electric descent, start from Eqs.~\eqref{eq:KKmetric}--\eqref{eq:KKgauge} and impose the electric truncation
\begin{equation}
\zeta^I=0,
\qquad J_0^H=0,
\qquad
x_e=
\begin{pmatrix}
\alpha_{\mathrm{FTS}} & A\\
0 & 0
\end{pmatrix}.
\end{equation}
Then the quartic Freudenthal invariant reduces to
\begin{equation}
I_4(x_e)=-4\alpha_{\mathrm{FTS}}N(A),
\end{equation}
so the charge scale used on the reference background is
\begin{equation}
Q_0=|I_4(x_e)|^{1/4}.
\end{equation}
The neutral-source background corresponds to $Q(v)=Q_0$ and $J_I^{H(5)N}=0$, so the source only drains the mass and
\begin{equation}
\eexc(v)=M(v)-M_{\mathrm{BPS}}(Q_0).
\end{equation}

\subsection{Projected blocks and WKB barrier factorization}
\label{app:projectedblocks}

Let $L(v,r)$ denote the linearized sourced generator in the attractor-adapted frame.
Choose active internal bases $\{k_a\}\subset K_{\mathrm{act}}$, $\{a_\alpha\}\subset A_{1,\mathrm{act}}$, and $\{q_i\}\subset Q_{\mathrm{act}}$, and write the active modes as
\begin{equation}
\begin{aligned}
|K,a\rangle&=k_a\otimes \phi_K(r;v),\\
|1,\alpha\rangle&=a_\alpha\otimes \phi_1(r;v),\\
|Q,i\rangle&=q_i\otimes \phi_Q(r;v).
\end{aligned}
\label{eq:E54}
\end{equation}
The tensor notation in Eq.~\eqref{eq:E54} is only the standard separation of the active finite internal label from the radial WKB profile. It does not introduce the global Hilbert-space tensor product whose availability is questioned in the AMPS factorization argument, nor a subsystem factorization of the retained Albertian horizon algebra.The opening block matrix elements are
\begin{equation}
(D_{1K})_{\alpha a}(v)=
\int \dd r\,
\phi_1(r;v)^*\,
\langle a_\alpha|L(v,r)|k_a\rangle\,
\phi_K(r;v).
\label{eq:E55}
\end{equation}
On the neutral-source background we impose two structural assumptions, both standard in a leading WKB barrier problem.

\begin{hypothesis}[Radial--internal separation in the active sourced sector]
After fixing the frozen attractor-adapted frame and restricting to the active pair-polarized support, the linearized operator separates into an internal matrix element times a radial scalar profile at leading WKB order.
\end{hypothesis}

\begin{hypothesis}[Single-barrier WKB regime]
The projected radial equation has a classically forbidden interval bounded by two turning points and admits standard WKB matching.
\end{hypothesis}

Under the first hypothesis,
\begin{equation}
\langle a_\alpha|L(v,r)|k_a\rangle=(M_{1K})_{\alpha a}\,\chi_{1K}(r;v),
\label{eq:E56}
\end{equation}
so Eq.~\eqref{eq:E55} becomes
\begin{equation}
(D_{1K})_{\alpha a}(v)=
(M_{1K})_{\alpha a}
\int \dd r\,
\phi_1(r;v)^*\chi_{1K}(r;v)\phi_K(r;v).
\label{eq:E57}
\end{equation}
This separation is the leading adiabatic WKB analogue of the usual
channel factorization in radial barrier problems: the finite active
Peirce labels carry the internal matrix element, while the RN radial
profile supplies the scalar tunnelling factor. If the internal frame
varies across the barrier, or if several coupled radial channels are
retained, the scalar factor is replaced by a matrix WKB transport
operator.

Define the scalar overlap
\begin{equation}
T_{1K}(v):=\int \dd r\,\phi_1(r;v)^*\chi_{1K}(r;v)\phi_K(r;v).
\label{eq:E58}
\end{equation}
Then
\begin{equation}
D_{1K}(v)=M_{1K}\,T_{1K}(v).
\label{eq:E59}
\end{equation}
The same argument yields
\begin{equation}
D_{Q1}(v)=M_{Q1}\,T_{Q1}(v).
\label{eq:E60}
\end{equation}

\paragraph{WKB barrier evaluation.}
\label{app:wkb}
Assume the radial equation in the forbidden region takes the scalar form
\begin{equation}
\bigl[-\partial_r^2+V_{\mathrm{eff}}(r;v)\bigr]\phi(r;v)=0,
\label{eq:E61}
\end{equation}
with turning points $r_-(v)<r_+(v)$ and positive Euclidean momentum
\begin{equation}
\kappa(r;v)=\sqrt{V_{\mathrm{eff}}(r;v)}.
\label{eq:E62}
\end{equation}
The evanescent WKB solutions are
\begin{align}
\phi_K(r;v)&\sim
\frac{C_K(v)}{\sqrt{\kappa(r;v)}}
\exp\!\left[-\int_{r_-(v)}^r \kappa(\rho;v)\,\dd\rho\right],
\label{eq:E63}
\\
\phi_1(r;v)&\sim
\frac{C_1(v)}{\sqrt{\kappa(r;v)}}
\exp\!\left[-\int_r^{r_+(v)} \kappa(\rho;v)\,\dd\rho\right].
\label{eq:E64}
\end{align}
Multiplying these expressions gives
\begin{equation}
\begin{aligned}
\phi_1(r;v)^*\phi_K(r;v)
&\sim
\frac{C_1(v)C_K(v)}{\kappa(r;v)}\\
&\quad\times\exp\!\left[-\int_{r_-(v)}^{r_+(v)} \kappa(\rho;v)\,\dd\rho\right].
\end{aligned}
\label{eq:E65}
\end{equation}
The exponent is independent of $r$. Define
\begin{equation}
S_{1K}(v):=\int_{r_-(v)}^{r_+(v)} \kappa(\rho;v)\,\dd\rho.
\label{eq:E67}
\end{equation}
Then Eq.~\eqref{eq:E58} becomes
\begin{equation}
T_{1K}(v)\sim e^{-S_{1K}(v)}
\int_{r_-(v)}^{r_+(v)}
\frac{C_1(v)C_K(v)}{\kappa(r;v)}\chi_{1K}(r;v)\,\dd r.
\label{eq:E66}
\end{equation}
Absorbing the finite integral into a redefined internal matrix element gives the factorized form
\begin{equation}
\begin{aligned}
D_{1K}(v)&=M_{1K}e^{-S_{1K}(v)},\\
D_{Q1}(v)&=M_{Q1}e^{-S_{Q1}(v)}.
\end{aligned}
\label{eq:factorizedD}
\end{equation}
This is the meaning of the notation used in the body.

\paragraph{Active-support matrices.}
\label{app:actsupport}

On the active support every state is parameterized by a single $z\in\Cu$. Write the active relay maps as
\begin{equation}
M_{1K}\bigl(\Psi_K(z)\bigr)=\Psi_1(z),
\qquad
M_{Q1}\bigl(\Psi_1(z)\bigr)=\Psi_Q(z).
\end{equation}
In the real basis $\{1,u\}$ of $\Cu$, this means
\begin{equation}
\begin{aligned}
M_{1K}&=\idtwo,\\
M_{Q1}&=\idtwo,\\
J&=
\begin{pmatrix}
0 & -1\\
1 & 0
\end{pmatrix},\\
J^2&=-\idtwo.
\end{aligned}
\label{eq:explicitM}
\end{equation}
An overall scale could be absorbed into $a(v)$ and $b(v)$, so Eq.~\eqref{eq:explicitM} is the adapted canonical choice. The matrices \(M_{1K}\) and \(M_{Q1}\) fix the canonical opening and closing identifications of the active support. The matrix \(J\) is not an additional opening or closing map; it is the complex-structure generator on \(\mathbb C_u\) and enters the intrinsic relay block \(D_{11}(v)=\omega(v)J\) and the propagator \(U_{11}\).

\subsection{Benchmark-canonical bridge data and sourced clock}
\label{app:bridgeclock}

For the neutral-source background, the metric in ingoing
Eddington--Finkelstein form is \cite{ParikhWilczek2000}
\begin{equation}
\dd s^2=-f(r;M,Q_0)\,\dd v^2+2\,\dd v\,\dd r+r^2\dd\Omega^2,
\label{eq:E83}
\end{equation}
with
\begin{equation}
f(r;M,Q_0)=1-\frac{2M}{r}+\frac{Q_0^2}{r^2},
\qquad M=Q_0+E.
\label{eq:E83b}
\end{equation}
Here the charge is held fixed and \(E=M-Q_0\) is the
excitation energy above the fixed-charge extremal endpoint. The
calculation supplies the scalar barrier factor entering the bridge
amplitudes; it is not a derivation of the full semiclassical stress
tensor.

For an outgoing radial null geodesic,
\begin{equation}
0=-f\,\dd v^2+2\,\dd v\,\dd r
\qquad\Longrightarrow\qquad
\dot r:=\frac{\dd r}{\dd v}=\frac{f(r;M,Q_0)}{2}.
\label{eq:E84}
\end{equation}
The tunneling action is
\begin{equation}
S=\int p_r\,\dd r.
\end{equation}
Using Hamilton's equation \(\dd p_r=\dd H/\dot r\) gives
\begin{equation}
\mathrm{Im}\,S
=\mathrm{Im}\int\!\!\int\frac{\dd H}{\dot r}\,\dd r.
\label{eq:E87}
\end{equation}
Backreaction is implemented in the Parikh--Wilczek shell variable
\(\epsilon\) by \(H=M-\epsilon\) and \(\dd H=-\dd\epsilon\). This
\(\epsilon\) is only the tunneling-energy variable; it is not the
relay frequency \(\omega(v)\) used in the reduced channel.

Near the outer horizon,
\begin{equation}
f(r;M,Q_0)\approx f'(r_+)(r-r_+),
\qquad
\dot r\approx \frac{f'(r_+)}{2}(r-r_+).
\label{eq:E90}
\end{equation}
With
\begin{equation}
\kappa_{\rm sg}=\frac{f'(r_+)}{2},
\qquad
T_H=\frac{\kappa_{\rm sg}}{2\pi}=\frac{f'(r_+)}{4\pi},
\end{equation}
the pole gives
\begin{equation}
\mathrm{Im}\int\frac{\dd r}{\dot r}=
\frac{\pi}{\kappa_{\rm sg}}=
\frac{1}{2T_H}.
\end{equation}
Therefore
\begin{equation}
\mathrm{Im}\,S
=\frac12\int\frac{\dd E'}{T_H(E')}.
\label{eq:E100}
\end{equation}
Using the first law at fixed charge,
\begin{equation}
\dd E'=T_H(E')\,\dd\SBH(E'),
\end{equation}
one finds
\begin{equation}
\mathrm{Im}\,S
=\frac12\bigl(\SBH(E)-\SBH(0)\bigr).
\end{equation}
Thus the source-fixed barrier action is Eq.~\eqref{eq:Sstar}, and the
scalar bridge amplitudes are
\begin{equation}
a(E)=b(E)=e^{-S_*(E)}.
\end{equation}
Here \(S_*(E)\) is the amplitude-level barrier action. The
corresponding tunneling probability would carry the square of this
amplitude, while the Volterra bridge uses the amplitude factor
\(e^{-S_*(E)}\).

Differentiating Eq.~\eqref{eq:SBH} yields the thermodynamic slice
scale
\begin{equation}
T_H(E)=\left(\frac{\partial\SBH}{\partial E}\right)^{-1}.
\end{equation}
The benchmark uses this same source-fixed slice scale as the intrinsic
relay frequency. Thus \(\omega(E)\) is not an additional independent
frequency; it is the canonical bridge choice
\begin{equation}
\omega(E)=T_H(E),
\end{equation}
used in \(D_{11}(v)=\omega(E_{\rm exc}(v))J\), equivalently
Eq.~\eqref{eq:omega}.

The near-extremal expansion used later follows from
\begin{equation}
r_+(E)=Q_0+\sqrt{2Q_0}\sqrt{E}+E+O(E^{3/2}),
\end{equation}
which gives
\begin{equation}
\SBH(E)-\SBH(0)
=2\sqrt2\,\pi Q_0^{3/2}\sqrt E+4\pi Q_0 E+O(E^{3/2}),
\end{equation}
so that
\begin{equation}
\begin{aligned}
a(E)&=e^{-\mu E-\nu\sqrt E+O(E^{3/2})},\\
\mu&=2\pi Q_0,\\
\nu&=\sqrt2\,\pi Q_0^{3/2}.
\end{aligned}
\label{eq:munu}
\end{equation}

\paragraph{Sourced evaporative clock.}
\label{app:sourcedclockdetail}

The neutral-source reference background is not only a geometric
background but also a dynamical clock. The sourced Einstein equation
relates the mass-loss law to the null flux. To obtain an explicit
source-fixed evaporation clock, we close the null-flux sector by the
minimal Hawking-temperature constitutive law compatible with the
neutral-source background. Concretely, one writes
\cite{Vaidya1943,Hawking1975,Page1993,Page1976I,Page1976II}
\begin{equation}
L(v):=-\dot M(v)=\sigma_{\mathrm{eff}}A_+(E)T_H(E)^4,
\label{eq:StefanLike}
\end{equation}
with
\begin{equation}
\begin{aligned}
A_+(E)&=4\pi r_+(E)^2,\\
T_H(E)&=\left(\frac{\partial\SBH}{\partial E}\right)^{-1}
      =\frac{\sqrt{E(E+2Q_0)}}{2\pi r_+(E)^2}.
\end{aligned}
\end{equation}
Here \(L(v)\) is the effective luminosity, or mass-loss rate, of the
source-fixed benchmark. The coefficient \(\sigma_{\rm eff}\) is an
effective Stefan coefficient for this closure; it absorbs the species,
greybody, and unit normalizations not resolved by the sourced clock
itself.

Since the neutral background keeps \(Q(v)=Q_0\) fixed, one has
\(\dot E_{\mathrm{exc}}(v)=\dot M(v)\) and therefore
\begin{equation}
\dot E_{\mathrm{exc}}(v)
=-\frac{\sigma_{\mathrm{eff}}}{4\pi^3}
\frac{[E(E+2Q_0)]^2}{r_+(E)^6}.
\label{eq:E_dot_dimful}
\end{equation}
Introduce the dimensionless excitation variable
\begin{equation}
e(v):=\frac{E_{\mathrm{exc}}(v)}{Q_0},
\end{equation}
so that \(E_{\mathrm{exc}}=Q_0e\) and
\[
r_+(E)=Q_0\bigl(1+e+\sqrt{e(e+2)}\bigr).
\]
Then Eq.~\eqref{eq:E_dot_dimful} becomes
\begin{equation}
\dot e(v)=-\kappa_{\rm evap} F(e),
\qquad
\kappa_{\rm evap}:=\frac{\sigma_{\mathrm{eff}}}{4\pi^3Q_0^3},
\end{equation}
with
\begin{equation}
F(e):=\frac{[e(e+2)]^2}{\bigl(1+e+\sqrt{e(e+2)}\bigr)^6}.
\end{equation}
Finally normalize time by the initial evaporative rate,
\begin{equation}
\hat v:=\Gamma_{\mathrm{evap}}(0)v,
\qquad
\Gamma_{\mathrm{evap}}(0)=\kappa_{\rm evap}\frac{F(e_0)}{e_0},
\qquad e_0=e(0).
\end{equation}
Then the source-fixed clock closes as
\begin{equation}
\boxed{\frac{de}{d\hat v}=-\frac{e_0}{F(e_0)}F(e).}
\label{eq:sourcedclock}
\end{equation}
Equivalently,
\begin{equation}
\hat v(e)=\frac{F(e_0)}{e_0}\int_e^{e_0}\frac{\dd\varepsilon}{F(\varepsilon)}.
\end{equation}
This is the source-fixed evaporative clock used throughout the
quantitative benchmark and in the hidden-sector closure below.

\paragraph{Kernel on the sourced background.}
\label{app:kernel_source}

Using Eqs.~\eqref{eq:explicitM}, \eqref{eq:abstar}, and
\eqref{eq:omega}, the projected blocks are
\begin{align}
D_{1K}(v)&=a\bigl(\eexc(v)\bigr)\,\idtwo,
\\
D_{Q1}(v)&=a\bigl(\eexc(v)\bigr)\,\idtwo,
\\
D_{11}(v)&=\omega\bigl(\eexc(v)\bigr)J.
\end{align}
Define
\begin{equation}
\Theta(v,s)=\int_s^v \omega\bigl(\eexc(u)\bigr)\,\dd u.
\end{equation}
Then
\begin{equation}
U_{11}(v,s)=\cos\Theta(v,s)\,\idtwo+\sin\Theta(v,s)J,
\end{equation}
and the kernel in this benchmark is
\begin{equation}
\begin{aligned}
K_{QK}(v,s)&=a\bigl(\eexc(s)\bigr)a\bigl(\eexc(v)\bigr)\\
&\quad\times
\bigl[\cos\Theta(v,s)\,\idtwo+\sin\Theta(v,s)J\bigr].
\end{aligned}
\label{eq:kernelmainApp}
\end{equation}
This is the explicit kernel used throughout the reduced-state analysis.

\section{Reduced qubit state, hidden-sector closure, and admissibility}
\label{app:qubit}

This appendix collects the technical support behind the reduced qubit state. The spectral gap, entropy, and readout formulas are kept in the main text; what remains here are the source-fixed hidden-sector closure and the reduced-channel positivity criterion.

\subsection{Closure of the hidden-sector history}
\label{app:hiddenclosure}

The hidden-sector history variable can be closed in this benchmark by eliminating the interface in the $K_u$ equation rather than only in the radiative equation.
Start from the block system
\begin{align}
\dot\Psi_K &= D_{KK}\Psi_K + D_{K1}\Psi_1 + D_{KQ}\Psi_Q, \\
\dot\Psi_1 &= D_{1K}\Psi_K + D_{11}\Psi_1 + D_{1Q}\Psi_Q, \\
\dot\Psi_Q &= D_{QK}\Psi_K + D_{Q1}\Psi_1 + D_{QQ}\Psi_Q.
\end{align}
Using the retarded interface solution,
\begin{equation}
\begin{aligned}
\Psi_1(v)=&\;U_{11}(v,0)\Psi_1(0)\\
&+\int_0^v U_{11}(v,s)\bigl[D_{1K}(s)\Psi_K(s)+D_{1Q}(s)\Psi_Q(s)\bigr] \, ds.
\end{aligned}
\end{equation}
one finds the hidden-sector Volterra equation
\begin{equation}
\begin{aligned}
\dot\Psi_K(v)= {}& D_{KK}(v)\Psi_K(v)+D_{KQ}(v)\Psi_Q(v)+J_K^{(0)}(v)\\
&+\int_0^v K_{KK}(v,s)\Psi_K(s)\,ds\\
&+\int_0^v K_{KQ}(v,s)\Psi_Q(s)\,ds.
\end{aligned}
\end{equation}
with
\begin{equation}
J_K^{(0)}(v):=D_{K1}(v)U_{11}(v,0)\Psi_1(0),
\end{equation}
and
\begin{equation}
\begin{aligned}
K_{KK}(v,s)&:=D_{K1}(v)U_{11}(v,s)D_{1K}(s),\\
K_{KQ}(v,s)&:=D_{K1}(v)U_{11}(v,s)D_{1Q}(s).
\end{aligned}
\end{equation}
For the one-way relay reduction one imposes
\begin{equation}
\Psi_Q(0)=0,
\qquad
\Psi_1(0)=0,
\qquad
D_{1Q}=0,
\qquad
D_{KQ}=0.
\end{equation}
This is a benchmark specialization, not the most general Albertian block reduction. Keeping $D_{1Q}$ would add a $Q$-self-memory kernel through the retarded interface solution, while keeping $D_{KQ}$ would feed the readout back into the hidden sector. Thus the Volterra mechanism is more general than the one-way kernel, but the closed pair-polarized expression studied in the main text is the source-fixed one-way active-support member.
We then choose the emission sign for this setup
\begin{equation}
D_{K1}(v)=-a\bigl(E_{\mathrm{exc}}(v)\bigr)\,\idtwo.
\end{equation}
Then the hidden-sector equation is
\begin{equation}
\begin{aligned}
\dot\Psi_K(v)= {}& D_{KK}(v)\Psi_K(v)\\
&-a(v)\int_0^v a(s)\,
\Bigl[\cos\Theta(v,s)\,\idtwo+\sin\Theta(v,s)\,J\Bigr]\\
&\qquad\times\Psi_K(s)\,ds.
\end{aligned}
\label{eq:hidden_volterra_app}
\end{equation}
Introduce the hidden propagator
\begin{equation}
\partial_vU_{KK}(v,s)=D_{KK}(v)U_{KK}(v,s),
\qquad U_{KK}(s,s)=\idtwo,
\end{equation}
and the rotating frame
\begin{equation}
R(v):=e^{\Theta(v,0)J}.
\end{equation}
In the co-moving variable
\begin{equation}
\widetilde\Xi_K(v):=R(v)^{-1}U_{KK}(0,v)\Psi_K(v),
\end{equation}
Eq.~\eqref{eq:hidden_volterra_app} reduces to
\begin{equation}
\frac{d}{dv}\widetilde\Xi_K(v)=-a(v)\int_0^v a(s)\widetilde\Xi_K(s)\,ds.
\label{eq:Xi_tilde_app}
\end{equation}
Set
\begin{equation}
Y(v):=\int_0^v a(s)\widetilde\Xi_K(s)\,ds.
\end{equation}
Then
\begin{equation}
\dot{\widetilde\Xi}_K(v)=-a(v)Y(v),
\qquad
\dot Y(v)=a(v)\widetilde\Xi_K(v).
\end{equation}
With the complexified variable
\begin{equation}
U(v):=\widetilde\Xi_K(v)+iY(v),
\end{equation}
one finds
\begin{equation}
\begin{aligned}
\dot U(v)&=i a(v)U(v),\\
U(v)&=e^{iA(v)}U(0),\\
A(v)&:=\int_0^v a(u)\,du.
\end{aligned}
\end{equation}
Since $Y(0)=0$ and $\widetilde\Xi_K(0)=\Xi_{K,0}$, it follows that
\begin{equation}
\widetilde\Xi_K(v)=\cos A(v)\,\Xi_{K,0},
\qquad
Y(v)=\sin A(v)\,\Xi_{K,0}.
\end{equation}
Therefore the hidden-sector history closes as
\begin{equation}
\Xi_K(v)=U_{KK}(v,0)\,e^{\Theta(v,0)J}\,\cos A(v)\,\Xi_{K,0}.
\label{eq:Xi_solution_app}
\end{equation}
Using the source-fixed clock of Eq.~\eqref{eq:sourcedclock}, the two source-fixed quadratures are
\begin{equation}
\begin{aligned}
A(\hat v)
&=\frac{F(e_0)}{e_0}\int_{e(\hat v)}^{e_0}\frac{a(\varepsilon)}{F(\varepsilon)}\,d\varepsilon,\\
\Theta(\hat v)
&=\frac{F(e_0)}{e_0}\int_{e(\hat v)}^{e_0}\frac{\omega(\varepsilon)}{F(\varepsilon)}\,d\varepsilon.
\end{aligned}
\label{eq:ATheta_quads}
\end{equation}
Thus $\Xi_K$ is explicit by quadrature once the sourced background trajectory is fixed.
In the source-fixed gauge $U_{KK}=\idtwo$, the projected hidden functions reduce to finite initial data times the two source-fixed phases. This benchmark choice does not assert the absence of autonomous hidden-sector evolution in general; it fixes the sourced closure in which that evolution does not introduce additional independent data beyond the two source-fixed quadratures and the finite initial datum $\Xi_{K,0}$.
For example,
\begin{equation}
X_I(v)=\cos A(v)\bigl[\cos\Theta(v)X_{I0}+\sin\Theta(v)X_{J0}\bigr],
\end{equation}
with analogous formulas for $(X_J,Y_I,Y_J,Z_I,Z_J)$. Substituting these into Eqs.~\eqref{eq:xv}--\eqref{eq:zv} yields the explicit reduced-state formulas
\begin{align}
x(v)
&=a(v)\sin A(v)\bigl[\cos\Theta(v)X_{I0}+\sin\Theta(v)X_{J0}\bigr], \\
y(v)
&=a(v)\sin A(v)\bigl[\cos\Theta(v)Y_{I0}+\sin\Theta(v)Y_{J0}\bigr], \\
z(v)
&=1-2r(v)+a(v)\sin A(v)\bigl[\cos\Theta(v)Z_{I0}+\sin\Theta(v)Z_{J0}\bigr].
\end{align}
The hidden sector is therefore determined by the sourced clock, two source-fixed quadratures, and a finite initial hidden datum rather than by an arbitrary history functional.

\subsection{Reduced-channel admissibility}
\label{app:admissibility}

For a qubit density matrix, positivity is equivalent to the Bloch bound
\begin{equation}
|n(v)|\le 1,
\qquad
n(v):=(x(v),y(v),z(v)).
\label{eq:blochpositivityapp}
\end{equation}
AQM already requires positive normalized states on the full observable algebra. The theorem below does not introduce a new axiom;
it rewrites that positivity on the reduced ordinary channel in source-fixed variables.
Since $\Tr\rho_Q(v)=1$ by construction, positivity is the nontrivial condition.
Using the closure formulas above, namely Eqs.~\eqref{eq:Xi_solution_app}, \eqref{eq:ATheta_quads}, and the explicit Bloch components in this benchmark, Eqs.~\eqref{eq:xv}--\eqref{eq:zv}, write
\begin{equation}
\Lambda_{\rm hid}(v):=a(v)\sin A(v),
\qquad
m(v):=1-2r(v),
\end{equation}
and define the two hidden-support vectors
\begin{equation}
b_I:=(X_{I0},Y_{I0},Z_{I0}),
\qquad
b_J:=(X_{J0},Y_{J0},Z_{J0}).
\end{equation}
Then
\begin{equation}
q(v):=\cos\Theta(v)\,b_I+\sin\Theta(v)\,b_J,
\end{equation}
so that the Bloch vector can be written compactly as
\begin{equation}
n(v)=m(v)\,\hat z+\Lambda_{\rm hid}(v)q(v).
\label{eq:nvqapp}
\end{equation}
Therefore the necessary-and-sufficient admissibility condition is
\begin{equation}
\rho_Q(v)\ge0
\iff
|m(v)\,\hat z+\Lambda_{\rm hid}(v)q(v)|\le1
\quad
\forall v.
\label{eq:generaladmissibilityapp}
\end{equation}
Expanding the norm gives the equivalent scalar condition
\begin{equation}
\Lambda_{\rm hid}(v)^2\bigl(q_x(v)^2+q_y(v)^2\bigr)+\bigl(m(v)+\Lambda_{\rm hid}(v)q_z(v)\bigr)^2\le1
\quad \forall v.
\label{eq:expandedadmissibilityapp}
\end{equation}
This is the positivity theorem.
It is pointwise in the source-fixed clock and it does not require any globally constant hidden-sector amplitude.
The physical content of Eq.~\eqref{eq:expandedadmissibilityapp} is immediate:
\begin{equation}
\begin{aligned}
\delta r(v)&=-\frac{\Lambda_{\rm hid}(v)q_z(v)}{2},\\
|c(v)|&=\frac{|\Lambda_{\rm hid}(v)|}{2}
\sqrt{q_x(v)^2+q_y(v)^2}.
\end{aligned}
\label{eq:deltarcfromqapp}
\end{equation}
Thus the longitudinal hidden component controls repopulation, whereas the transverse hidden components control coherence. For the transport observables displayed in Fig.~\ref{fig:rates}, one evaluates the fixed-support photonic representative of Appendix~\ref{app:fixedsupport} and applies the criterion \eqref{eq:expandedadmissibilityapp}. On that representative \(\delta r(v)=0\), while \(|c(v)|\) is fixed by the accumulated-history formula \eqref{eq:photonic-c}; positivity is checked pointwise on the reduced channel.

\section{Photonic realization, polarization dictionary, and channel rates}
\label{app:photonic}

The neutral-source geometry remains an effective Einstein--Maxwell-plus-null-fluid background rather than a pure-photon spacetime. The photonic reduction is therefore narrower: one solves the gauge-invariant spin-$1$ problem on adiabatic RN slices, identifies the fixed transverse support that realizes the ordinary channel, and then translates the density-matrix representative of the restricted readout state into detector-level polarization and rate observables.

\subsection{Gauge-invariant spin-1 problem on adiabatic RN slices}
\label{app:spin1}

Once the geometry is specified in this way, the correct photonic calculation is the gauge-invariant spin-1 problem on instantaneous RN slices with slowly varying $M(u)$.
Define the background variation time
\begin{equation}
\tau_{\rm bg}(v):=\min\left\{\left|\frac{M}{\dot M}\right|,\left|\frac{T_H}{\dot T_H}\right|,\left|\frac{r_+}{\dot r_+}\right|\right\}.
\label{eq:taubg}
\end{equation}
The instantaneous greybody description used below assumes $\omega\tau_{\rm bg}(v)\gg1$ for the modes retained in the adiabatic readout. When this fails, the fixed-support spectral forecast must be replaced by a genuinely time-dependent scattering problem.
In the Moncrief formulation one arrives at the same physical master equation \cite{moncrief1974a,moncrief1975,zerilli1974,chandrasekhar1983},
\begin{equation}
\begin{aligned}
\frac{d^2\Psi_{\ell\omega}}{dr_*^2}+\left[\omega^2-V_\ell^{(1)}(r)\right]\Psi_{\ell\omega}&=0,\\
V_\ell^{(1)}(r)&=f(r)\frac{\ell(\ell+1)}{r^2},\qquad \ell\ge 1.
\end{aligned}
\label{eq:master}
\end{equation}
The tortoise coordinate is
\begin{equation}
r_*(r)=r+\frac{r_+^2}{r_+-r_-}\ln(r-r_+)-\frac{r_-^2}{r_+-r_-}\ln(r-r_-)+\mathrm{const}.
\label{eq:tortoise}
\end{equation}
The photonic problem is therefore radial only after harmonic decomposition.
Any spin-1 WKB or numerical greybody calculation must be applied mode by mode to Eq.~\eqref{eq:master};
it is not obtained by treating the sourced metric itself as a one-dimensional photon equation.
At low frequency the dominant spin-1 mode has the quartic greybody scaling
\begin{equation}
\Gamma^{(1)}_1(\omega;E)=a_1(E)(\omega r_+)^4+O(\omega^6),
\label{eq:quartic}
\end{equation}
so the spectral energy flux is
\begin{equation}
\frac{d^2E_\gamma}{du\,d\omega}=
\sum_{\ell=1}^{\infty}\frac{2(2\ell+1)}{2\pi}\,
\frac{\omega\,\Gamma^{(1)}_\ell(\omega;E)}{e^{\omega/T_H}-1}.
\label{eq:spectral}
\end{equation}
Keeping only the dominant quartic term gives
\begin{align}
L_\gamma(E)
&=\int_0^\infty d\omega\,\frac{3}{\pi}
\frac{\omega\,a_1(E)(\omega r_+)^4}{e^{\omega/T_H}-1}+\cdots \nonumber\\
&=\frac{3a_1(E)r_+^4T_H^6}{\pi}
\int_0^\infty dx\,\frac{x^5}{e^x-1}+\cdots \nonumber\\
&=\frac{8\pi^5}{21}a_1(E)\,r_+^4T_H^6+\cdots.
\label{eq:T6law}
\end{align}
The corresponding photon number flux removes one power of \(\omega\). Here \(\dot N_\gamma(E)\) denotes the number of photons emitted per unit time on the adiabatic slice with excitation energy \(E\):
\begin{equation}
\begin{aligned}
\dot N_\gamma(E)&=\sum_{\ell,m}\int_0^\infty\frac{d\omega}{2\pi}\,
\frac{\Gamma^{(1)}_\ell(\omega;E)}{e^{\omega/T_H(E)}-1}\\
&=\frac{36\zeta(5)}{\pi}a_1(E)r_+^4(E)T_H^5(E)+\cdots .
\end{aligned}
\label{eq:NgammadotApp}
\end{equation}
In Eq.~\eqref{eq:NgammadotApp}, \(\Gamma_\ell^{(1)}\) is the spin-$1$ greybody transmission coefficient and \(a_1(E)\) is the leading low-frequency electromagnetic coefficient; \(T_H\) and \(r_+\) are the RN quantities defined above. Thus the near-extremal pure-photon energy law is \(L_\gamma\propto r_+^4T_H^6\), while the number-flux envelope relevant for a two-window polarimetric forecast is \(\dot N_\gamma\propto r_+^4T_H^5\).
The frozen sourced background clock therefore remains an effective sourced clock rather than the luminosity law of the pure Maxwell sector alone \cite{hod2016,ngampitipan2013}.

\subsection{Transverse structure, fixed-support selection, and history kernel}
\label{app:fixedsupport}

The physical Maxwell sector on each RN slice has only the two transverse helicities $\lambda=\pm$.
In the photonic realization one identifies this two-dimensional support with the active support $\Cu=\mathrm{span}_{\mathbb R}\{1,u\}$, and the relay propagator acts on that support as
\begin{equation}
\begin{aligned}
R(v,s):=U_{11}(v,s)&=\cos\Theta(v,s)\,\id+\sin\Theta(v,s)\,J,\\
J^2&=-\id.
\end{aligned}
\label{eq:Rvs}
\end{equation}
Hence every physically admissible emitted-history two-point kernel in this setting has the form
\begin{equation}
g^{\ell\omega m}_{\lambda\lambda'}(s,s')=
A_{\ell\omega m}(s,s')\,\delta_{\lambda\lambda'}+
B_{\ell\omega m}(s,s')\,J_{\lambda\lambda'}.
\label{eq:ABform}
\end{equation}
There is no physical longitudinal sector. In particular, the natural photonic representative obeys
\begin{equation}
\delta r(v)=0.
\label{eq:photonic-drzero}
\end{equation}

The benchmark uses the fixed-support emitted-history kernel selected by the sourced channel. On the fixed transverse support this kernel is
\begin{equation}
\begin{aligned}
g_{\ell\omega m;\lambda\lambda'}(s,s')
&=\frac12\sqrt{n_{\ell\omega}(s)n_{\ell\omega}(s')}
\,e^{i[\delta_{\ell\omega}(s)-\delta_{\ell\omega}(s')]}\\
&\quad\times
\Bigl[\cos\Theta(s,s')\,\mathbf 1
+\sin\Theta(s,s')\,J\Bigr]_{\lambda\lambda'} .
\end{aligned}
\label{eq:gcol_app}
\end{equation}
At equal times it reproduces the unbiased transverse Hawking occupations, while away from equal times it keeps the relay phase that carries the emitted-history coherence. The kernel is positive on the selected support because it has the factorized form
\begin{equation}
\begin{aligned}
h_{\ell\omega m}(s)&:=\sqrt{n_{\ell\omega}(s)}e^{-i\delta_{\ell\omega}(s)}R(s,0),\\
g^{\ell\omega m}(s,s')&=\frac12 h_{\ell\omega m}(s)h_{\ell\omega m}(s')^\dagger .
\end{aligned}
\label{eq:gfactor}
\end{equation}
The accumulated off-diagonal channel coherence is the complex scalar \(c(v)\) appearing in the reduced readout matrix \(\rho_Q(v)\). In the fixed-support electromagnetic realization it is computed from the mode-resolved two-time history as
\begin{subequations}\label{eq:photonic-c-pair}
\begin{align}
c(v)
&= -\frac{1}{2\mathcal N(v)}
\sum_{\ell,m}\int_0^\infty d\omega
\int_0^v ds\int_0^v ds'\,
\sqrt{n_{\ell\omega}(s)n_{\ell\omega}(s')}
\nonumber\\
&\quad\times e^{i[\delta_{\ell\omega}(s)-\delta_{\ell\omega}(s')]}\sin\Theta(s,s'),
\label{eq:photonic-c}
\\
2\mathcal N(v)
&=
2\sum_{\ell,m}\int_0^\infty d\omega
\int_0^v ds\int_0^v ds'\,
\sqrt{n_{\ell\omega}(s)n_{\ell\omega}(s')}.
\label{eq:photonic-2N}
\end{align}
\end{subequations}
Thus \(c(v)\) is not a new tensorial object: it is the complex off-diagonal coefficient of the \(2\times2\) readout representative. Its modulus is the scalar quantity entering \(\Delta(v)\), \(S_A(v)\), and the coherence-transport rate. On the natural transverse representative, Eq.~\eqref{eq:photonic-drzero} holds, so the photonic specialization leaves the diagonal occupation at \(r(v)\) and modifies the density-matrix representative through the off-diagonal sector \(c(v)\). More general positive history kernels describe different emitted-history ensembles; the source-fixed benchmark used here is the fixed-support representative above.

\paragraph{Spectral-overlap forecast.}
For the quantitative photonic forecast one contracts the mode-resolved kernel to the positive spectral overlap
\begin{equation}
 W(s,s')=\sum_{\ell,m}\int_0^\infty d\omega\,
 \sqrt{n_{\ell\omega}(s)n_{\ell\omega}(s')} .
\label{eq:Wapp}
\end{equation}
In the low-frequency spin-$1$ approximation this becomes
\begin{equation}
\begin{aligned}
W(s,s')&\propto r_+^2(s)r_+^2(s')
\int_0^\infty d\omega\\
&\quad\times
\frac{\omega^4}
{\sqrt{(e^{\omega/T_H(s)}-1)(e^{\omega/T_H(s')}-1)}} .
\end{aligned}
\label{eq:WlowfreqApp}
\end{equation}
The mode-integrated separable envelope \(g_\gamma(s)=[\dot N_\gamma(s)/\dot N_\gamma(0)]^{1/2}\), with \(\dot N_\gamma\) given by Eq.~\eqref{eq:NgammadotApp}, is useful for scale estimates. The quantitative plots evaluate the nonseparable overlap \(W(s,s')\) of Eq.~\eqref{eq:WlowfreqApp}. The phase-sensitive cross-helicity observable obeys
\begin{equation}
C_{+-}(s,s')\propto W(s,s')\exp\{2i[\Theta(s)-\Theta(s')]\},
\label{eq:CpmApp}
\end{equation}
up to the mode-resolved scattering phase already displayed in Eq.~\eqref{eq:gcol_app}.  Thus the falsifiable fixed-support signature is not only an excess intensity covariance, but a polarimetric covariance phase-locked to the relay angle.

\subsection{Bloch--Stokes dictionary and fixed-support rates}
\label{app:blochstokes}

On the active support one has $Q_{\mathrm{act}}\cong\Cu\cong\R^2$, together with the fixed antisymmetric matrix $J$ of Eq.~\eqref{eq:fixedsupportmatrices}. The density-matrix representative
\begin{equation}
\rho_Q(v)=\frac12\bigl(\idtwo+x(v)\sigma_x+y(v)\sigma_y+z(v)\sigma_z\bigr)
\label{eq:IrhoQ}
\end{equation}
may therefore be read as a $2\times2$ coherency matrix for an effective two-component radiative mode. Introduce the effective Stokes components
\begin{equation}
\begin{aligned}
S_0(v)&:=1,\\
S_1(v)&:=x(v),\\
S_2(v)&:=y(v),\\
S_3(v)&:=z(v).
\end{aligned}
\label{eq:IStokes}
\end{equation}
Then Eq.~\eqref{eq:IrhoQ} becomes the standard coherency decomposition
\begin{equation}
\rho_Q(v)=\frac12
\begin{pmatrix}
S_0(v)+S_3(v) & S_1(v)-iS_2(v)\\[2pt]
S_1(v)+iS_2(v) & S_0(v)-S_3(v)
\end{pmatrix}.
\label{eq:Icoherency}
\end{equation}
The same state admits a natural ladder-operator form. Define
\begin{equation}
\begin{aligned}
\sigma_+&:=|1\rangle\langle 0|,\\
\sigma_-&:=|0\rangle\langle 1|,\\
n&:=|1\rangle\langle 1|=\sigma_+\sigma_- = \frac{\idtwo-\sigma_z}{2}.
\end{aligned}
\label{eq:Iladders}
\end{equation}
Then, using Eq.~\eqref{eq:rhoQ},
\begin{equation}
\begin{aligned}
\rho_Q(v)=&\;\bigl(1-r(v)-\delta r(v)\bigr)|0\rangle\langle0|\\
&+\bigl(r(v)+\delta r(v)\bigr)|1\rangle\langle1|\\
&+c(v)\sigma_-+c(v)^*\sigma_+.
\end{aligned}
\label{eq:IrhoLadder}
\end{equation}
Consequently,
\begin{equation}
\langle n\rangle_v=r(v)+\delta r(v),
\qquad
\langle\sigma_-\rangle_v=c(v),
\qquad
\langle\sigma_+\rangle_v=c(v)^*.
\label{eq:IladderExpect}
\end{equation}
Thus $\delta r(v)=\langle n\rangle_v-r(v)$ is the population shift relative to the leading flux variable, while $|c(v)|$ is the magnitude of the off-diagonal coherence sector.
Since
\begin{equation}
J=-i\sigma_y=\sigma_- - \sigma_+,
\label{eq:IJsigma}
\end{equation}
the bridge propagator is equally a pseudospin rotation,
\begin{equation}
U_{11}(v,s)=\cos\Theta(v,s)\,\idtwo+\sin\Theta(v,s)\,J=e^{-i\Theta(v,s)\sigma_y}.
\label{eq:IU11SU2}
\end{equation}
The effective degree of polarization is
\begin{equation}
\begin{aligned}
P(v)&:=\frac{\sqrt{S_1(v)^2+S_2(v)^2+S_3(v)^2}}{S_0(v)}\\
&=\sqrt{x(v)^2+y(v)^2+z(v)^2}.
\end{aligned}
\label{eq:IP}
\end{equation}
Using Eqs.~\eqref{eq:gap} and \eqref{eq:readout}, one finds immediately
\begin{equation}
P(v)=\Delta(v).
\label{eq:IPDelta}
\end{equation}
Hence the same entropy may be rewritten as the polarization-entropy identity
\begin{equation}
S_A(v)=H_2\!\left(\frac{1-P(v)}{2}\right).
\label{eq:ISAP}
\end{equation}
The same representative in this benchmark is thus simultaneously a qubit density matrix, a coherency matrix of an effective two-component radiative mode, and a pseudospin density matrix whose diagonal and off-diagonal sectors are resolved by Eq.~\eqref{eq:IladderExpect}.
This is the meaning of the vector-polarization reading invoked in Sec.~\ref{sec:state}.
Pure channel states in this setting correspond to the usual $SU(2)$-coherent limit of this effective two-component mode.
For background on coherency matrices, quantum Stokes variables, and polarization geometry, see Refs.~\cite{ParisRehacek2004,Alonso2023PolarizationGeometry,GilOssikovski2022}.

\paragraph{Rates from the sourced trajectory.}
\label{app:ratesclock}

The sourced background fixes the excitation history $\eexc(v)$ and the evaporative clock
\begin{equation}
\begin{aligned}
\Gamma_{\mathrm{evap}}(v)&=-\frac{\dd}{\dd v}\log \eexc(v),\\
\dot{\eexc}(v)&=-\Gamma_{\mathrm{evap}}(v)\,\eexc(v).
\end{aligned}
\label{eq:IRGamma}
\end{equation}
The energy-drain rate is therefore
\begin{equation}
\mathcal{R}_E(v):=-\dot{\eexc}(v)=\Gamma_{\mathrm{evap}}(v)\,\eexc(v).
\label{eq:IRE}
\end{equation}
Next, differentiate the barrier weight $a(E)=e^{-S_*(E)}$ along this background. Since
\begin{equation}
S_*'(E)=\frac12\,T_H(E)^{-1}=\frac{1}{2\,\omega(E)},
\label{eq:ISprimed}
\end{equation}
one obtains
\begin{align}
\frac{\dd}{\dd v}\log a\bigl(\eexc(v)\bigr)
&=-S_*'\bigl(\eexc(v)\bigr)\,\dot{\eexc}(v)
\nonumber\\
&=\frac{\Gamma_{\mathrm{evap}}(v)\,\eexc(v)}{2\,\omega\bigl(\eexc(v)\bigr)}.
\end{align}
Thus the entropic relay-opening rate is
\begin{equation}
\mathcal{R}_a(v):=\frac{\dd}{\dd v}\log a\bigl(\eexc(v)\bigr)
=\frac{\Gamma_{\mathrm{evap}}(v)\,\eexc(v)}{2\,\omega\bigl(\eexc(v)\bigr)}.
\label{eq:IRa}
\end{equation}
Because the matrix factor in Eq.~\eqref{eq:kernelmain} is a rotation, its operator norm is unity.
Therefore
\begin{equation}
\|K_{QK}(v,s)\|=a\bigl(\eexc(v)\bigr)a\bigl(\eexc(s)\bigr),
\label{eq:IKnorm}
\end{equation}
and the cumulative memory intensity becomes
\begin{equation}
\mathcal{M}(v):=\int_0^v \|K_{QK}(v,s)\|\,\dd s
=a\bigl(\eexc(v)\bigr)\int_0^v a\bigl(\eexc(s)\bigr)\,\dd s.
\label{eq:IM}
\end{equation}
Differentiating Eq.~\eqref{eq:IM} gives the closed evolution law
\begin{equation}
\dot{\mathcal{M}}(v)=\mathcal{R}_a(v)\,\mathcal{M}(v)+a\bigl(\eexc(v)\bigr)^2.
\label{eq:IMdot}
\end{equation}
A further local relay scale is obtained by combining the bridge weight and bridge frequency,
\begin{equation}
\mathcal{R}_{\mathrm{ch}}(v):=a\bigl(\eexc(v)\bigr)^2\,\omega\bigl(\eexc(v)\bigr).
\label{eq:IRch}
\end{equation}
Eqs.~\eqref{eq:IRE}--\eqref{eq:IRch} are intrinsic observables of the channel in this setting. They are not species-resolved Hawking emissivities.

On the natural transverse photonic representative, Eq.~\eqref{eq:photonic-drzero} and Eq.~\eqref{eq:ladderexpectmain} give \(\langle n\rangle_v=r(v)\). Since \(r(v)=1-\eexc(v)/\eexc(0)\), the diagonal occupation bookkeeping follows the sourced clock directly, while the fixed-support coherence remains off-diagonal in the channel basis.

\paragraph{Population, coherence, spectral-gap, and entropy rates.}
\label{app:ratesentropy}

From Eq.~\eqref{eq:rdef},
\begin{equation}
\dot r(v)=\Gamma_{\mathrm{evap}}(v)\,[1-r(v)].
\label{eq:IrDot}
\end{equation}
From Eq.~\eqref{eq:readout},
\begin{equation}
\dot{\delta r}(v)=-\dot r(v)-\frac12\dot z(v).
\label{eq:Ideltardot}
\end{equation}
Likewise,
\begin{equation}
\dot c(v)=\frac12\bigl(\dot x(v)-i\dot y(v)\bigr),
\label{eq:Icdot}
\end{equation}
and, whenever $|c(v)|\neq0$,
\begin{equation}
\frac{\dd}{\dd v}|c(v)|
=\frac{x(v)\dot x(v)+y(v)\dot y(v)}{4|c(v)|}.
\label{eq:Icohdot}
\end{equation}
The effective radiative occupation rate is therefore
\begin{equation}
\mathcal{R}_Q(v):=\frac{\dd}{\dd v}\bigl[r(v)+\delta r(v)\bigr]
=\dot r(v)+\dot{\delta r}(v),
\label{eq:IRQ}
\end{equation}
while the coherence-transport rate is
\begin{equation}
\mathcal{R}_{\mathrm{coh}}(v):=\frac{\dd}{\dd v}|c(v)|.
\label{eq:IRcoh}
\end{equation}
Now differentiate the spectral gap. Since
\begin{equation}
\Delta(v)=\sqrt{x(v)^2+y(v)^2+z(v)^2},
\label{eq:IDelta}
\end{equation}
one has, for $\Delta(v)\neq0$,
\begin{equation}
\dot\Delta(v)=\frac{x(v)\dot x(v)+y(v)\dot y(v)+z(v)\dot z(v)}{\Delta(v)}.
\label{eq:IDeltadot}
\end{equation}
Finally, differentiating Eq.~\eqref{eq:SAchannel} gives the reduced-channel entropy production rate
\begin{equation}
\mathcal{R}_S(v):=\dot S_A(v)
=-\frac12\log_2\!\left(\frac{1+\Delta(v)}{1-\Delta(v)}\right)\dot\Delta(v).
\label{eq:IRS}
\end{equation}
Equations~\eqref{eq:IRQ}, \eqref{eq:IRcoh}, \eqref{eq:IDeltadot}, and \eqref{eq:IRS} show that the construction predicts a genuine kinematics of population transfer, coherence transport, polarization-gap flow, and entropy production.

\section{Euclidean branch measures and endpoint canonization}
\label{app:branchappendix}

This appendix records the branch-side checks used in Sec.~\ref{sec:superstatistics}: the outgoing-history identification, the pushforward construction of the positive branch measure, the replica/Mellin tower carried by that measure, the regular-opening Tsallis/Lomax endpoint, the near-extremal shifted-L{\'e}vy endpoint, and the no-global-positive-transform argument that justifies the branch-admissible hard envelope.

\subsection{Outgoing-history identification on the sourced background}
\label{app:outgoinghistory}

Define the outgoing algebra up to time $v$ by
\begin{equation}
\mathfrak A_{\mathrm{out}}([0,v])
:=\Jor\Big\langle \mathcal O_{\mathrm{asymp}}(s):0\le s\le v\Big\rangle,
\end{equation}
where $\mathcal O_{\mathrm{asymp}}(s)$ denotes any asymptotic observable that has reached the selected readout sector by time $s$.
The proof is benchmark-sectorial.
Fix the readout sector on the neutral-source RN geometry to be the gauge-invariant transverse spin-1 sector carried by the two Maxwell helicities of the ordinary block $\Qu$ \cite{moncrief1974a,moncrief1975,zerilli1974,chandrasekhar1983}.
Readout completeness holds in this benchmark because every gauge-invariant asymptotic observable in that sector is generated by the ordinary channel described in Secs.~\ref{sec:Peirce} and \ref{sec:channel}.
No independent leakage holds on the same geometry because $J_I^{H(5)N}=0$ and the fixed-support photonic realization closes on the same two-helicity support, as established in Sec.~\ref{sec:state} and Appendix~\ref{app:fixedsupport}.

Every element of $\mathfrak A_{\mathrm{out}}([0,v])$ therefore belongs to the algebra generated by the projected ordinary readout and lies in $R_v^{\min}$.
Conversely, every generator appearing in Eq.~\eqref{eq:Rmin} is by construction an outgoing asymptotic observable of that same selected sector, so $R_v^{\min}\subseteq \mathfrak A_{\mathrm{out}}([0,v])$.
Hence
\begin{equation}
\mathfrak A_{\mathrm{out}}([0,v])=R_v^{\min}.
\end{equation}
Restricting the same background state $\varpi_v$ to the two equal algebras gives Eq.~\eqref{eq:SAout}. This proves Proposition~\ref{prop:AoutEqualsRmin}.

\subsection{Positive branch measure as a pushforward}
\label{app:branch_canonization}

Let $H_b$ be a branch-history domain and let $h=(E,E')\in H_b$ denote an ordered Euclidean transfer segment.  The source-fixed kernel of Sec.~\ref{sec:superstatistics} assigns to such histories the positive barrier-clock measure
\begin{equation}
\dd\nu_b(h)=W(E)W(E')\,\dd\nu_{\rm clock}(E,E'),
\label{eq:dnub_app}
\end{equation}
with $W(E)=e^{-S_*(E)}$.  A branch scale map
\begin{equation}
\lambda_b:H_b\to\mathbb R_+
\end{equation}
assigns the effective intensive transfer scale conjugate to the branch coordinate $y_b$.  The associated branch measure is the pushforward
\begin{equation}
\dd\mu_b(\lambda)=(\lambda_b)_*\dd\nu_b.
\label{eq:pushforward_app}
\end{equation}
Equivalently, when written with a Dirac measure,
\begin{equation}
\dd\mu_b(\lambda)=\int_{H_b}\dd\nu_b(h)\,
\delta\!\left(\lambda-\lambda_b(h)\right).
\label{eq:pushforward_delta_app}
\end{equation}
The branch transfer function is therefore
\begin{equation}
G_b(y_b)=\int_{H_b}e^{-\lambda_b(h)y_b}\,\dd\nu_b(h)
=\int_0^\infty e^{-\lambda y_b}\,\dd\mu_b(\lambda).
\label{eq:Gb_push_app}
\end{equation}
This is the positive canonical sum over branch histories, in the branchwise superstatistical sense of a positive Laplace superposition of local exponential weights \cite{BeckCohen2003,Beck2007,Bernstein1929,Widder1941}.  If $\dd\mu_b(\lambda)=e^{\Sigma_b(\lambda)}\dd\lambda$ locally, then the saddle approximation gives
\begin{equation}
A_b(y_b):=-\log G_b(y_b)\simeq
\min_\lambda\{\lambda y_b-\Sigma_b(\lambda)\}.
\label{eq:Legendre_branch_app}
\end{equation}
Thus the Laplace form converts the scale-resolved entropy of Euclidean histories into the branch free action \cite{Bernstein1929,Widder1941}.

\subsection{Replica/Mellin tower of a branch measure}
\label{app:replica_mellin}

The branch Laplace representation also determines a canonical replica
tower. Here \(n\) is first a positive integer labeling a Mellin/replica
moment of the branch measure. The replicated objects are formal
readings of the same temporal branch ensemble, not independent readings
with separately sampled branch scales. For fixed branch coordinate
\(y_b\), set
\begin{equation}
z=e^{-\lambda y_b}\in(0,1].
\end{equation}
The pushforward of \(\dd\mu_b(\lambda)\) under this map defines a
positive measure \(\dd\nu_{b,y_b}(z)\). Its \(n\)-th Mellin moment is
\begin{equation}
Z^{\rm rel}_{n,b}(y_b):=\int_0^1 z^n\,\dd\nu_{b,y_b}(z).
\end{equation}
Equivalently,
\begin{equation}
Z^{\rm rel}_{n,b}(y_b)
=\int_0^\infty e^{-n\lambda y_b}\,\dd\mu_b(\lambda)
=G_b(ny_b).
\label{eq:MellinTowerApp}
\end{equation}
Thus the branch measure fixes a distinguished shared-scale replica
deformation: the replicated readings are conditioned on the same
effective branch scale \(\lambda\). This is the channel-side signature
of temporal memory.

If the \(n\) replicas were statistically independent, each would carry
its own branch scale, and the corresponding factorized weight would be
\begin{equation}
Z^{\rm fact}_{n,b}(y_b)=G_b(y_b)^n.
\end{equation}
The ratio
\begin{equation}
\mathcal R_{n,b}(y_b)
=\frac{Z^{\rm rel}_{n,b}(y_b)}{Z^{\rm fact}_{n,b}(y_b)}
=\frac{G_b(ny_b)}{G_b(y_b)^n}
\label{eq:ReplicaRatioApp}
\end{equation}
therefore measures the non-factorizing part of the branch replica
tower. Writing \(A_b=-\log G_b\), the associated connected replica
action is
\begin{equation}
I^{\rm conn}_{n,b}(y_b):=-\log \mathcal R_{n,b}(y_b)
=A_b(ny_b)-nA_b(y_b).
\label{eq:ConnectedReplicaActionApp}
\end{equation}
In the temporal-ensemble reading, \(G_b(ny_b)\) describes replicated
measurements tied to the same slow branch scale, whereas
\(G_b(y_b)^n\) describes independent draws of that scale. The connected
action is therefore the channel-side measure of temporal memory in the
branch ensemble.

For the regular-opening endpoint,
\begin{equation}
G_H(y)=\int_0^\infty e^{-\lambda y}
\frac{\lambda_c}{(\lambda+\lambda_c)^2}\,\dd\lambda,
\end{equation}
so the replica/Mellin tower is
\begin{equation}
Z^{\rm rel}_{n,H}(y)=G_H(ny).
\end{equation}
The same branch gives an explicit connected activation action. Setting
\(a=\lambda_c y\), its transfer function can be written as
\begin{equation}
\begin{aligned}
G_H(y)&=\Phi(a),\qquad a=\lambda_c y,\\
\Phi(a)&:=1-ae^aE_1(a),\\
E_1(a)&=\int_a^\infty \frac{e^{-t}}{t}\,\dd t .
\end{aligned}
\label{eq:LomaxTransferClosedD3}
\end{equation}
Therefore
\begin{equation}
I^{\rm conn}_{n,H}(y)
=
-\log\Phi(n\lambda_c y)
+n\log\Phi(\lambda_c y).
\label{eq:LomaxConnectedD3}
\end{equation}
For \(\lambda_c y\ll1\),
\begin{equation}
I^{\rm conn}_{n,H}(y)
=
-n(\lambda_c y)\log n
+
O\!\left((\lambda_c y)^2|\log(\lambda_c y)|^2\right).
\label{eq:LomaxConnectedSmallD3}
\end{equation}
The regular-opening branch therefore carries a connected
replica/Mellin response with a logarithmic activation signature. This
is the \(q=2\) Lomax tail appearing as shared-scale temporal memory in
the opening regime.

For the near-extremal residence endpoint,
\begin{equation}
G_I(\epsilon)=\exp[-\mu_e\epsilon-\nu_e\sqrt{\epsilon}],
\end{equation}
and hence
\begin{equation}
Z^{\rm rel}_{n,I}(\epsilon)
=G_I(n\epsilon)
=\exp[-\mu_e n\epsilon-\nu_e\sqrt{n\epsilon}].
\end{equation}
The connected action in the residence branch is therefore
\begin{equation}
I^{\rm conn}_{n,I}(\epsilon)
=A_I(n\epsilon)-nA_I(\epsilon)
=\nu_e(\sqrt n-n)\sqrt{\epsilon}.
\label{eq:LevyConnectedD3}
\end{equation}
The linear term cancels in the connected ratio; the non-analytic
near-extremal residence term carries the non-factorizing replica
response. In the temporal-ensemble reading, the replicated measurements
are tied to the same long near-extremal waiting scale.

The object controlled here is the effective connected replica action of
the selected channel. The branch measure gives the channel-side
realization of the connected replica sector: replicated readings are
tied by a shared slow temporal scale, and the connected action
reproduces the same non-factorizing saddle bookkeeping on the selected
readout.
\subsection{Regular-opening endpoint and the Tsallis/Lomax class}
\label{app:opening_lomax}

At the opening endpoint set
\begin{equation}
\begin{gathered}
r=1-\frac{E}{E_0},\qquad E=E_0(1-r),\\
E'=E_0(1-u),\qquad 0\le u\le r.
\end{gathered}
\end{equation}
Regularity of the sourced clock and barrier gives
\begin{equation}
\begin{aligned}
S_*(E)&=S_*(E_0)+O(r),\\
\Gamma(E)&=\Gamma_0+O(r),\\
\tau(E,E')&=O(r).
\end{aligned}
\end{equation}
Substitution in Eq.~\eqref{eq:KEVII} yields
\begin{equation}
\begin{aligned}
K_E(r,u)&=C_0[1+O(r)],\\
C_0&=\frac{\exp[-2S_*(E_0)]}{\Gamma_0}.
\end{aligned}
\end{equation}
For a regular branch preparation $X_H(u)=X_0+O(u)$,
\begin{equation}
Z_H(r)\propto\int_0^rK_E(r,u)X_H(u)\,\dd u
=C_0X_0r+O(r^2).
\label{eq:ZH_app}
\end{equation}
This linear opening gives the endpoint branch action
\begin{equation}
A_H(r)=-r\log_2 r+O(r).
\label{eq:AH_app}
\end{equation}
The corresponding canonical scale density must reproduce
\begin{equation}
G_H(r)=e^{-A_H(r)}=1+\frac{r\ln r}{\ln2}+O(r).
\label{eq:GHopening_app}
\end{equation}
The Lomax/Pareto-II representative
\begin{equation}
w_H^{\rm can}(\lambda)=\frac{\lambda_c}{(\lambda+\lambda_c)^2},
\qquad \lambda\ge0,
\label{eq:lomax_app}
\end{equation}
is normalized and has Laplace transform
\begin{equation}
\int_0^\infty e^{-\lambda r}
\frac{\lambda_c}{(\lambda+\lambda_c)^2}\,\dd\lambda
=1+\lambda_c r\ln r+O(r).
\end{equation}
Comparison with Eq.~\eqref{eq:GHopening_app} fixes
\begin{equation}
\lambda_c=\frac{1}{\ln2}.
\end{equation}
Therefore the regular-opening endpoint canonizes to the Lomax/Pareto-II class, equivalently the one-sided Tsallis/Lomax branch \cite{Tsallis1988,Tsallis2009,Shalizi2007,Lomax1954}.

\subsection{Near-extremal RN residence and the shifted-L{\'e}vy class}
\label{app:late_levy}

From Eq.~\eqref{eq:Sstar}, the source-fixed action has the near-extremal expansion
\begin{equation}
\begin{aligned}
S_*(E)&=\mu E+\nu\sqrt E+c_{3/2}E^{3/2}+O(E^2),\\
c_{3/2}&=\frac{5\sqrt2\,\pi}{4}\sqrt{Q_0}.
\end{aligned}
\label{eq:Sstar_subleading}
\end{equation}
The coefficient follows by expanding $r_+(E)=Q_0+E+\sqrt{E(E+2Q_0)}$ to order $E^{3/2}$ and substituting into $S_*(E)=\tfrac12[\SBH(E)-\SBH(0)]$.
The branch-defining late term is the square root.  Setting $\epsilon=1-r=E/E_0$, the representative late weight is
\begin{equation}
\begin{aligned}
G_I(r)&=\exp[-\mu_e(1-r)-\nu_e\sqrt{1-r}],\\
\mu_e&=\mu E_0,
\qquad \nu_e=\nu\sqrt{E_0}.
\end{aligned}
\label{eq:GI_app}
\end{equation}
The identity
\begin{equation}
\begin{aligned}
e^{-\nu_e\sqrt\epsilon}
&=\int_0^\infty
\frac{\nu_e}{2\sqrt\pi}\alpha^{-3/2}\\
&\quad\times
\exp\!\left(-\frac{\nu_e^2}{4\alpha}\right)
e^{-\alpha\epsilon}\,\dd\alpha
\end{aligned}
\label{eq:Levyidentity_app}
\end{equation}
gives, after the shift $\beta=\mu_e+\alpha$,
\begin{equation}
G_I(r)=\int_{\mu_e}^{\infty}e^{-\beta(1-r)}
\frac{\nu_e}{2\sqrt\pi}(\beta-\mu_e)^{-3/2}
\exp\!\left[-\frac{\nu_e^2}{4(\beta-\mu_e)}\right]\dd\beta.
\label{eq:shifted_levy_app}
\end{equation}
This is the shifted one-sided L{\'e}vy residence class.  The $O(E^{3/2})$ terms in Eq.~\eqref{eq:Sstar_subleading} renormalize the smooth late profile while preserving the square-root branch class.

\subsection{No global positive transform and branch-envelope selection}
\label{app:noglobal}

\begin{proposition}[No single global positive transform for a genuine branch switch]
Let $A_H(x)$ and $A_I(x)$ be two real-analytic branch actions on a neighborhood of a dominance-exchange point $x_\times$, and suppose they exchange dominance there:
\begin{equation}
A_H(x_\times)=A_I(x_\times),
\end{equation}
with $A_H<A_I$ on one side and $A_I<A_H$ on the other side.  Assume that $A_H$ and $A_I$ do not define the same analytic germ at $x_\times$.  Then the branch envelope
\begin{equation}
A_{\rm br}(x)=\min\{A_H(x),A_I(x)\}
\end{equation}
cannot be represented, on any neighborhood of $x_\times$, as
\begin{equation}
e^{-A_{\rm br}(x)}=
\int_0^\infty e^{-\lambda y(x)}\,\dd\mu(\lambda),
\qquad \dd\mu\ge0,
\end{equation}
with $y(x)$ real analytic and monotone and with the transform finite and nonzero.
\end{proposition}

\begin{proof}
If the assumed global Laplace representation held, the right-hand side would be a positive Laplace transform of the variable $y$.  By the Bernstein--Widder theorem it is completely monotone in $y$ on the interior of its domain \cite{Bernstein1929,Widder1941}.  Where the transform is finite, and where $y(x)$ is real analytic and monotone, the composition
\begin{equation}
F(x)=\int_0^\infty e^{-\lambda y(x)}\,\dd\mu(\lambda)
\end{equation}
is real analytic in $x$ on that neighborhood.  Since $F(x)>0$, the action $-\log F(x)$ is also real analytic.  But $-\log F(x)=A_{\rm br}(x)=\min\{A_H(x),A_I(x)\}$.  A lower envelope of two distinct analytic branches that exchange dominance is analytic at the exchange point only in the degenerate case in which the two branches agree to all orders there, i.e. define the same analytic germ.  By hypothesis this is not the case.  Hence $A_{\rm br}(x)$ is not real analytic at the genuine branch exchange, contradicting the analyticity implied by a single positive global Laplace representation. Therefore no such global positive transform exists.
\end{proof}

\begin{remark}[Why this justifies the hard envelope]
Each branch action $A_b=-\log G_b$ is generated by its own positive pushforward measure on its branch-history domain.  The obstruction is to one global positive measure spanning the two regimes as a single smooth transform.  The source-fixed Page reconstruction is therefore the branch-admissible envelope $A_{\rm br}=\min(A_H,A_I)$. A soft sum $-\log(e^{-A_H}+e^{-A_I})$ would be a different physical model: it would add a global two-branch ensemble absent from the benchmark construction.
\end{remark}

\end{document}